\documentclass[12pt]{article}

\usepackage[OT1]{fontenc}
\usepackage{amsthm,amsmath,amssymb,amsfonts,bm,enumerate,xcolor,graphicx,natbib,color,url}
\usepackage[colorlinks,citecolor=blue,urlcolor=blue]{hyperref}
\usepackage{booktabs}
\usepackage{float} 
\usepackage{lscape}
\usepackage{mathrsfs}
\usepackage{mdframed}
\usepackage{verbatim}
\allowdisplaybreaks
\def\pr{\text{Pr}}%
\def\dt{\text{d}}
\numberwithin{equation}{section}
\theoremstyle{plain}
\newtheorem{proposition}{{\bf{Proposition}}}

\newtheorem{lemma}{{\bf{Lemma}}}

\newcommand{\blind}{1}

\begin{document}
\thispagestyle{empty}
\baselineskip=28pt
\vskip 5mm
\begin{center} {\Large{Local likelihood estimation of complex tail dependence structures, applied to U.S. precipitation extremes}}
\end{center}

\baselineskip=12pt

\vskip 5mm

\if1\blind
{
\begin{center}
\large
Daniela Castro-Camilo$^1$ and Rapha\"el Huser$^1$
\end{center}

\footnotetext[1]{
\baselineskip=10pt Computer, Electrical and Mathematical Sciences and Engineering (CEMSE) Division, King Abdullah University of Science and Technology (KAUST), Thuwal 23955-6900, Saudi Arabia. E-mails: daniela.castro@kaust.edu.sa, raphael.huser@kaust.edu.sa}
} \fi

\baselineskip=17pt
\vskip 4mm
\centerline{\today}
\vskip 6mm

\begin{center}
{\large{\bf Abstract}}
\end{center}
To {disentangle} the complex non-stationary dependence structure of precipitation extremes over the entire contiguous U.S., we propose a flexible local approach based on factor copula models. Our sub-asymptotic spatial modeling framework yields non-trivial tail dependence structures, with a weakening dependence strength as events become more extreme, a feature commonly observed with precipitation data but not accounted for in classical asymptotic extreme-value models. To estimate the local extremal behavior, we fit the proposed model in small regional neighborhoods to high threshold exceedances, under the assumption of local stationarity, {which allows us to gain in flexibility}. Adopting a local censored likelihood approach, inference is made on a fine spatial grid, and local estimation is performed by taking advantage of distributed computing resources and the embarrassingly parallel nature of this estimation procedure. The local model is efficiently fitted at all grid points, and uncertainty is measured using a block bootstrap procedure. An extensive simulation study shows that our approach can adequately capture complex, non-stationary dependencies, while our study of U.S. winter precipitation data reveals interesting differences in local tail structures over space, which has important implications on regional risk assessment of extreme precipitation events.\baselineskip=16pt

\par\vfill\noindent
{\bf Keywords:} factor copula model, local likelihood, non-stationarity, spatial extremes, threshold exceedance.\\

\newpage
\baselineskip=26pt
\section{Introduction}
\label{Introduction}
Water-related extremes such as floods and droughts can heavily impact human life, affecting our society, economic stability, and environmental sustainability. In recent years, we have witnessed an acceleration of the water cycle in some areas of the globe, including an increase in the frequency and intensity of heavy precipitation, sadly illustrated with a number of unprecedented hurricane events that recently hit the Caribbean islands and the Southeastern United States (U.S.), {including hurricane Harvey (August--September, 2017), and hurricane Florence (August--September, 2018).} 
Climate change has motivated the development of stochastic models for risk assessment and uncertainty quantification of extreme weather events. In the univariate context, the generalized extreme-value distribution with time-varying parameters has been used to study the effect of climate change on global precipitation annual maxima \citep{Westra.etal:2013}. More recently, \citet{Fischer.Knutti:2016} have adopted an empirical approach to study heavy rainfall intensification, validating the theory predicted by early generations of general circulation models. Similarly, \citet{Hoerling.etal:2016} analyzed trends in U.S. heavy precipitation data. Beyond the univariate context, several studies have focused on the spatial or spatio-temporal modeling of precipitation extremes within and across different catchments \citep{Cooley.etal:2007,Thibaud.etal:2013,huser2014space, opitz2018inla}, leading to extreme river discharges \citep{Asadi.etal:2015} or on the risk of heavy snowfall in mountainous areas \citep{Blanchet.Davison:2011}. In these small regions, however, the spatial dependence structure is often assumed to be stationary. 
By contrast, in this paper we focus on modeling the complex, non-stationary, spatial dependence structure of U.S. winter precipitation extremes.

Assessing the behavior of extreme events over space entails many challenges. First, data are measured at a finite number of locations, but extrapolation is typically required at any other location of the study region. In our case, the study region is the contiguous U.S. and the monitoring stations at which data are collected are displayed in Figure~\ref{fig:usmap.pdf}. 
\begin{figure}[t!]
\centering
\includegraphics[width=0.8\linewidth]{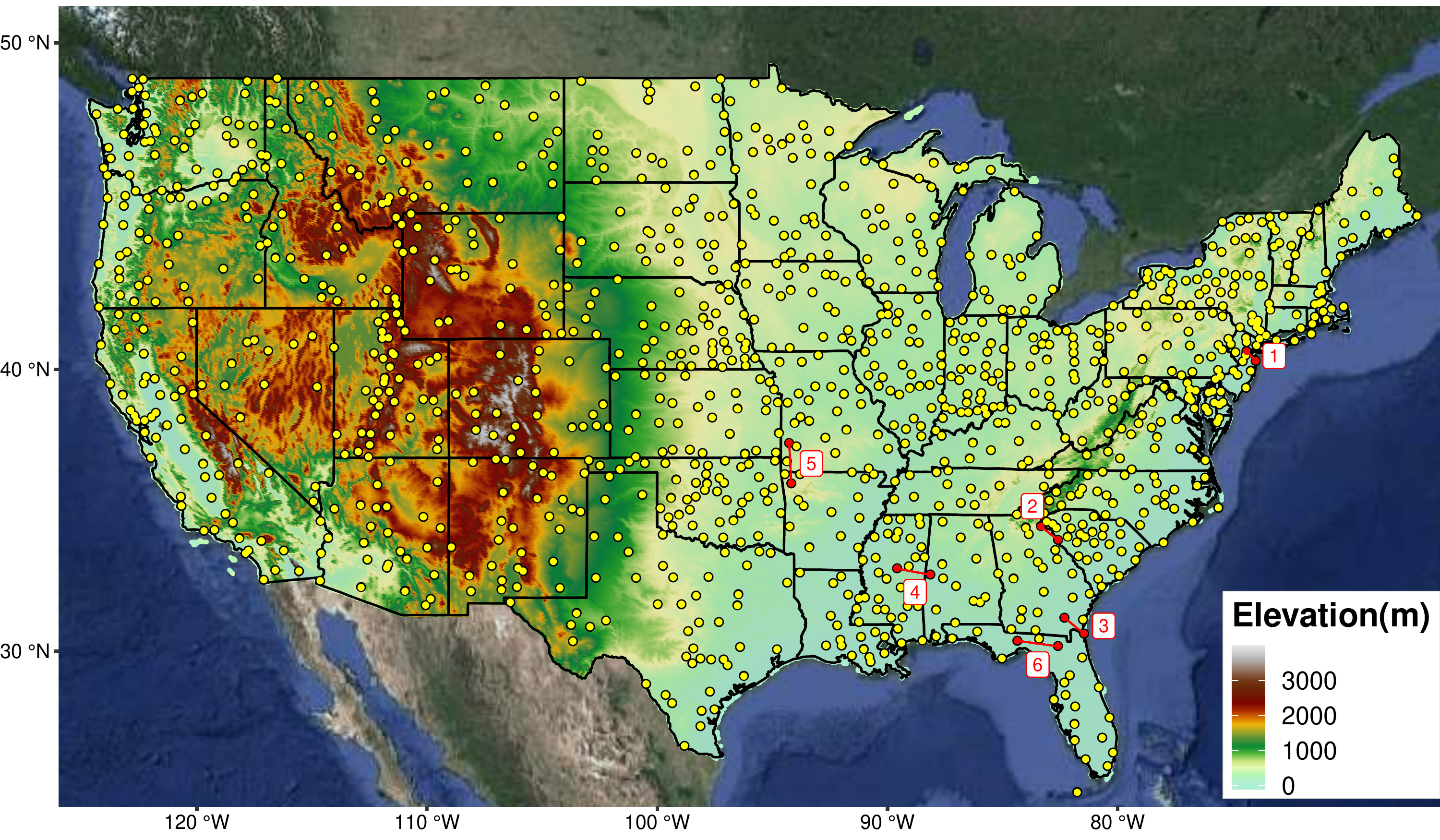}
\caption{\footnotesize Topographic map of the contiguous U.S. with the 1218 weather stations (dots). Red dots connected by a red segment represent the pairs of stations chosen in Figure~\ref{fig:chiu}, sorted by their distances (labels).}\label{fig:usmap.pdf}
\end{figure}
Second, by contrast with classical geostatistics where the estimation target is usually the (conditional) mean, statistics of extremes focuses on high (or low) quantiles, and extrapolation is not only required over space, but also into the joint upper (or lower) tail of the distribution. Third, the scarcity of extremes results in uncertainty inflation, especially when very high quantiles have to be estimated, and thus it is important to develop efficient inference methods that can potentially be applied in high dimensions. Finally, a recurrent problem with large heterogeneous regions, such as the contiguous U.S., is the modeling of spatial non-stationarity. This may concern both marginal distributions and the dependence structure. As pointed out by \citet{Huser2016}, misspecification of the joint distribution of extremes may result in poor estimation of spatial risk measures. There is no universal consensus on how to handle non-stationarity; however, it is often useful and realistic to assume some weak form of local stationarity. In this work, we fit {a set of locally stationary spatial models} to U.S. winter precipitation extremes, which contrasts with the quite rigid fully parametric model fitted by \citet{Huser2016}, and the non-parametric approach adopted by \citet{castro2016spectral}, whereby a family of bivariate extremal dependence structures at different sites are linked through the action of covariates, while neglecting spatial dependence. 

The literature on spatial extremes is mostly divided into two mainstream approaches: the first approach defines extreme events as block (e.g., annual) maxima, which are typically modeled using max-stable processes \citep{westra2011detection,Davison.Gholamrezaee:2012,Huser2016, davison2019spatial,Vettori.etal:2019}; the latter can be regarded as the functional generalization of multivariate extreme-value distributions, and they arise as limits of properly renormalized maxima of random fields. The second approach defines extreme events as high threshold exceedances, which are usually modeled using generalized Pareto processes, i.e., the threshold counterpart of max-stable processes \citep{Ferreira.deHaan:2014,ThibaudOpitz15,deFondevilleDavison16}. Both max-stable and generalized Pareto processes are asymptotic models, in the sense that they have a limiting characterization for block maxima and threshold exceedances, respectively. However, their tail dependence structure is fairly rigid: max-stable copulas are invariant to the operation of taking componentwise maxima, while generalized Pareto copulas are invariant to thresholding at higher levels. This lack of tail flexibility has recently motivated the development of sub-asymptotic spatial models for threshold exceedances \citep{Wadsworth.Tawn:2012b,Opitz:2016,Huser.etal:2016,Huser.Wadsworth:2017}, which, unlike asymptotic models, can capture weakening extremal dependence strength on the way to their limiting generalized Pareto process. {See also~\cite{huser2018penultimate} and~\cite{bopp2018hierarchical} for the case of maxima}. To illustrate this, Figure~\ref{fig:chiu} shows, for a range of quantiles {$u\in[0.5,0.995]$}, the conditional probability
\begin{align}\label{eq:chiu}
\chi_h(u)=\pr\{Y_1>F_1^{-1}(u)\mid Y_2>F_2^{-1}(u)\},\quad Y_j=Y(\mathbf{s}_j)\sim F_j,\quad j = 1,2,
\end{align}

\begin{figure}[!t]
\centering
\includegraphics[width=0.31\linewidth]{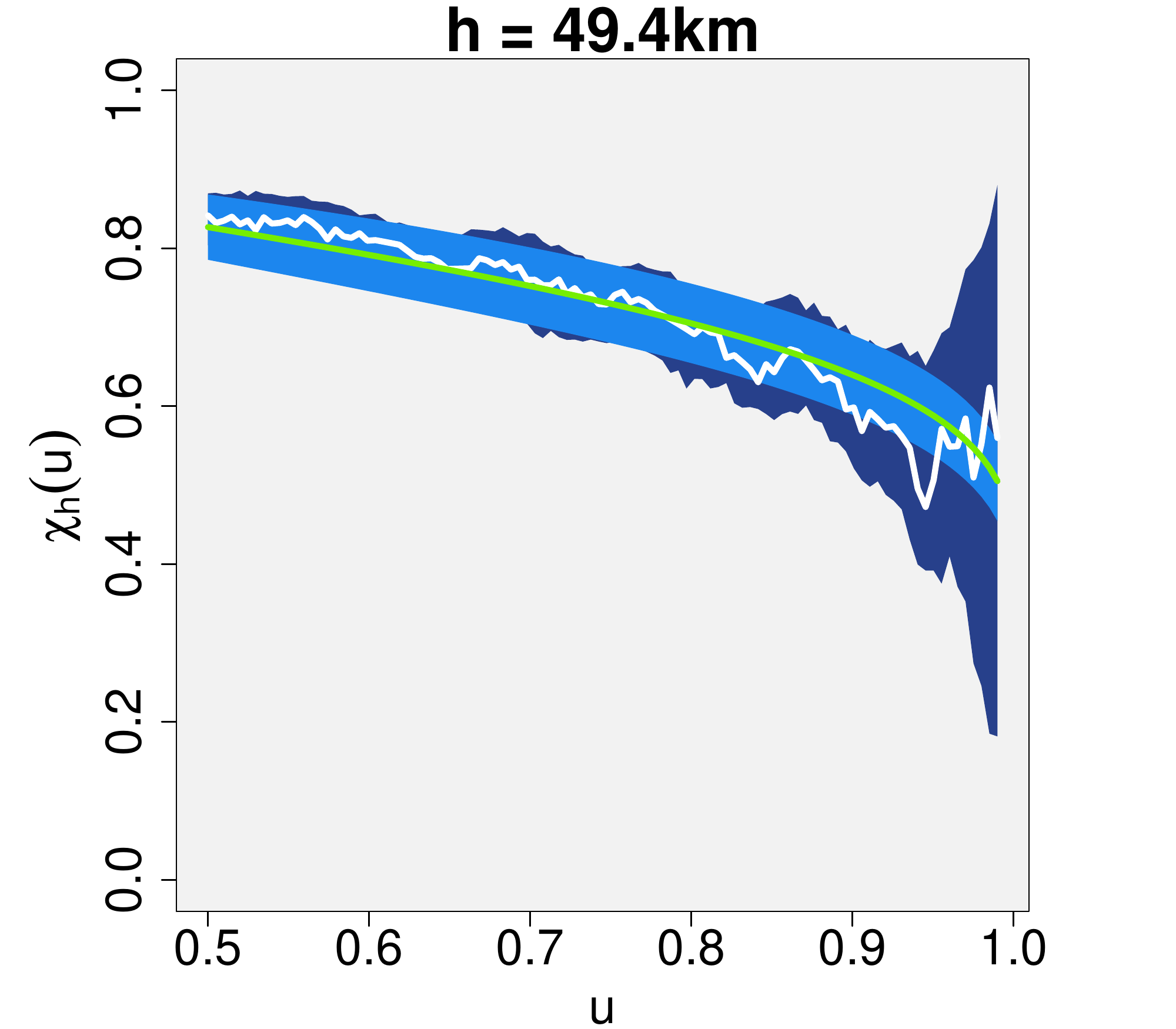}
\includegraphics[width=0.31\linewidth]{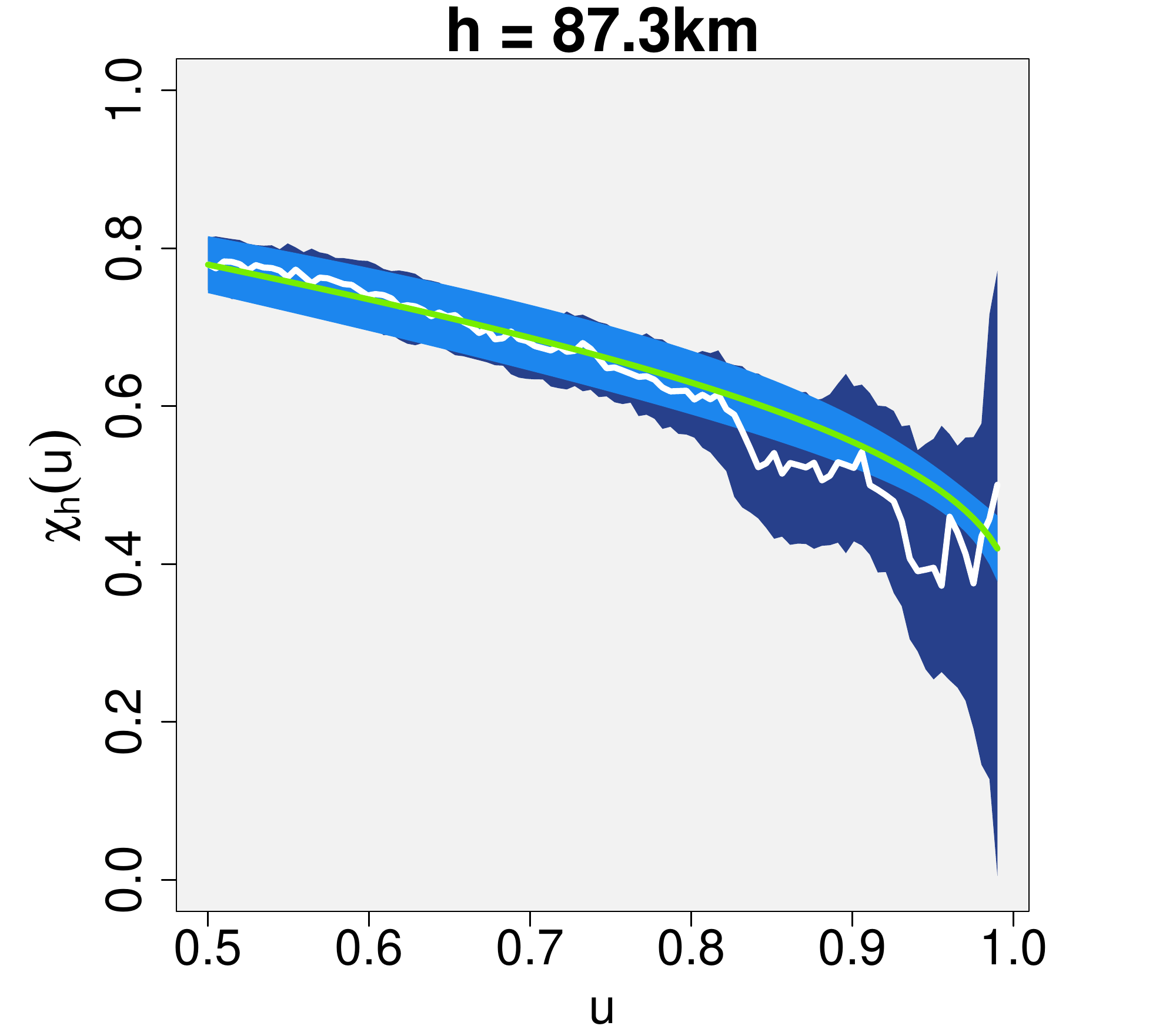}
\includegraphics[width=0.31\linewidth]{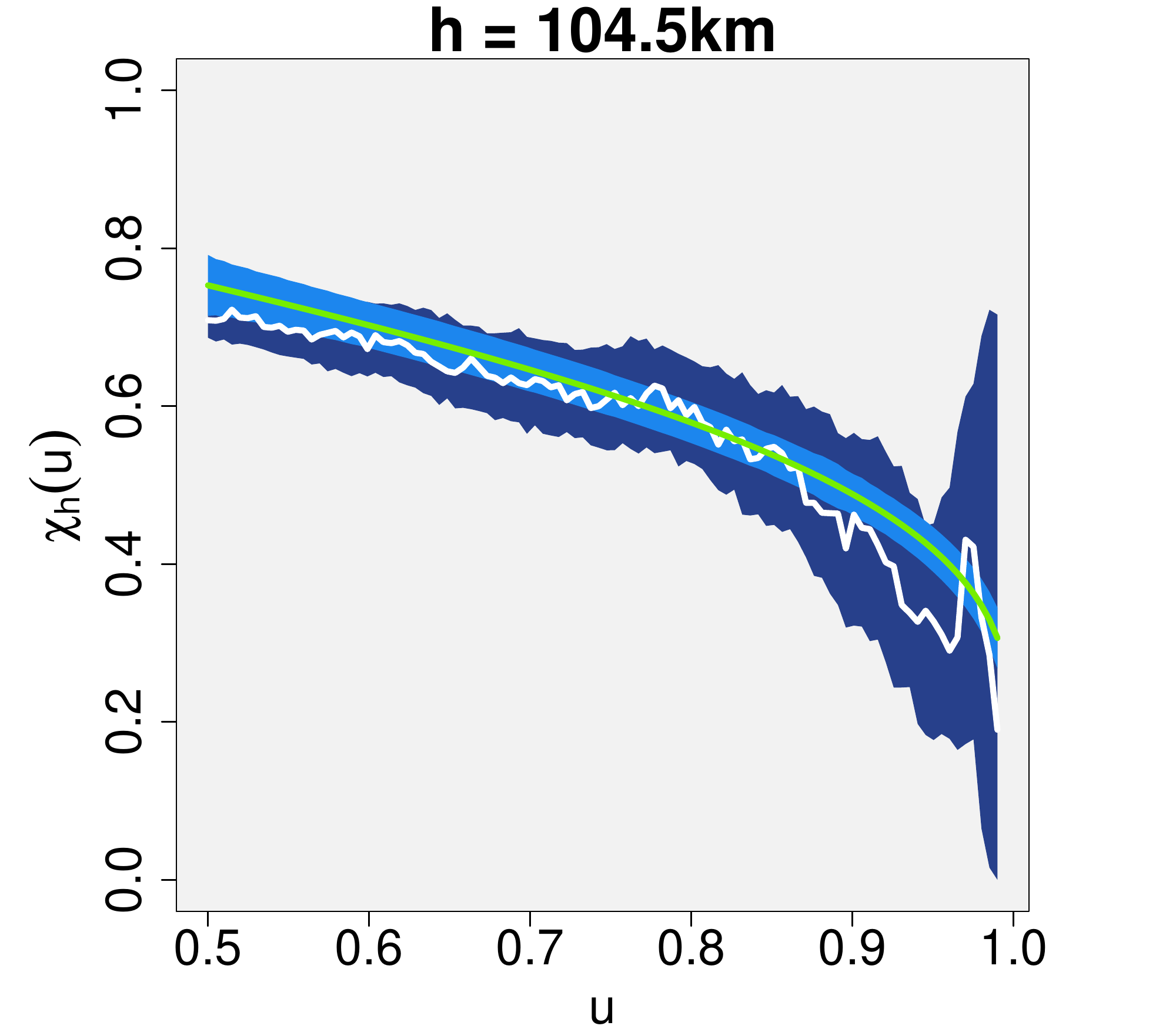}\\[5pt]
\hspace{3pt}\includegraphics[width=0.31\linewidth]{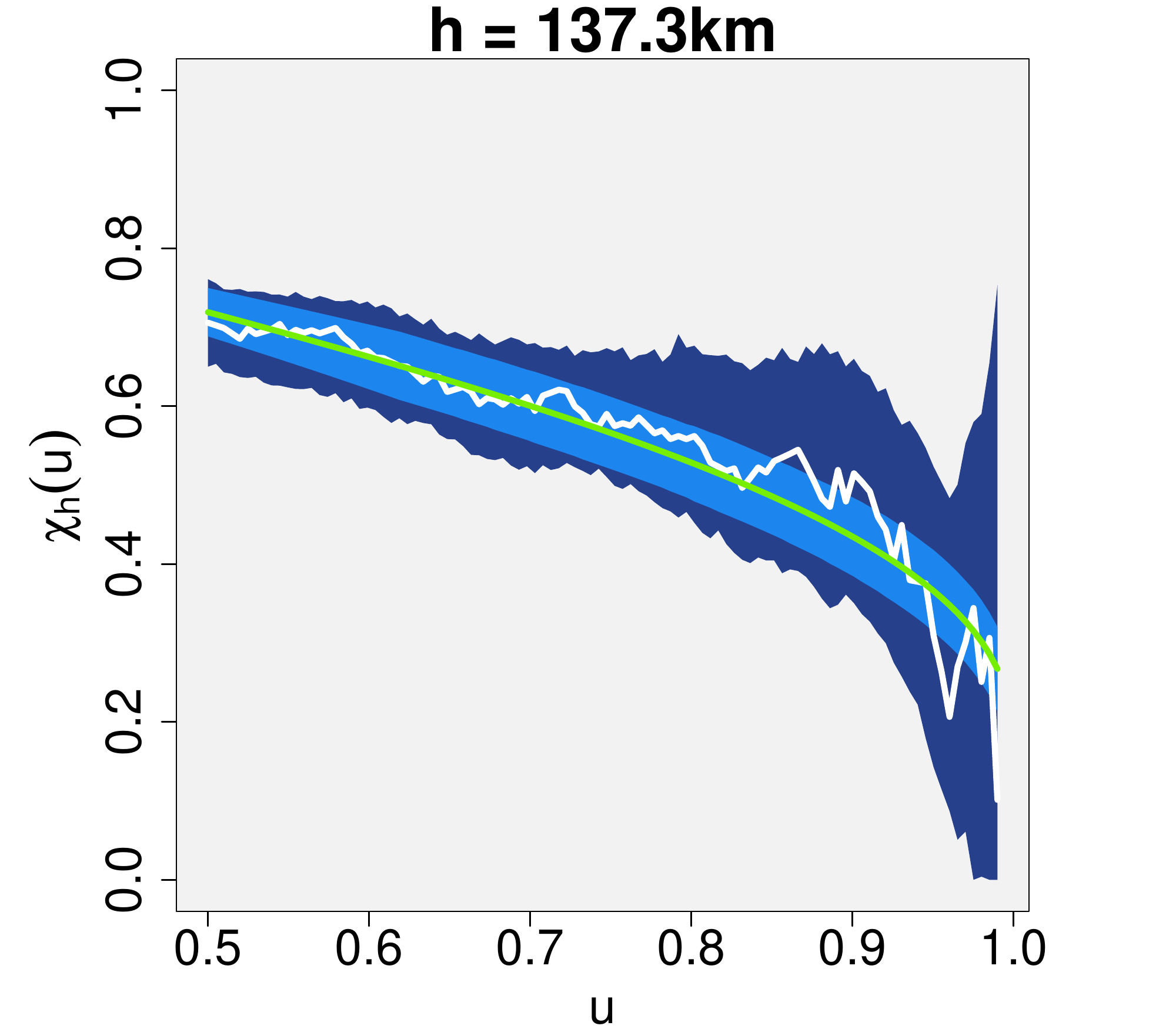}
\includegraphics[width=0.31\linewidth]{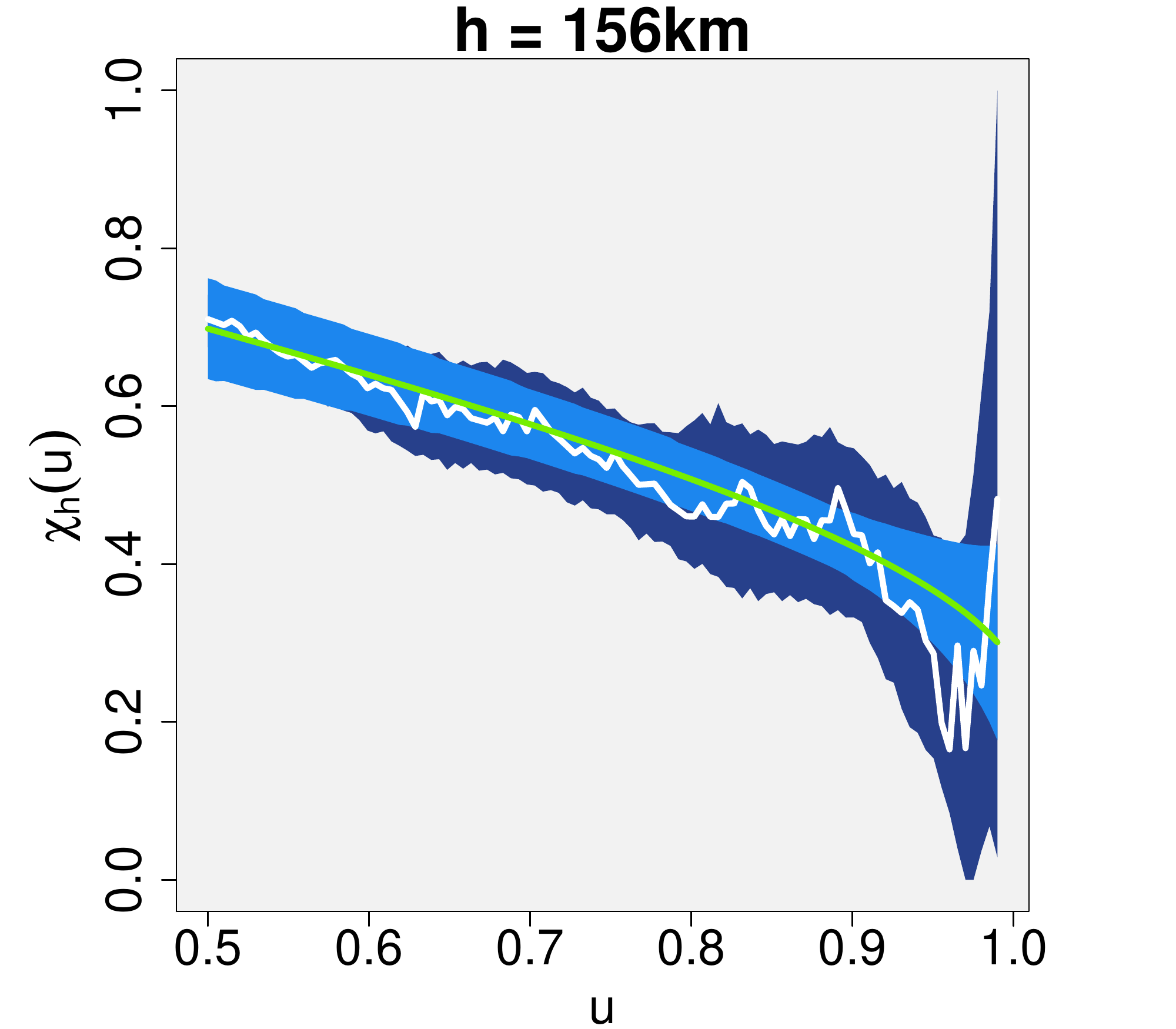}
\includegraphics[width=0.31\linewidth]{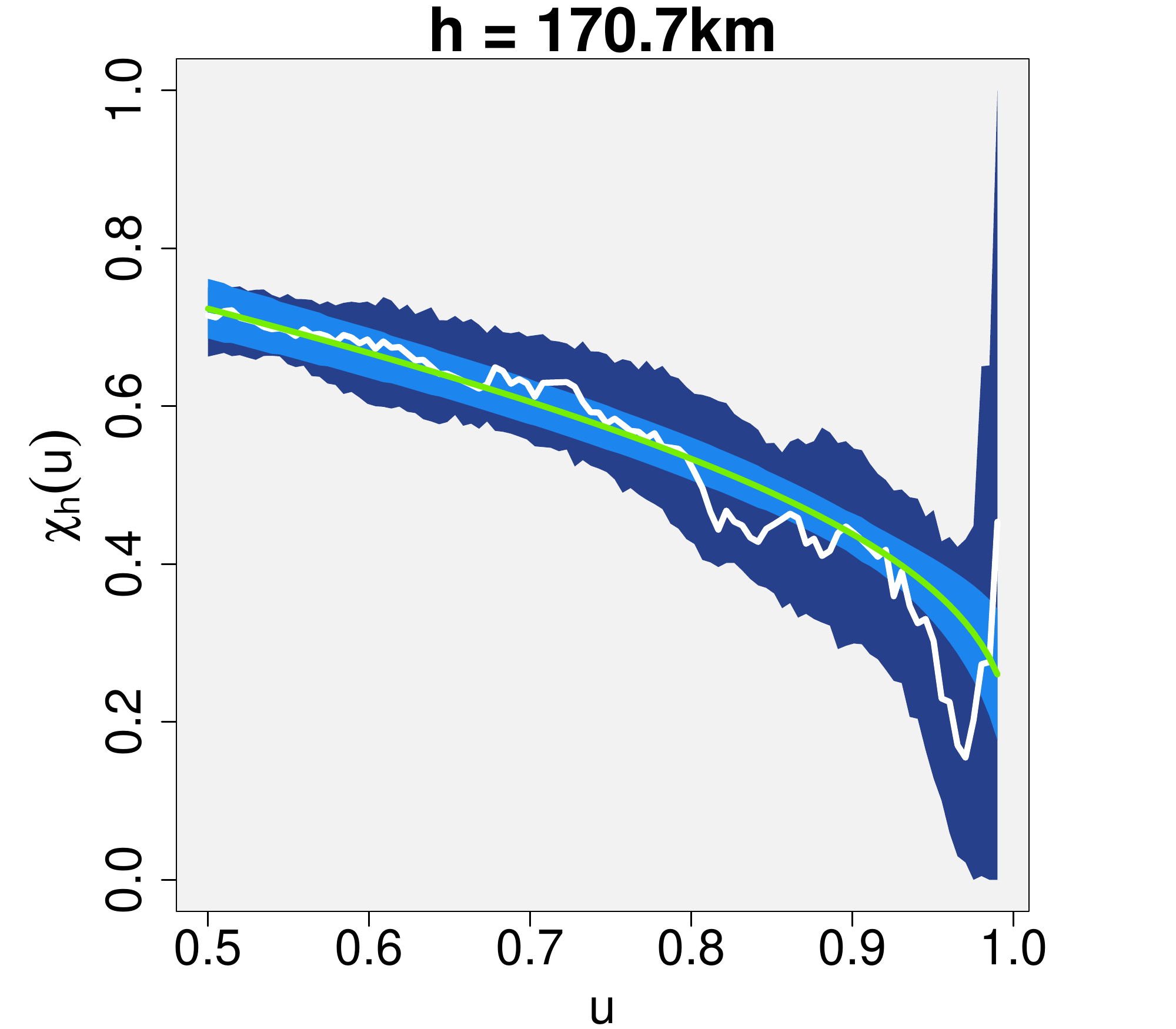}
\caption{\footnotesize {Empirical (white) and model-based (green) estimates for $\chi_h(u)$ as a function of the threshold $u\in[0.5, 0.995]$ for six selected pairs of stations, sorted according to their distance (labels 1--6 in Figure~\ref{fig:usmap.pdf}). Confidence envelopes for the empirical estimates (dark blue) are based on 300 block bootstrap samples with monthly blocks. Model-based confidence envelopes are also provided (light blue). Details on the fitted model are given in Sections~\ref{sec:Modeling} and \ref{DataApp}.}}
\label{fig:chiu}
\end{figure}
estimated for six selected pairs of stations $\{\mathbf{s}_1,\mathbf{s}_2\}\subset\mathcal{S}$ at various distances $h=\|\mathbf{s}_1-\mathbf{s}_2\|$, where $Y(\mathbf{s})$, $\mathbf{s}\in\mathcal{S}$, denotes the 5 day-cumulative winter precipitation random field defined over the contiguous U.S. (denoted {by} $\mathcal{S}$), and $F_j$ denotes the marginal distribution of $Y(\mathbf{s}_j)$ $j = 1,2$. While the function $\chi_h(u)$ in \eqref{eq:chiu} appears to {be decreasing as a function of the threshold $u$}, a generalized Pareto process \emph{always} assumes that $\chi_h(u)\equiv\chi_h$ is constant in $u$ \citep{rootzen2018multivariate} {for each lag $h$}, hence potentially leading to {an} overestimation of the dependence strength at higher quantiles. In Section~\ref{sec:Modeling}, we {describe a stationary factor copula model capturing} non-trivial tail dependence, i.e., allowing for a positive limit $\chi_h=\lim_{u\to1}\chi_h(u)>0$, while at the same time having a decreasing $\chi_h(u)$ function with known limiting extreme-value dependence structure. This approach has attractive modeling features, and is computationally more convenient than the models proposed by \citet{Huser.etal:2016} and \citet{Huser.Wadsworth:2017}. It differs fundamentally from the Laplace model of \citet{Opitz:2016}, which strongly focuses on capturing tail independence ($\chi_h=0$) and the decay rate towards the limit, rather than tail dependence ($\chi_h>0$), which is the type of extremal behavior suggested by our precipitation data at small {spatial} scales (or at least, $\chi_h(u)>0$ {remains positive} for quite high thresholds $u$); recall Figure~\ref{fig:chiu}.


Inference for extreme-value models is known to be particularly cumbersome. When max-stable processes are directly or indirectly involved, the likelihood function becomes excessively prohibitive to compute in moderate to large dimensions \citep{castruccio2016high, huser2019full}. This has lead to the use of less efficient composite likelihood approaches \citep{padoan2010likelihood,westra2011detection,Huser.Davison:2013a,Thibaud.etal:2013,huser2014space}. Alternatively, threshold approaches lead to simpler, significantly less demanding likelihoods, although their routine application to high-dimensional datasets is hampered by the need to censor observations involving non-extreme data (i.e., below a high threshold), which entails expensive multidimensional integrals; see, e.g., \citet{WadsworthTawn14}, \citet{ThibaudOpitz15}, \citet{Huser.etal:2016} and \citet{Huser.Wadsworth:2017}. Other censoring schemes have been investigated by \citet{Engelkeetal15} and \citet{Opitz:2016}, but they {may} cause {some} bias \citep{Huseretal16}. More recently, \citet{morris2017space} {proposed to} impute the censored, non-extreme, observations using a simulation-based algorithm within a Bayesian framework. In this work, although our full dataset comprises 1218 monitoring stations over the contiguous U.S., the dimensionality is considerably reduced by adopting a local estimation approach, making it possible to perform traditional censored likelihood inference in a robust and efficient way. 

Our main {methodological contributions} can be summarized as follows: we extend the stationary exponential factor copula model by allowing the parameters to change smoothly with locations, {and provide further theoretical results for it.} We {develop} a local likelihood approach adapted to the context of extremes to fit the stationary exponential factor copula model {locally}, censoring non-extreme observations. Finally, we {derive} simplified expression for the terms involved in the local censored likelihood function in order to compute maximum likelihood estimates in a reasonable time, while avoiding numerical {instabilities}.

In Section~\ref{sec:Modeling}, the stationary exponential factor copula model {is} presented, and {its} tail properties are studied.~Section~\ref{sec:Inference} discusses censored local likelihood inference based on high threshold exceedances.~In Section~\ref{Simulation}, a simulation study is conducted {based on a non-stationary factor copula model with spatially-varying parameters, in order} to assess the performance of our approach in various non-stationary contexts.~In Section~\ref{DataApp}, we apply {our proposed} model to study the dependence structure of heavy precipitation over the whole contiguous U.S., and we {use it to estimate return periods of spatial extreme events}.~Section~\ref{conclusion} concludes with {a comprehensive and critical} discussion.

\section{Modeling spatial extremes using factor copulas}\label{sec:Modeling}
\subsection{Copula models}
A copula is a multivariate probability distribution with $\mbox{Unif}(0,1)$ margins. Copulas are used to describe the dependence between random variables, and may be used to link univariate marginal distributions to construct a joint distribution. Specifically, let $(X_1,\ldots, X_D)^T\in\mathbb{R}^D$ be a random vector with continuous marginals $F_j(x) = \pr(X_j\leq x)$, $j = 1,\ldots,D$. The copula of $(X_1,\ldots,X_D)^T$ is defined through the random vector $(U_1,\ldots, U_D)^T = \{F_1(X_1), \ldots, F_D(X_D)\}^T$ {(with standard uniform margins)} as $C(u_1,\ldots,u_D) = \pr(U_1\leq u_1,\ldots,U_D\leq u_D)$. \citet{sklar1959fonctions} showed that each multivariate distribution $F(x_1,\ldots, x_D)$ with continuous margins $F_j(x)$ has a unique copula $C$, which may be expressed as
\begin{align}\label{copula}
F(x_1,\ldots x_D) &= C\{F_1(x_1),\ldots,F_D(x_D)\}\;\;\Longleftrightarrow\;\; C(u_1,\ldots, u_D) = F\{F_1^{-1}(u_1),\ldots,F_D^{-1}(u_D)\};
\end{align}
this implies that $F$ can be written in terms of univariate marginal distributions chosen independently from the dependence structure between the variables. This result can be associated with a two-step approach for inference, where margins are treated separately from the dependence structure. Based on \eqref{copula}, several copula families have been proposed and applied in practice for the modeling of environmental data, the most common one being the Gaussian copula obtained by choosing the joint distribution $F$ to be the standard multivariate Gaussian distribution $\Phi_D(\cdot;\mathbf{\Sigma})$ with correlation matrix $\mathbf{\Sigma}$. Other more flexible elliptical copulas may be derived similarly, such as the Student-$t$ copula, which is tail-dependent as opposed to the Gaussian copula. An alternative general family of copulas generating interesting tail dependence structures are factor copula models \citep{Krupskii.Joe:2015,krupskii2018factor}, in which a random and unobserved factor affects all measurements simultaneously. In Section~\ref{ExpStatLocalModel}, we describe the construction of the stationary exponential factor copula model proposed by \citet{krupskii2018factor}, which has appealing modeling and inference properties, and we then embed it in a more general non-stationary model in Section~\ref{ExpNonStatLocalModel}.

\subsection{The stationary exponential factor copula model}\label{ExpStatLocalModel}
Let $Z(\mathbf{s})$, $\mathbf{s}\in\mathcal{S}\subset\mathbb R^2$, be a standard Gaussian process with stationary correlation function $\rho(h)$ (see \citet{gneiting2006geostatistical} for a review of correlation functions), and let $V\geq0$ be an exponentially distributed random variable with rate parameter $\lambda>0$, independent of $Z(\mathbf{s})$, and which does not depend on the spatial location $\mathbf{s}\in\mathcal{S}$. The exponential factor copula model may be expressed in continuous space through the random process
\begin{align}\label{LocalModel}
W(\mathbf{s}) &= Z(\mathbf{s}) + V,\qquad \mathbf{s}\in\mathcal{S}.
\end{align}
The $W$ process in \eqref{LocalModel} is a Gaussian \emph{location} mixture, i.e., a standard Gaussian process with a random (exponentially distributed), constant mean. In this sense, it has similarities with the Student-$t$ process {(see e.g.,~\citeauthor{roislien2006t}, \citeyear{roislien2006t})}, which is a Gaussian \emph{scale} mixture (with standard deviation following a specific inverse gamma distribution). However, these models are not particular cases of each other and their dependence structures have significant differences.

Although Model \eqref{LocalModel} may seem quite artificial, it is only used to generate a flexible upper tail dependence structure, which we then fit to precipitation extremes. In other words, we disregard the margins and only consider the copula associated to \eqref{LocalModel}, which we fit to high threshold exceedances using a censored likelihood approach, reducing the contribution of small precipitation values; more details are given in Sections~\ref{sec:Inference} and \ref{DataApp}.

From~\eqref{LocalModel}, each configuration of $D$ sites $\{\mathbf{s}_1,\ldots,\mathbf{s}_D\}\subset\mathcal{S}$ yields a $D$-variate copula as follows: let $Z_j = Z(\mathbf{s}_j)$ and $W_j = W(\mathbf{s}_j)$, $j=1,\ldots,D$. The random vector $\mathbf{Z} = (Z_1,\ldots, Z_D)^T$ has a multivariate normal distribution with correlation matrix $\mathbf{\Sigma}_{\mathbf{Z}}$ that depends on the correlation function $\rho(h)$ and the sites' configuration, i.e., $\mathbf{Z}\sim\Phi_D(\cdot;\mathbf{\Sigma}_{\mathbf{Z}})$. The components of the random vector $\mathbf{W} = (W_1,\ldots, W_D)^T$ are $W_j = Z_j + V$, $j=1,\ldots,D$, where $V\sim{\rm Exp}(\lambda)$ is independent of $\mathbf{Z}$; by {integrating out} $V$, the joint distribution of $\mathbf{W}$ may be expressed as
\begin{equation}\label{cdfW}
F_D^{\mathbf{W}}(w_1,\ldots, w_D) =  \lambda \int_{0}^{\infty} \Phi_D(w_1 - v,\ldots, w_D - v;\mathbf{\Sigma}_{\mathbf{Z}}) \exp(-{\lambda v})\text{d}v,
\end{equation}
whilst its density is 
\begin{equation}\label{pdfW}
f_D^{\mathbf{W}}(w_1,\ldots, w_D) = \lambda \int_{0}^{\infty} \phi_D(w_1 - v,\ldots, w_D - v;\mathbf{\Sigma}_{\mathbf{Z}}) \exp(-{\lambda v})\text{d}v,
\end{equation}
where $\phi_D(\cdot;\mathbf{\Sigma})$ is the multivariate standard normal density with correlation matrix $\mathbf{\Sigma}$. In Appendix~\ref{simplified} we provide simpler and computationally efficient expressions to compute~\eqref{cdfW}, \eqref{pdfW}, and {partial derivatives of~\eqref{cdfW}}, avoiding the integral in $v$. Applying \eqref{copula}, the resulting copula and its density may be written as
\begin{align}\label{CopulaModel}
C_D^{\mathbf{W}}(u_1,\ldots, u_D) = F_D^{\mathbf{W}}(w_1,\ldots,w_D),\qquad c_D^{\mathbf{W}}(u_1,\ldots, u_D) &= \frac{f_D^{\mathbf{W}}(w_1,\ldots, w_D)}{\prod_{j = 1}^D f_1^{\mathbf{W}}(w_j)},
\end{align}
where $w_j = {(F_1^{\mathbf{W}}})^{-1}(u_j;\lambda)$, $j = 1,\dots, D$, and $F_1^{\mathbf{W}}(\cdot;\lambda)$ and $f_1^{\mathbf{W}}(\cdot;\lambda)$ denote the marginal distribution and density of the $W$ process, respectively. In particular, {we} can show that
\begin{align}
\label{cdfW1}
F_1^{\mathbf{W}}(w;\lambda) &= \Phi(w) - \exp(\lambda^2/2 - \lambda w)\Phi(w - \lambda),
\end{align}
where $\Phi(\cdot)$ is the standard normal distribution function; see \eqref{eq:marginalcdf} in Appendix \ref{simplified}.

To illustrate the tail flexibility of the stationary exponential factor copula model, Figure~\ref{fig:chiuStatFCM} displays the function $\chi_h(u)$ defined in \eqref{eq:chiu} as a function of the quantile level $u\in[0.95,1]$ and the Euclidean distance $h=\|\mathbf{s}_1-\mathbf{s}_2\|\geq0$, for different rate $\lambda>0$ and range $\delta>0$ parameters, assuming an exponential correlation function $\rho(h)=\exp(-h/\delta)$, $h\geq0$. While the range $\delta$ controls the correlation decay with distance, the rate $\lambda$ determines the overall tail dependence strength and strongly impacts the value of $\chi_h(u)$ and {its limit} $\chi_h = \lim_{u\to1}\chi_h(u)$ at large distances, i.e., as $h\to\infty$. 
\begin{figure}[t!]
\centering 
\includegraphics[scale=0.3]{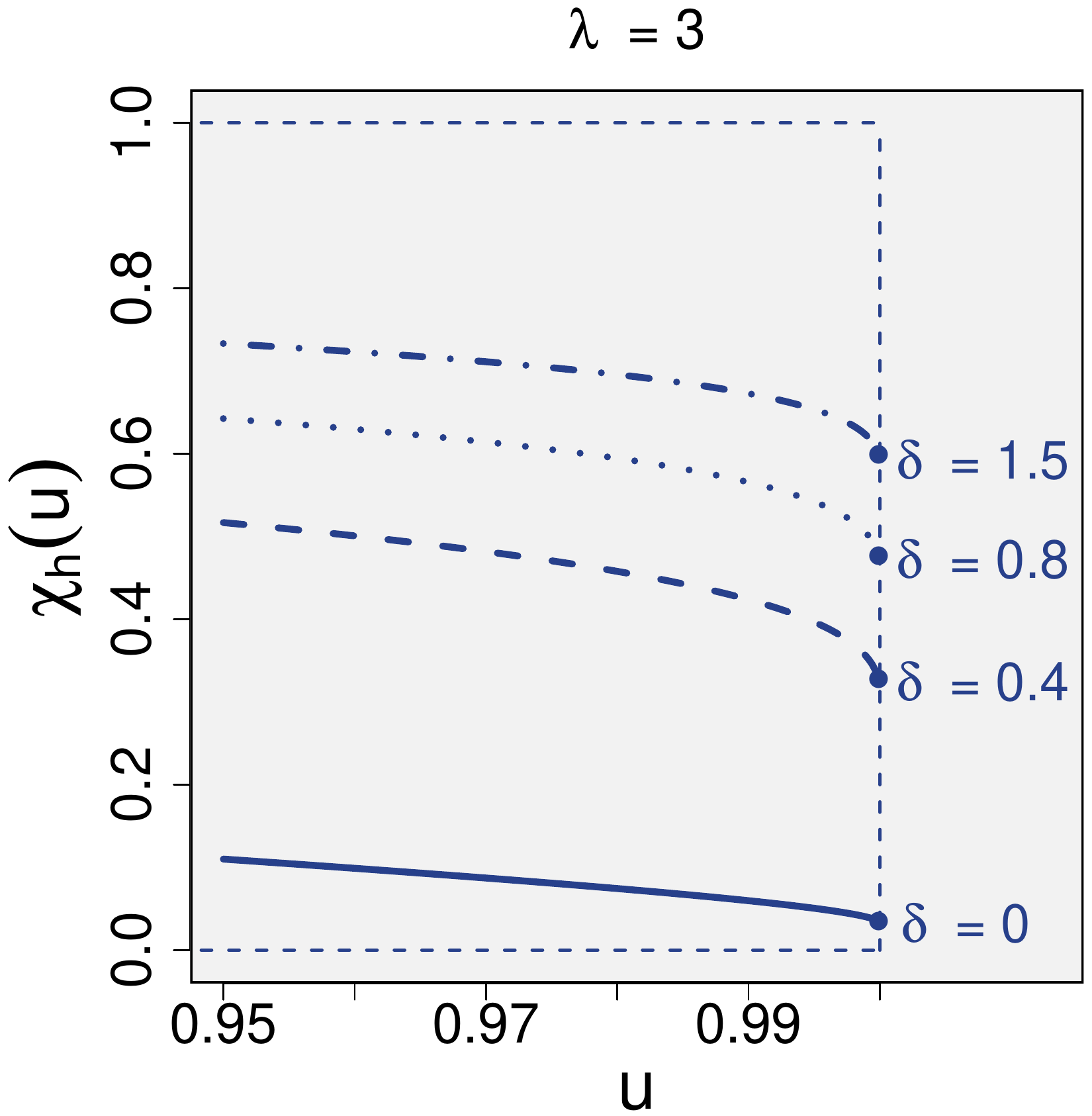}\hspace{20pt}
\includegraphics[scale=0.3]{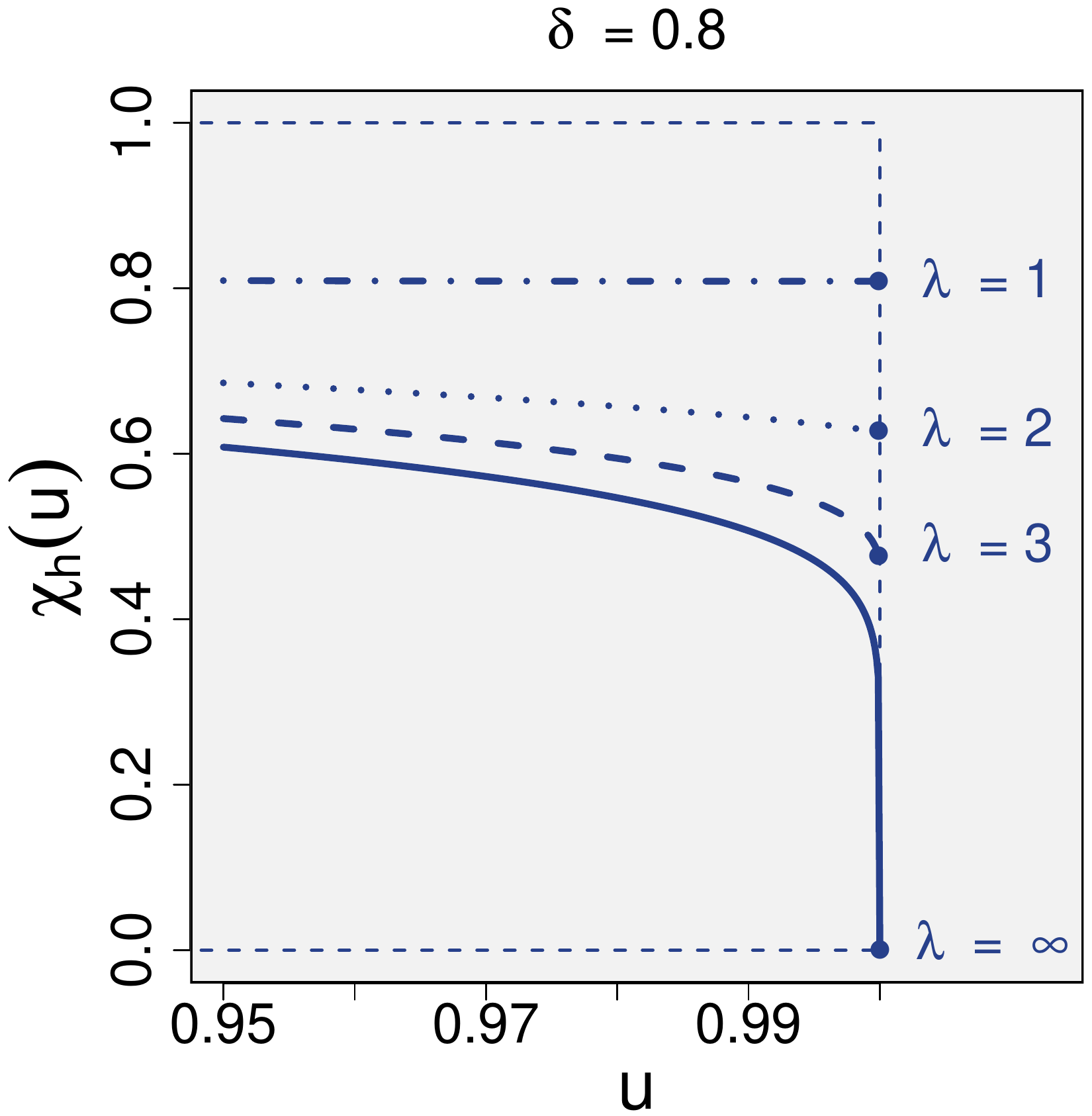}\\[5pt]
\includegraphics[scale=0.3]{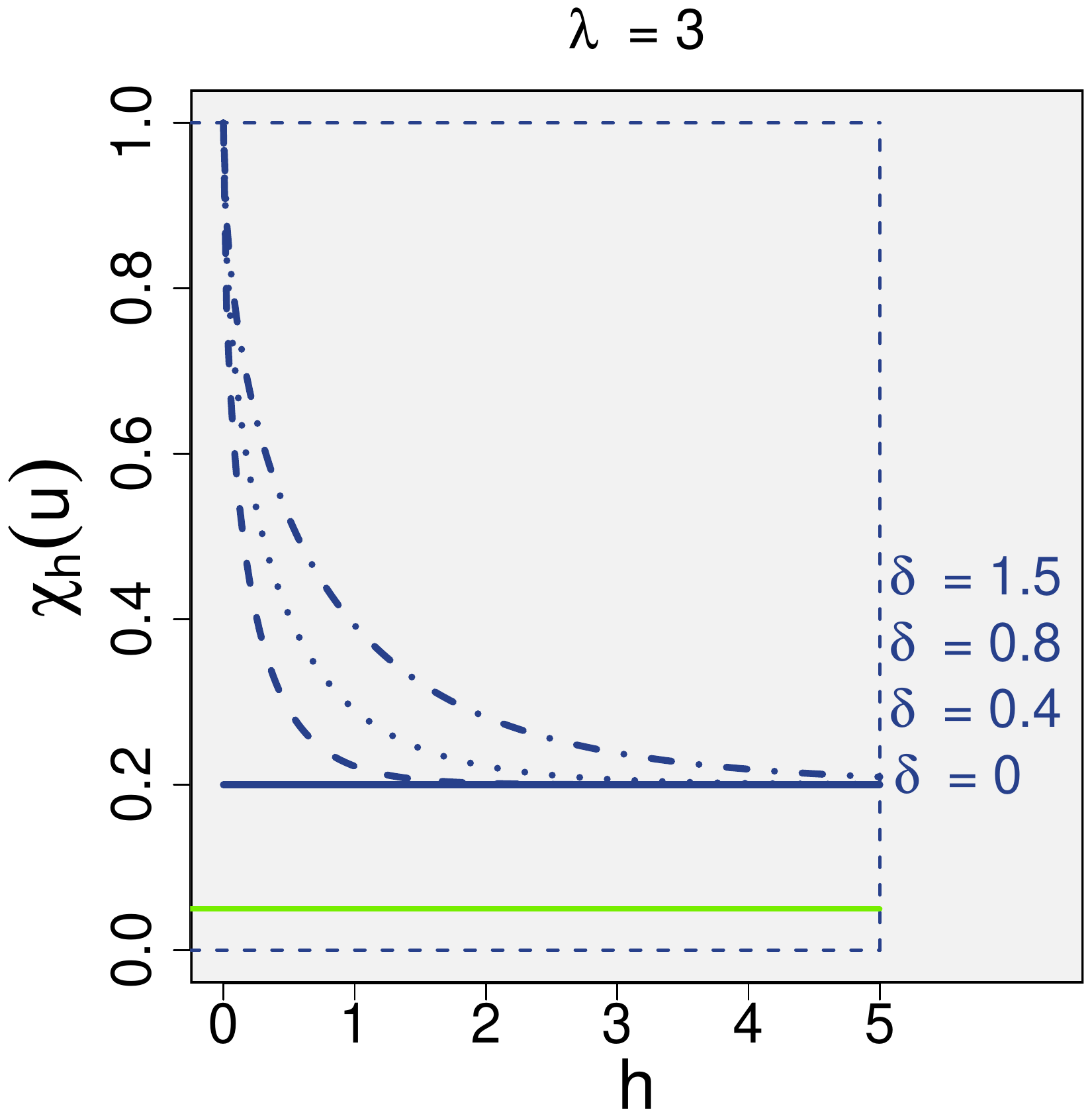}\hspace{20pt}
\includegraphics[scale=0.3]{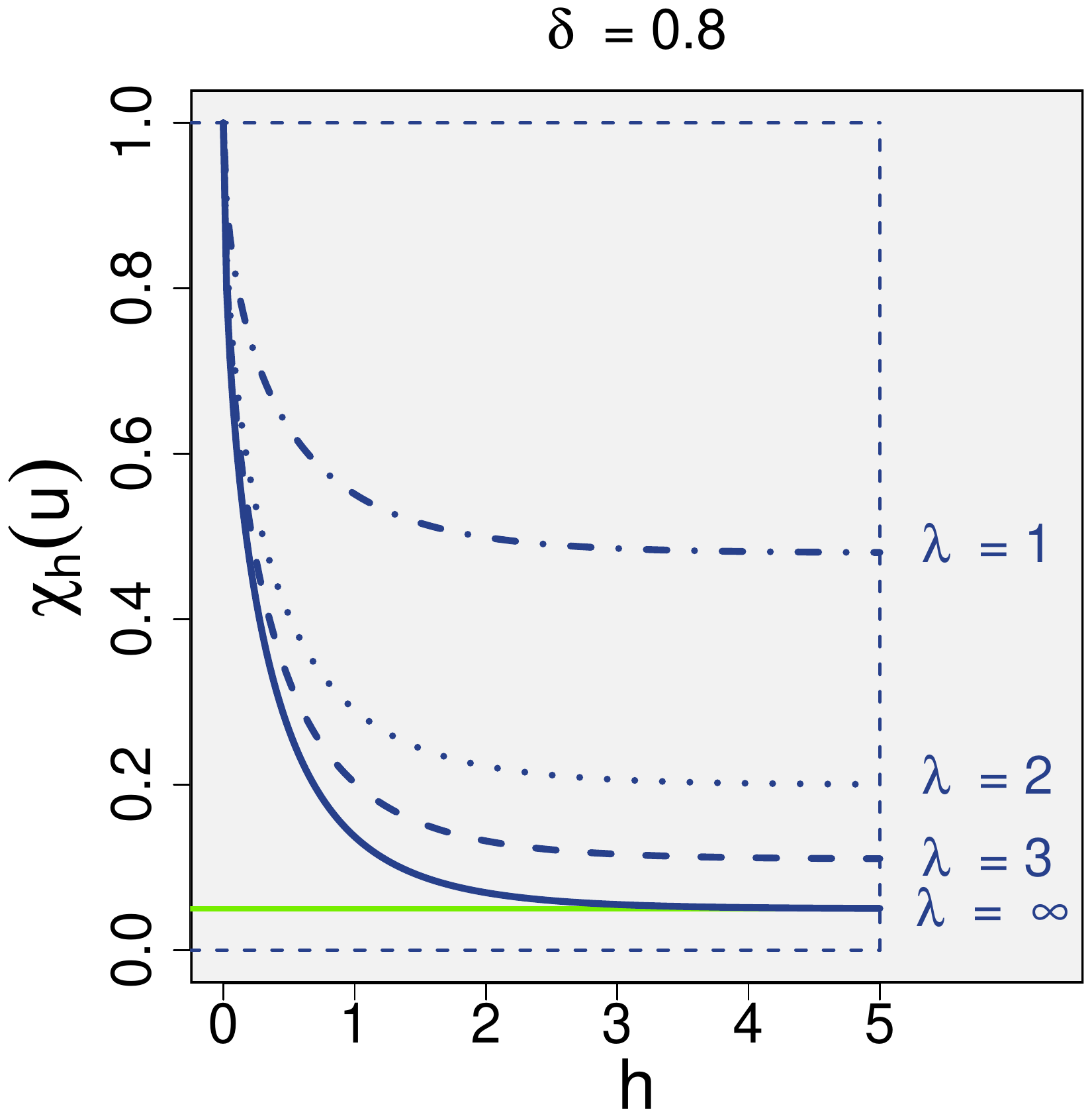}
\caption{\footnotesize Conditional probability $\chi_h(u)$ as defined in~\eqref{eq:chiu} for the stationary exponential factor model with rate parameter $\lambda>0$ and exponential correlation function $\rho(h)=\exp(-h/\delta)$, plotted with respect to the level $u\in[0.95,1]$ at distance $0.1$ (top) and with respect to the distance $h\in[0,5]$ for fixed $u=0.95$ (bottom). Left: fixed $\lambda=3$ and $\delta=0,0.4,0.8,1.5$. Right: fixed $\delta=0.8$ and $\lambda=1,2,3,\infty$. The dots in the top panels correspond to the limit $\chi_h=\lim_{u\to1}\chi_h(u)$, and the green horizontal lines in the bottom panels correspond to the complete independence case.}    
\label{fig:chiuStatFCM}
\end{figure}
In particular, because the random factor $V$ in \eqref{LocalModel} is common to all sites $\mathbf{s}\in\mathcal{S}$, spatial dependence in the $W$ process does not vanish as $h\to\infty$. Therefore, this stationary model may be useful for (replicated) spatial data collected on a local or regional scale, but it is generally not realistic on larger scales, such as the whole contiguous U.S. In Section~\ref{sec:Inference}, {we develop a local estimation approach assuming that the stationary model~\eqref{LocalModel} is only valid locally, and in Section~\ref{ExpNonStatLocalModel} we discuss} a flexible extension to {generate} non-stationary tail dependencies {structures}. As shown by \citet{krupskii2018factor}, the distribution of the random factor characterizes the tail properties of the copula~\eqref{CopulaModel}; in particular, the stationary exponential factor copula model is tail-dependent (i.e., $\chi_h>0$) for fixed $\lambda>0$, and its upper-tail dependence structure is governed the \citet{husler1989maxima} copula, which has been widely used for multivariate and spatial extremes \citep{davison2012statistical}. Moreover, as $\lambda\to\infty$, Model~\eqref{LocalModel} boils down to the Gaussian copula, which is tail-independent (i.e., $\chi_h=0$). Conversely, as $\lambda\to0$, the $W$ process is perfectly dependent over space. Thus, the exponential factor copula model interpolates between tail independence as $\lambda\to\infty$ and perfect dependence as $\lambda\to0$, while capturing a wide range of non-trivial tail dependence structures for $\lambda\in(0,\infty)$. 

We now provide more detailed information on the sub-asymptotic behavior of Model \eqref{LocalModel}, refining the description of its tail structure. The rate at which $\chi_h(u)$ converges to $\chi_h$, as $u\to1$, characterizes the flexibility of the process to capture sub-asymptotic extremal dependence. In practice, this is important, as the model will always be fitted at a finite threshold. Proposition~\ref{prop2} shows that the parameter $\lambda$ regulates this rate of convergence; the smaller $\lambda$, the faster the convergence, and vice versa.  We defer the proof to Appendix~\ref{tailprop}.

\begin{proposition}\label{prop2}
Consider Model \eqref{LocalModel} with rate $0<\lambda<\infty$. {We have} the expansion
\begin{eqnarray*}
\chi_h(u)-\chi_h=-f(u) + (2-\chi_h)\left\{\lambda s(u) + k(u)\right\}\{1+o(1)\},\qquad u\to1,
\end{eqnarray*}
where $f(u)\sim {2\phi\{z(u)\}\over (1-u)z(u)}$, $s(u)\sim{\phi\{z(u)\}\over(1-u)z(u)\{z(u)-\lambda\}}$ and, for $0<\rho(h)<1$, $k(u)\sim{2\phi\{z(u)-\lambda\}\over(2-\chi_h)\{z(u)-\lambda\}}$, with $z(u)={(F_1^{\mathbf{W}}})^{-1}(u;\lambda)\sim -\lambda^{-1}\log(1-u)+\lambda/2\{1+o(1)\}$ such that in the limit as $u\to1$,
\begin{eqnarray*}
\chi_h(u)-\chi_h=-f(u)\{1+o(1)\}.
\end{eqnarray*}
\end{proposition}

{For model~\eqref{ExpStatLocalModel}, we have $\chi_{h}=2\{1 - \Phi(\sqrt{\gamma_h}/2)\}$, where $\gamma_h=2\lambda^2\{1-\rho(h)\}$; see~\cite{krupskii2018factor} and Section 1 of the Supplementary Material}. The {measures $\chi_{h}(u)$ and $\chi_h$} summarize the effects of all dependence parameters, and we use them in our simulation experiments (Section~\ref{Simulation}) and the data application (Section~\ref{DataApp}) to study the performance of our model and to assess its ability to flexibly capture different levels of extremal dependence.

Other types of tail dependence structures can be obtained using alternative distributions for the random factor $V$ (see \citealp{krupskii2018factor}), but we here restrict ourselves to the exponential factor copula model, which yields flexible tail structures and fast inference.

\section{Local likelihood inference with partial censoring}\label{sec:Inference}
The assumption of stationarity underlying \eqref{LocalModel} is unrealistic over large heterogeneous regions, such as the whole contiguous U.S., but it may be the starting point for more sophisticated models; in particular, the true precipitation data generating process might be approximately stationary in small regions, providing support for non-stationary, but locally stationary, models. We here assume that Model \eqref{LocalModel} provides a good approximation to the local tail dependence structure while stemming from a more complex global data generating process (e.g., such as the one used in our simulation study in Section~\ref{ExpNonStatLocalModel}).

Local likelihood estimation for univariate threshold exceedances was proposed by \citet{Davison.Ramesh:2000}, while \citet{anderes2011local} investigated how such an approach may be applied in the spatial context based on a single realization from a Gaussian process. \cite{risser2015local} {used {a} local likelihood {approach} to estimate the spatially-varying parameters of a non-stationary Gaussian process, modeling locally-varying anisotropies using an approach similar to the discrete mixture kernel convolution of~\cite{higdon1998process}}. We here detail how to perform local likelihood estimation based on high threshold exceedances by adapting the methodology {of} \citet{anderes2011local} to the joint upper tail of the stationary copula model~\eqref{LocalModel}, under the assumption that it is valid in small regional neighborhoods. {Unlike~\cite{anderes2011local}}, we assume that multiple replicates of the process are observed, with possibly arbitrary marginal distributions.

Since our focus is on the data's {tail} dependence structure, we advocate a two-step semi-parametric estimation procedure, whereby margins are first estimated at each site separately using the empirical distribution function, and the copula model is then estimated locally in a second step by maximum likelihood, censoring low (i.e., non-extreme) values to prevent them from influencing the fit. Such a two-step approach is standard in the copula literature, and has been studied in depth; see, e.g., \citet{Genest.etal:1995} and \citet{joe2014dependence}, Chapter 5. 

More complex parametric approaches are also possible to estimate margins {in the first step}. Although marginal estimation is not our primary target, we studied the effect of the rank transformation {on the estimation of the} copula, by comparing it to {an approach based on} the generalized Pareto distribution, and the {optimal scenario where} marginal distributions {are known exactly}. We found that the use of the rank-based empirical distribution function provides a {fast}, robust, and reliable method, which does not affect much the subsequent estimation of the copula (see {{Section 2.1}} {of} the Supplementary Material).

The first step of our proposed estimation procedure consists in transforming the observed data non-parametrically to the uniform scale. Let $y_{1j},\ldots,y_{Nj}$ denote $N$ independent and identically distributed observations at the $j$th monitoring station, with essentially arbitrary margins. Pseudo-uniform scores may be obtained using ranks as follows:
$$u_{ij}={{\rm rank}(y_{ij})\over N+1},\qquad i=1,\ldots,N.$$
In the second step, the scores $u_{ij}$, $i=1,\ldots,N$, are treated as a perfect random sample from the ${\rm Unif}(0,1)$ distribution. To estimate the spatial copula structure, we then discretize the space $\mathcal S$ (in our case, the whole contiguous U.S.) into a fine grid $\mathcal G\subset\mathcal S$. For each grid point $\mathbf{s}_0\in\mathcal{G}$, we assume that the local stationary copula model \eqref{LocalModel} with parameters $\boldsymbol{\theta}_0\in \mathbf{\Theta}\subset\mathbb{R}^p$ is valid in a small neighborhood $\mathcal N_{\mathbf{s}_0}\subset\mathcal S$ around $\mathbf{s}_0$. In what follows, these regional neighborhoods will be determined by a small number, $D_0$, of nearest stations from the site $\mathbf{s}_0$, so we will write $\mathcal N_{\mathbf{s}_0}\equiv\mathcal N_{\mathbf{s}_0;D_0}$.
Obviously, the choice of neighborhood is important: the stationary model \eqref{LocalModel} might be a poor approximation for large neighborhoods (i.e., for large $D_0$), whereas model fitting might be cumbersome for small neighborhoods characterized by a small number $D_0$ of stations. This bias-variance trade-off is tricky to deal with, and Section~\ref{DataApp} describes one possible approach to mitigate this issue.

To estimate the local tail dependence structure, we suggest using a censored likelihood approach, which is standard in statistics for spatial extremes, though it has never been applied to Model \eqref{LocalModel}; see, e.g., \citet{Ledford.Tawn:1996}, \citet{Thibaud.etal:2013}, \citet{huser2014space}, \citet{WadsworthTawn14}, \citet{ThibaudOpitz15} and \citet{Huser.etal:2016}. Essentially, vectors with non-extreme observations (i.e., lower than a high threshold) are partially or fully censored to prevent these low values from affecting the estimation of the extremal dependence structure. More precisely, for each grid point $\mathbf{s}_0\in\mathcal{G}$ with associated neighborhood $\mathcal N_{\mathbf{s}_0;D_0}$, let $\mathbf{u}^0_j=(u_{1j}^0,\ldots,u_{Nj}^0)^T$, $j=1,\ldots,D_0$, denote the pseudo-uniform scores for each of the nearby stations in $\mathcal N_{\mathbf{s}_0;D_0}$, and let $u^\star_j$, $j=1,\ldots,D_0$, be high marginal thresholds; in our application in Section~\ref{DataApp}, we take $u^\star_j=0.80$ for all $j$. Using the notation introduced in \eqref{cdfW}--\eqref{cdfW1}, the censored local log-likelihood may be expressed as
\begin{align}
\ell(\boldsymbol{\theta}_0 \mid \mathcal N_{\mathbf{s}_0;D_0}) 
&= \sum_{i \in \mathcal I_{\rm nc}} \log f_{D_0}^{\mathbf{W}}(w_{i1}^0,\ldots,w_{iD_0}^0;\boldsymbol{\theta}_0) - \sum_{i \in \mathcal I_{\rm nc}}\sum_{j = 1}^{D_0} \log f_1^{\mathbf{W}}(w_{ij}^0;\lambda_0) \label{WLL}\nonumber\\
&+ N_{\rm fc}\times \log F_{D_0}^{\mathbf{W}}(w_{1}^\star,\ldots,w_{D_0}^\star;\boldsymbol{\theta}_0)\nonumber\\
&+  \sum_{i \in \mathcal I_{\rm pc}}\log \partial_{J_i}F_{D_0}^{\mathbf{W}}(\max(w_{i1}^0,w_{1}^\star),\ldots,\max(w_{iD_0}^0,w_{D_0}^\star);\boldsymbol{\theta}_0) - \sum_{i \in \mathcal I_{\rm pc}}\sum_{j\in J_i}\log f_1^{\mathbf{W}}(w_{ij}^0; \lambda_0),
\end{align}
where $w_{ij}^0 = (F_1^{\mathbf{W}})^{-1}(u_{ij}^0; \lambda_0)$ and $w_{j}^\star = (F_1^{\mathbf{W}})^{-1}(u_{j}^\star; \lambda_0)$, $j=1,\ldots,D_0$, $\mathcal I_{\rm nc}=\{i\in\{1,\ldots,N\}: u_{ij}>u_j^\star, \forall j = 1,\ldots, D_0\}$ is the index set of all \emph{non-censored} observations (i.e., all vector components are extreme), $\mathcal I_{\rm fc}=\{i\in\{1,\ldots,N\}: u_{ij}\leq u_j^\star, \forall j = 1,\ldots, D_0\}$ is the index set of all \emph{fully censored} observations (i.e., none of the vector components are extreme) with $N_{\rm fc}=|\mathcal I_{\rm fc}|$, $\mathcal I_{\rm pc}=\{1,\ldots,N\}\setminus\{\mathcal I_{\rm nc}\cup\mathcal I_{\rm fc}\}$ is the index set of all \emph{partially censored} observations (i.e., some, but not all, vector components are extreme), $J_i=\{j\in\{1,\ldots,D_0\}:u_{ij}> u_j^\star\}$ is the index set of threshold exceedances for the $i$th vector of observations, and $\partial_{J_i}$ denotes differentiation with respect to the variables indexed by the set $J_i$. Numerical maximization of \eqref{WLL} yields the maximum likelihood estimates $\hat{\boldsymbol{\theta}}_0$ for location $\textbf{s}_0$. In Appendix~\ref{simplified}, we provide simple expressions for $f_{D_0}^{\mathbf{W}}$, $F_{D_0}^{\mathbf{W}}$ and $\partial_{J_i}F_{D_0}^{\mathbf{W}}$, which involve the $D_0$-variate Gaussian density and the multivariate Gaussian distribution in dimension $D_0$ and $D_0-|J_i|$, respectively. When $D_0$ is large, the computation of the multivariate Gaussian distribution can significantly slow down the estimation procedure. However, thanks to our local approach, $D_0$ is typically quite small, and this allows to estimate the model at a reasonable computational cost. Furthermore, because the likelihood maximizations can be done independently at each grid point $\mathbf{s}_0\in\mathcal{G}$, we can easily take advantage of distributed computing resources to perform each fit in parallel.

Following \citet{anderes2011local}, it {is} possible to generalize \eqref{WLL} to obtain smoother parameter estimates over space. Specifically, {we can} instead maximize the weighted log-likelihood function defined as
\begin{align}
\label{WLL2}
\ell_\omega(\boldsymbol{\theta}_0 \mid \mathcal N_{\mathbf{s}_0;D_0})  &= \sum_{j = 1}^{D_0} \omega(\|\mathbf{s}_j-\mathbf{s}_0\|)\left\{\ell(\boldsymbol{\theta}_0 | \mathcal N_{\mathbf{s}_0;j})  - \ell(\boldsymbol{\theta}_0 | \mathcal N_{\mathbf{s}_0;j-1})\right\},
\end{align}
where $\omega(h)\geq0$ is a non-negative weight function, and $\{\mathcal N_{\mathbf{s}_0;j};j=0,\ldots,D_0\}$, denotes the nested sequence of subsets comprising the first $j$ nearest neighbors of $\mathbf{s}_0$, with the convention that $\mathcal N_{\mathbf{s}_0;0}=\emptyset$ and $\ell(\boldsymbol{\theta}_0 | \emptyset)=0$. Choosing hard-thresholding weights with $\omega(\|\mathbf{s}_j-\mathbf{s}_0\|)=1$ for all $\mathbf{s}_j\in \mathcal N_{\mathbf{s}_0;D_0}$ boils down to \eqref{WLL}. However, smoother parameter estimates may be obtained by selecting soft-thresholding weights that smoothly decay to zero near the neighborhood boundaries, e.g., using the biweight function $\omega(h)=(1-(h/\tau_0)^2)_+^2$, for some bandwidth $\tau_0>0$. Although this estimation approach seems quite appealing, it significantly increases the computational burden, since the weighted log-likelihood function \eqref{WLL2} requires computing \eqref{WLL} in dimensions $1,\ldots,D_0$ instead of just once in dimension $D_0$. {In {Section 2.2 of} the Supplementary Material, we compare hard- and soft-thesholding weights in a simulation study, and do not notice any major difference in the results. This may be explained by the fact that, unlike \cite{anderes2011local}, multiple time replicates are available to fit the model}. For {these} reasons, and because the censored likelihood procedure is already quite intensive, we do not pursue this approach further in this paper. Subsequent simulation and real data experiments are all based on \eqref{WLL}.

{Remark that}, as mentioned above, we assume that Model \eqref{LocalModel} provides a good approximation to the local tail dependence structure while stemming from a more complex global data generating process. To make sure that the local stationary model \eqref{LocalModel} can truly come from a well-defined global stochastic process, {we describe in Section~\ref{ExpNonStatLocalModel}} one possible way to embed the model \eqref{LocalModel} into a non-stationary process, which we use in our simulations.

\section{Simulation study}\label{Simulation}

\subsection{Goals of our simulation experiments}
In this section, we study the flexibility of our copula model \eqref{LocalModel} to locally describe complex extremal dynamics across space, and we assess the performance of our {censored} local estimation approach based on \eqref{WLL} at capturing such non-stationary dependence structures. We also analyze the sensitivity of the parameter estimates to the neighborhood size {$D_0$}, {and the sensitivity of a profile likelihood procedure to select the smoothness parameter {$\nu$}}.

\subsection{Data-generating scheme and simulation scenarios}\label{ExpNonStatLocalModel}
{We} first describe how we can extend the stationary exponential factor copula model~\eqref{LocalModel} to a non-stationary process that we use in our simulations. Further details and asymptotic results are provided in Section 1 {of} the Supplementary Material. Our proposed model extension, equivalent to \eqref{LocalModel} on infinitesimal regions, is to consider the process
\begin{align}\label{GlobalModel}
W(\mathbf{s}) &= Z(\mathbf{s}) + \lambda_{\mathbf{s}}^{-1}E,\qquad \mathbf{s}\in\mathcal{S},
\end{align}
where $Z(\mathbf{s})$ is a zero mean Gaussian process with non-stationary correlation function $\rho(\mathbf{s}_1,\mathbf{s}_2)$ and $E$ is a standard exponentially distributed common factor, independent of $Z(\mathbf{s})$. The rate parameter $\lambda_{\mathbf{s}}>0$, $\mathbf{s}\in\mathcal{S}$, is assumed to be a smooth surface, which dictates different degrees of tail dependence in distinct regions. {The model in~\eqref{GlobalModel} is also tail-dependent and the limit $\chi_{12} \equiv \chi_{12}(\mathbf{s}_1, \mathbf{s}_2)=\lim_{u\to1}\chi_{12}(u)$, with $\chi_{12}(u)$ defined analogously to~\eqref{eq:chiu} but in the non-stationary context, is $\chi_{12}=2\{1-\Phi(\sqrt{\gamma_{12}}/2)\}$ with $\gamma_{12} = \lambda_{\mathbf{s}_1}^2 -2\rho(\mathbf{s}_1,\mathbf{s}_2)\lambda_{\mathbf{s}_1}\lambda_{\mathbf{s}_2}+\lambda_{\mathbf{s}_2}^2$.} Although \eqref{GlobalModel} may not realistically capture long-distance dependencies owing to the latent factor $E$ being constant over space, its spatially-varying parameters describe the local dependence structure more flexibly; Model \eqref{GlobalModel} is therefore useful to ``think'' about the results  \emph{globally}, while generating various forms of extremal dependence \emph{locally} (or \emph{regionally}).

In the literature, different non-stationary correlation functions $\rho(\mathbf{s}_1,\mathbf{s}_2)$ have been proposed; see, e.g., \citet{Fuentes:2001}, \citet{Nychka.etal:2002}, \citet{Stein:2005b}, and \citet{paciorek2006spatial}. Here, we focus on a non-stationary, locally isotropic, Mat\'ern correlation function with constant smoothness parameter $\nu>0$, constructed through the kernel convolution approach advocated by \citet{paciorek2006spatial}. {Specifically, we choose}
\begin{align}\label{corr}
\rho(\mathbf{s}_1,\mathbf{s}_2) &= \frac{2^{2-\nu}\delta_{\mathbf{s}_1} \delta_{\mathbf{s}_2}}{\Gamma(\nu)(\delta_{\mathbf{s}_1}^2 + \delta_{\mathbf{s}_2}^2)}\mathcal{K}_\nu\left(\frac{2\sqrt{2\nu}}{\sqrt{\delta_{\mathbf{s}_1}^2+\delta_{\mathbf{s}_2}^2}}\|\mathbf{s}_1 - \mathbf{s}_2\|\right),
\end{align}
where $\delta_{\mathbf{s}}>0$, $\mathbf{s}\in\mathcal{S}$, is a smoothly varying range parameter, $\Gamma(\cdot)$ is the Gamma function and $\mathcal{K}_{\nu}$ is the modified Bessel function of second kind of order $\nu$. The stationary Mat\'ern correlation function (obtained by setting $\delta_{\mathbf{s}}\equiv\delta$ for all $\mathbf{s}\in\mathcal{S}$) has become popular because of its appealing properties \citep{Stein:1999}. In particular, a Gaussian process with Mat\'ern correlation function is $m$ times mean-square differentiable if and only if $\nu>m$. For $\nu=0.5$, it boils down to the exponential correlation function, which yields continuous but non-differentiable sample paths. As $\nu\to\infty$, sample paths are infinitely differentiable. The non-stationary correlation function defined in \eqref{corr} is locally Mat\'ern, and therefore it inherits these attractive properties. An extension of \eqref{corr} allowing for varying degrees of smoothness over space has been proposed in the unpublished manuscript of \citet{Stein:2005b} (see also \citeauthor{anderes2011local},~\citeyear{anderes2011local}), but in practice, estimating $\nu$ is cumbersome and conservative approaches are usually adopted. 

We simulated data on a $25\times25$ grid in $\mathcal S=[1,10]^2$ from the copula model stemming from \eqref{GlobalModel} based on the non-stationary Mat\'ern correlation function \eqref{corr}; $500$ independent replicates were generated. We chose three different levels of smoothness by fixing $\nu=0.5,1.5,2.5$, and we considered three scenarios for the range $\delta_{\mathbf s}$ and the rate $\lambda_{\mathbf s}$ parameters, representing various levels of non-stationarity in the tail behavior{; see Figure~\ref{SimSce.pdf}}. The true parameter values and the {bivariate $\chi_{12}$ measure} with respect to three different reference points are shown in Figure~\ref{SimSce.pdf} for the weakly, mildly and strongly non-stationary scenarios. In all simulations, the smoothness parameter $\nu$ was held fixed, while the rate $\lambda_{\mathbf{s}}$ and range $\delta_{\mathbf{s}}$ parameters were estimated on a $10\times 10$ grid $\mathcal G\subset[1,10]^2$ using the local estimation approach with censoring thresholds $u_j^\star=0.95$ for all $j$, as described in Section~\ref{sec:Inference}. Thus, there were $500\times0.05=25$ exceedances per location. To compute performance metrics and to measure the uncertainty of estimated parameters, we replicated all simulation experiments $1000$ times.

\begin{figure}[t!]
    	\begin{tabular}{ccc}
		\vspace{0.3cm}
    		\footnotesize \rotatebox{90}{\textbf{\hspace{-1.2cm}{\color{white}{aa}}Log-rate}}
		\hspace{-1cm}
    		\begin{minipage}[c]{0.5\linewidth}
			\begin{center}
				{\hspace{-2 cm}\footnotesize \textbf{Weakly non-stat.}}
			\end{center}
			\vspace{-0.5cm}
    			\includegraphics[scale=0.23]{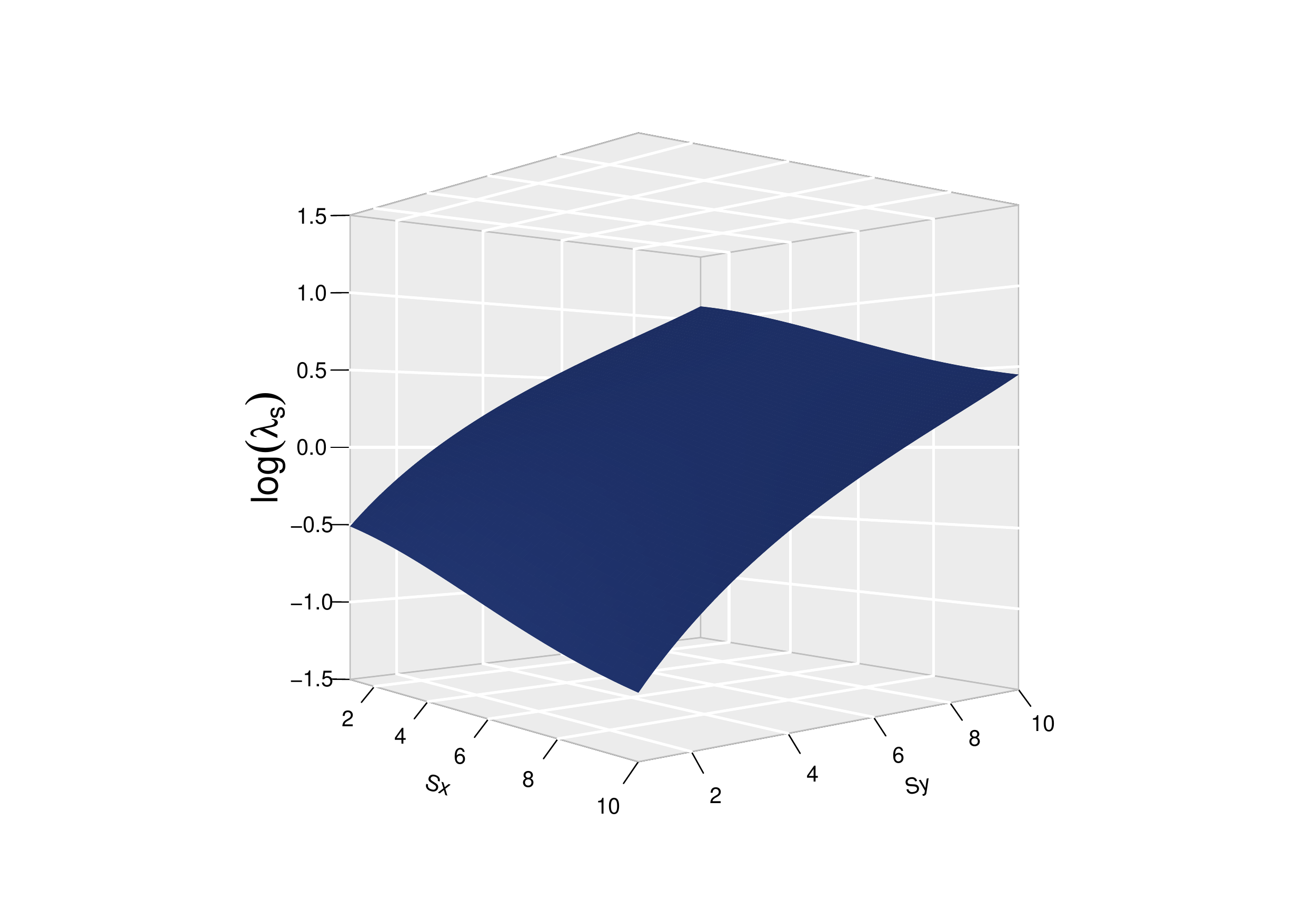}
      			\vspace{-.3cm}
     		\end{minipage}&
     		\hspace{-5cm}
     		\begin{minipage}[c]{0.5\linewidth}
			\begin{center}
				{\hspace{-2 cm}\footnotesize \textbf{Mildly non-stat.}}
			\end{center}
			\vspace{-0.5cm}
     			\includegraphics[scale=0.23]{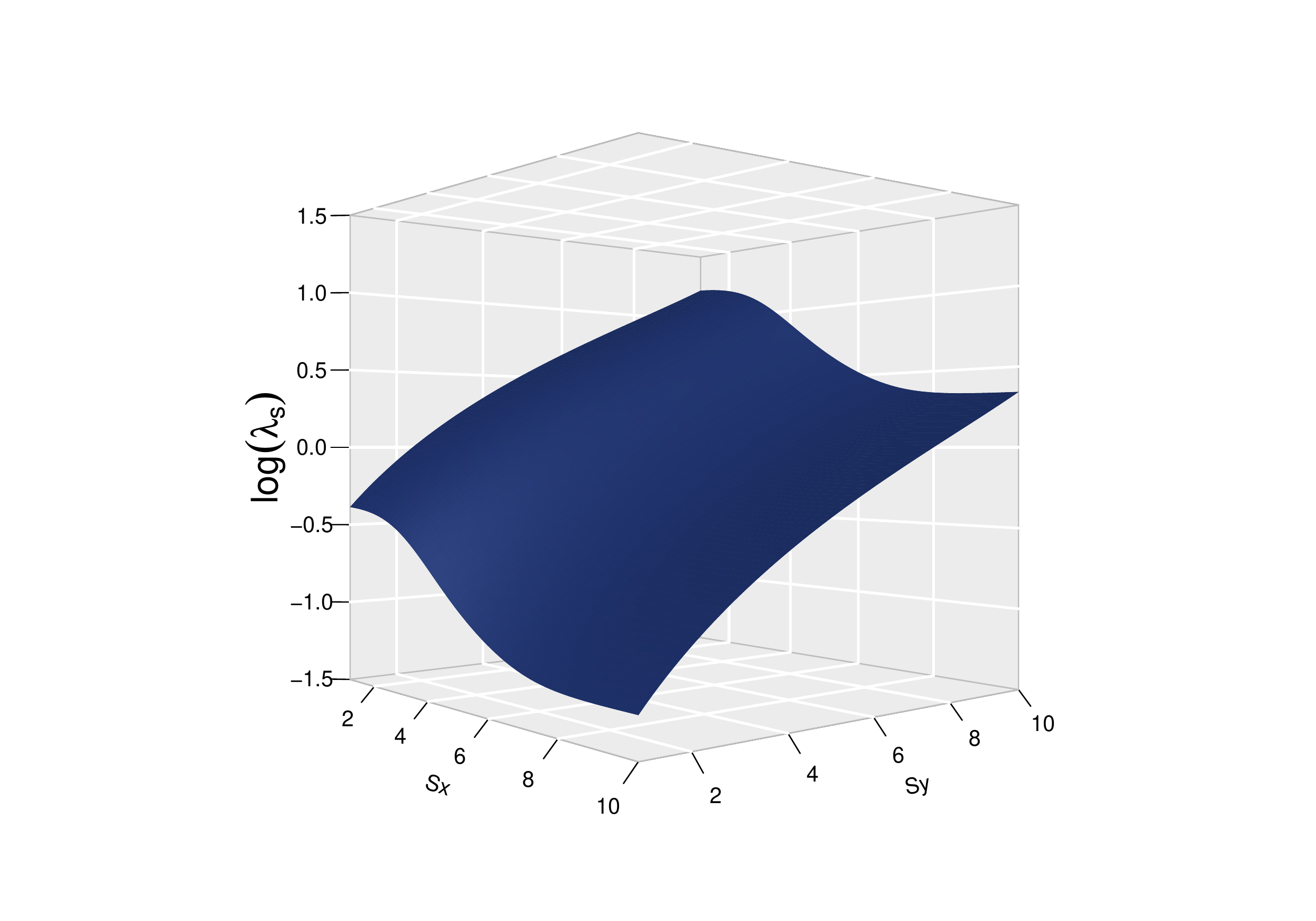}
    			 \vspace{-.3cm}
     		\end{minipage}
           	\hspace{-4.6cm}
           	\begin{minipage}[c]{0.5\linewidth}
			\begin{center}
				{\hspace{-2 cm}\footnotesize \textbf{Strongly non-stat.}}
			\end{center}
			\vspace{-0.5cm}
           		\includegraphics[scale=0.23]{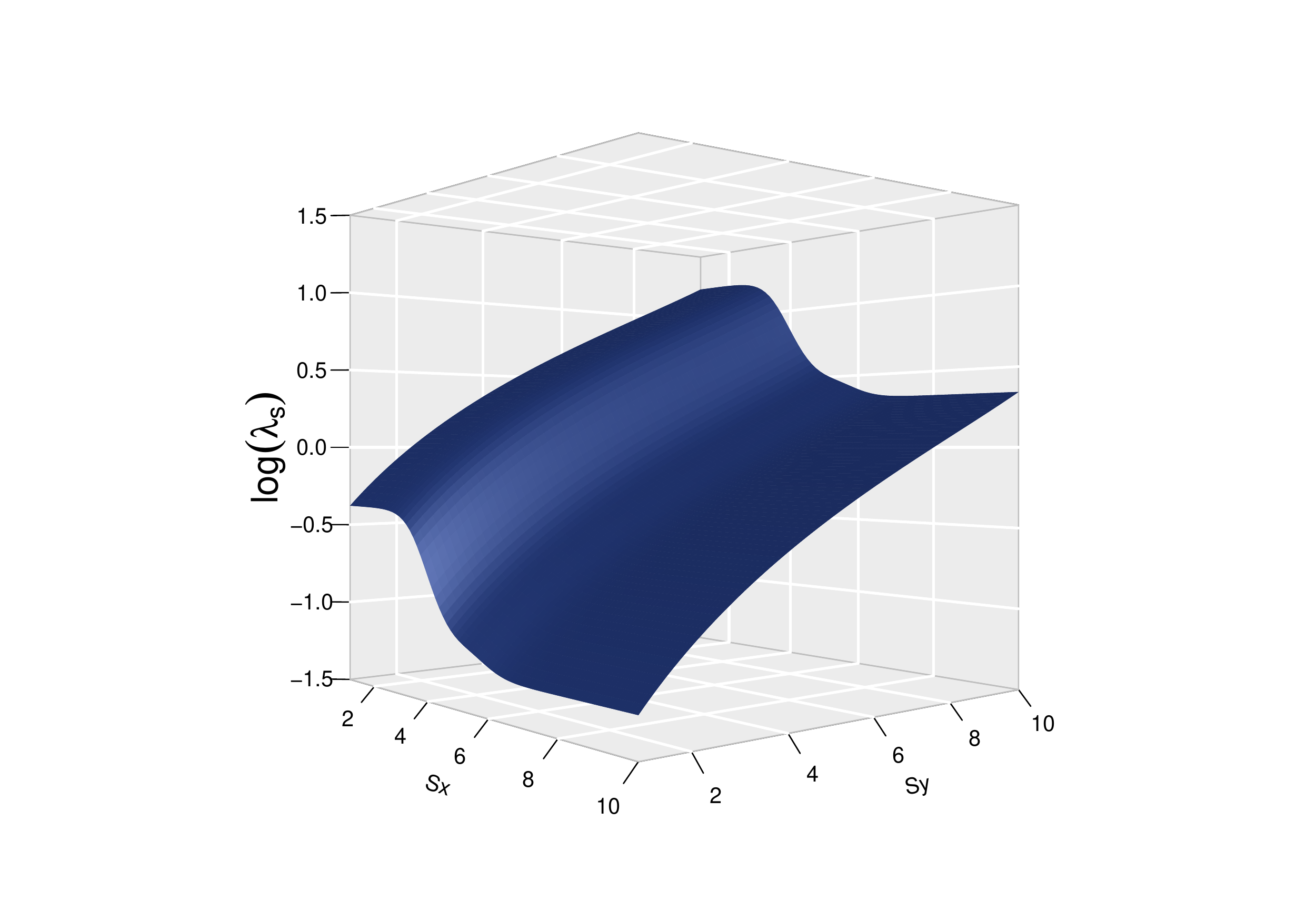}
                		\vspace{-.3cm}
           	\end{minipage}
		\hspace{-3.7cm}
           	\begin{minipage}[c]{0.5\linewidth}
			\begin{center}
				{\hspace{-2 cm}\footnotesize \textbf{Estimates (mild case)}}
			\end{center}
			\vspace{-0.5cm}
           		\includegraphics[scale=0.23]{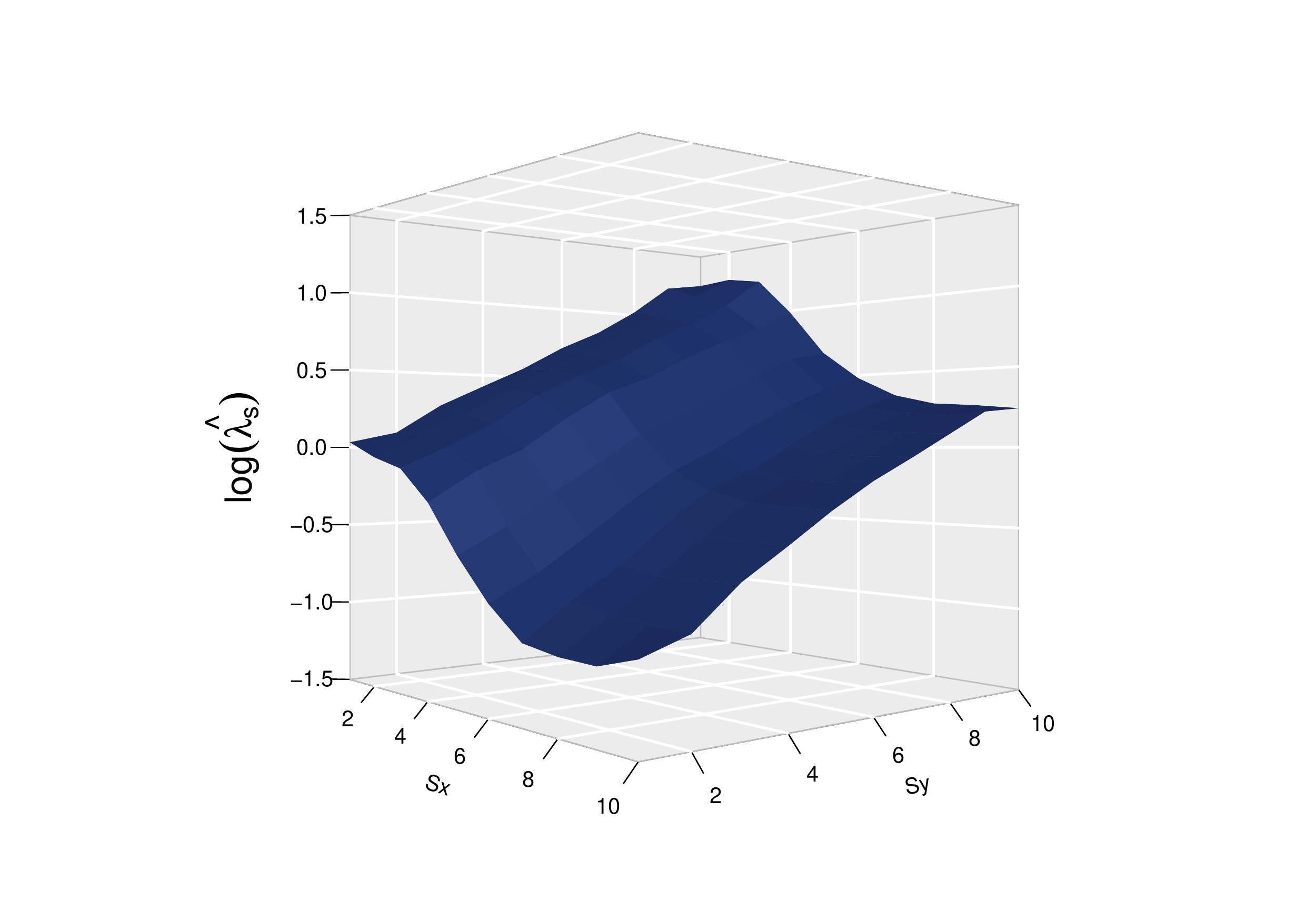}
                		\vspace{-.3cm}
           	\end{minipage}\\
		\vspace{0.5cm}
      		\footnotesize \rotatebox{90}{\textbf{\hspace{-1.2cm}{\color{white}{aaaa}}Range}}
		\hspace{-1cm}
      		\begin{minipage}[c]{0.5\linewidth}
        			\includegraphics[scale=0.23]{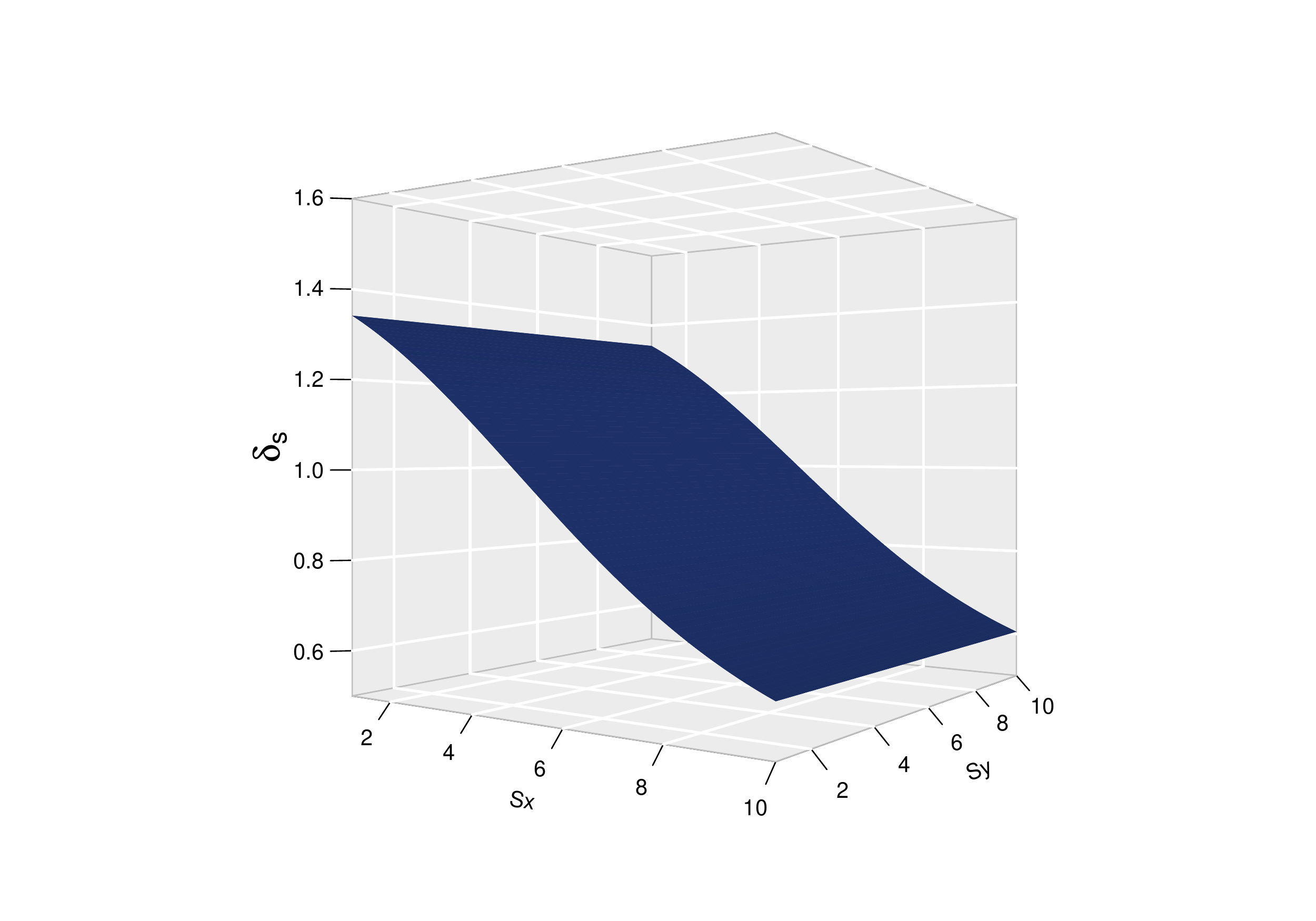}
        			\vspace{-.3cm}
      		\end{minipage}&
          	\hspace{-5cm}
          	\begin{minipage}[c]{0.5\linewidth}
          		\includegraphics[scale=0.23]{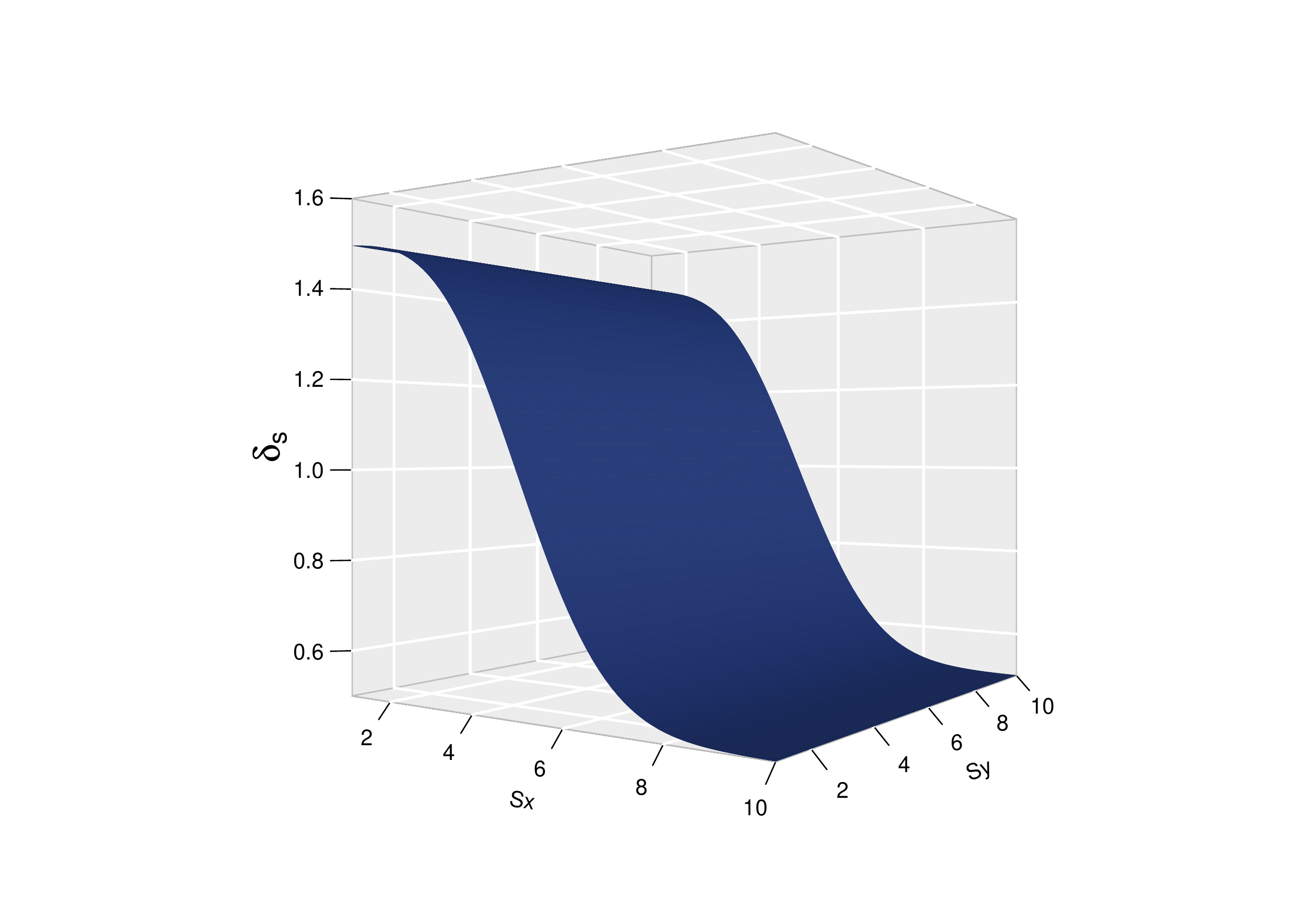}
                		\vspace{-.3cm}
                \end{minipage}
                \vspace{-.3cm}
                \hspace{-4.6cm}
                \begin{minipage}[c]{0.5\linewidth}
                		\includegraphics[scale=0.23]{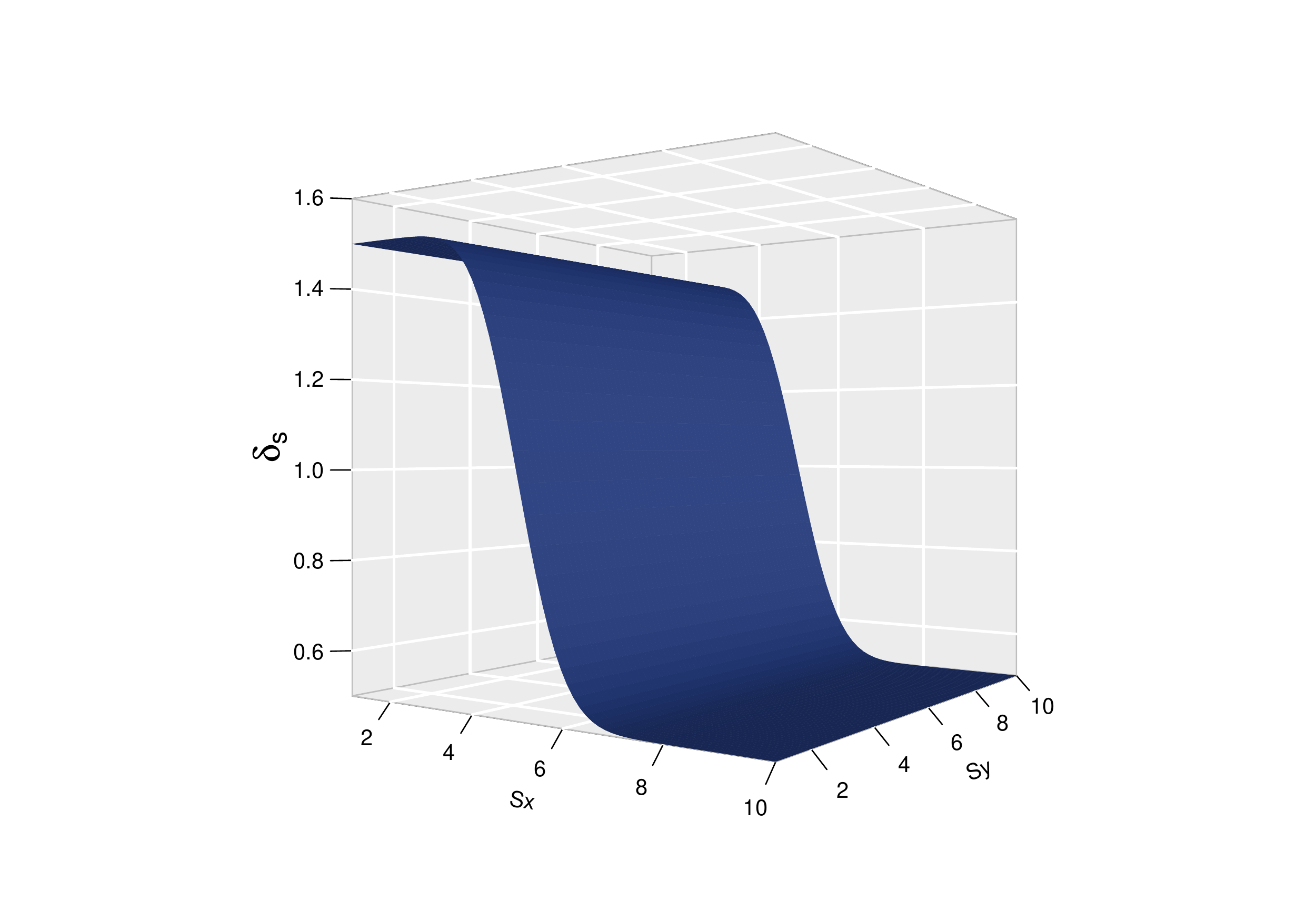}
                     	\vspace{-.3cm}
                \end{minipage}
                \hspace{-3.7cm}
                \begin{minipage}[c]{0.5\linewidth}
                		\includegraphics[scale=0.23]{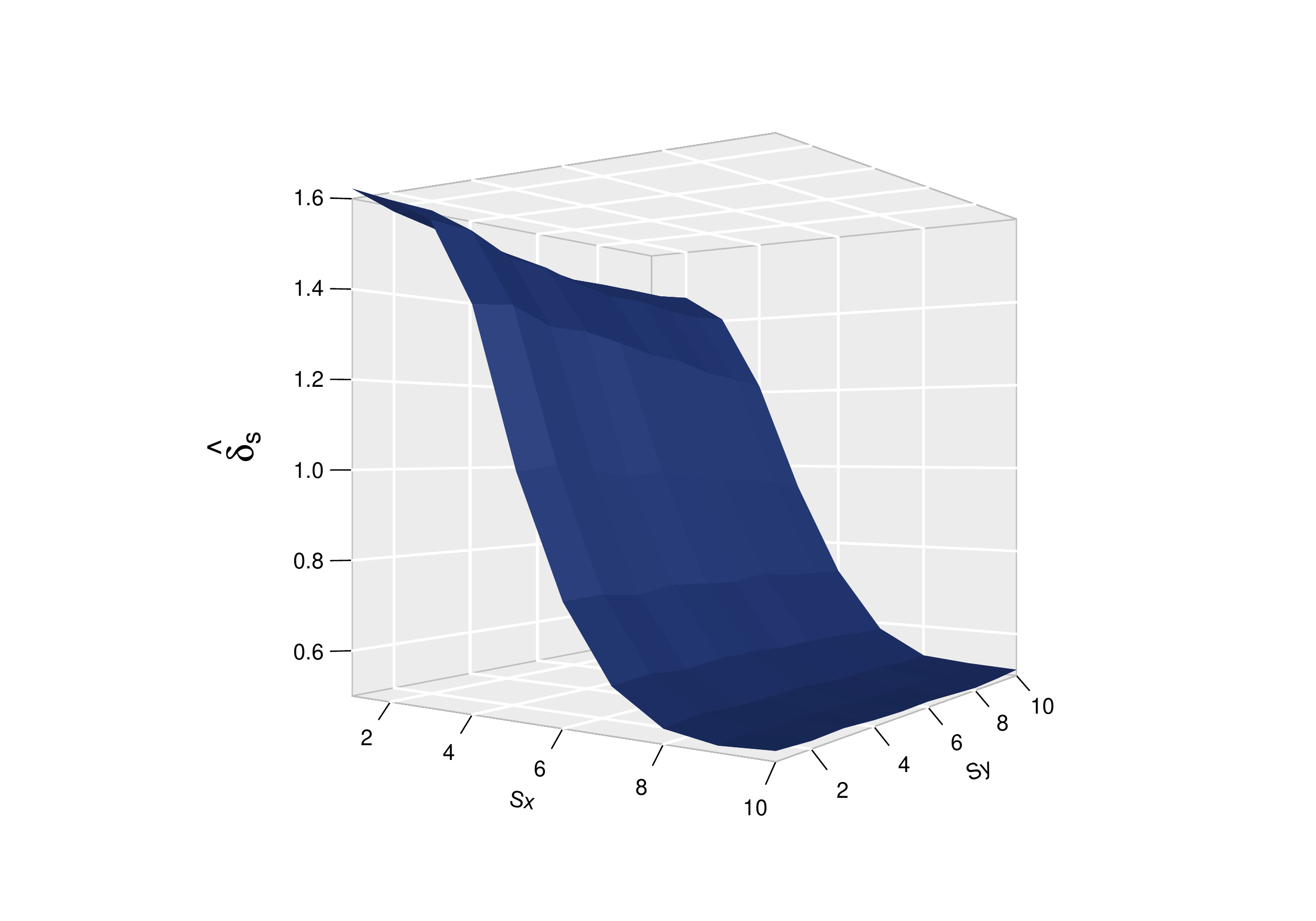}
                     	\vspace{-.3cm}
                \end{minipage}\\ 
                \hspace{-.7cm}
      		\begin{minipage}[c]{0.5\linewidth}
        			\includegraphics[scale=0.23]{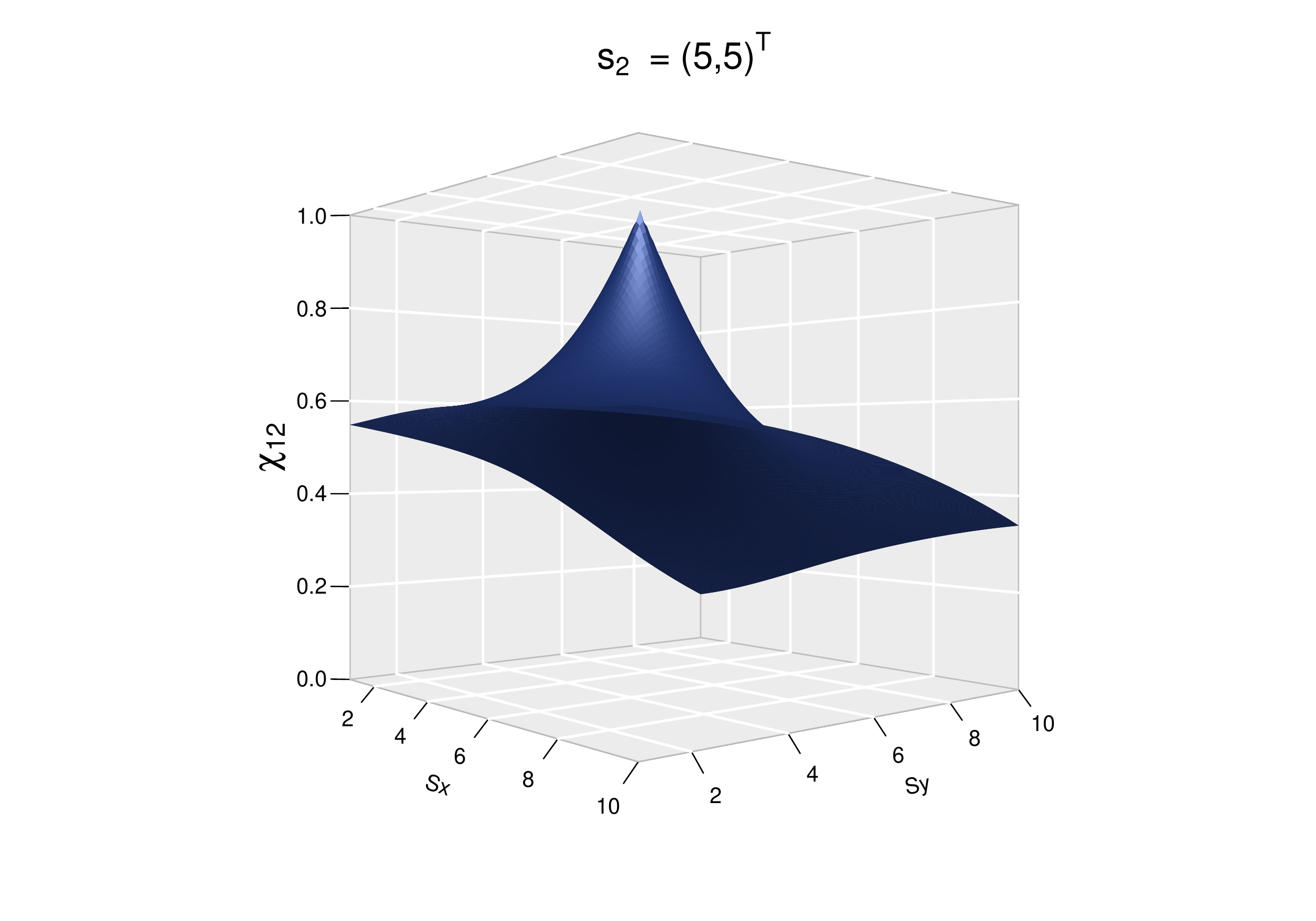}
        			\vspace{-.3cm}
      		\end{minipage}&
                \hspace{-5cm}
                \begin{minipage}[c]{0.5\linewidth}
                		\includegraphics[scale=0.23]{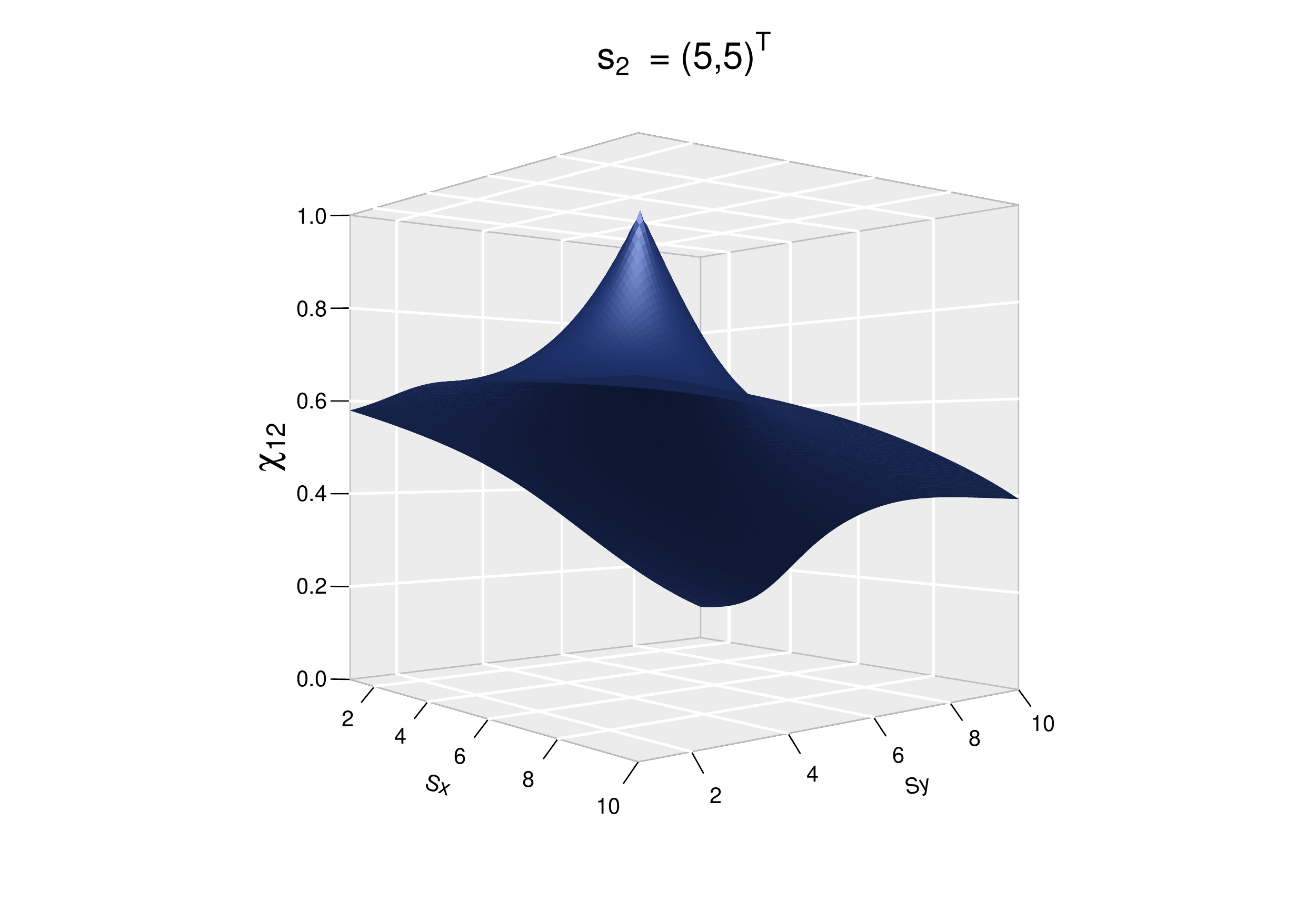}
                    	\vspace{-.3cm}
                \end{minipage}
                \vspace{-.3 cm}
                \hspace{-4.6cm}
                \begin{minipage}[c]{0.5\linewidth}
                     	\includegraphics[scale=0.23]{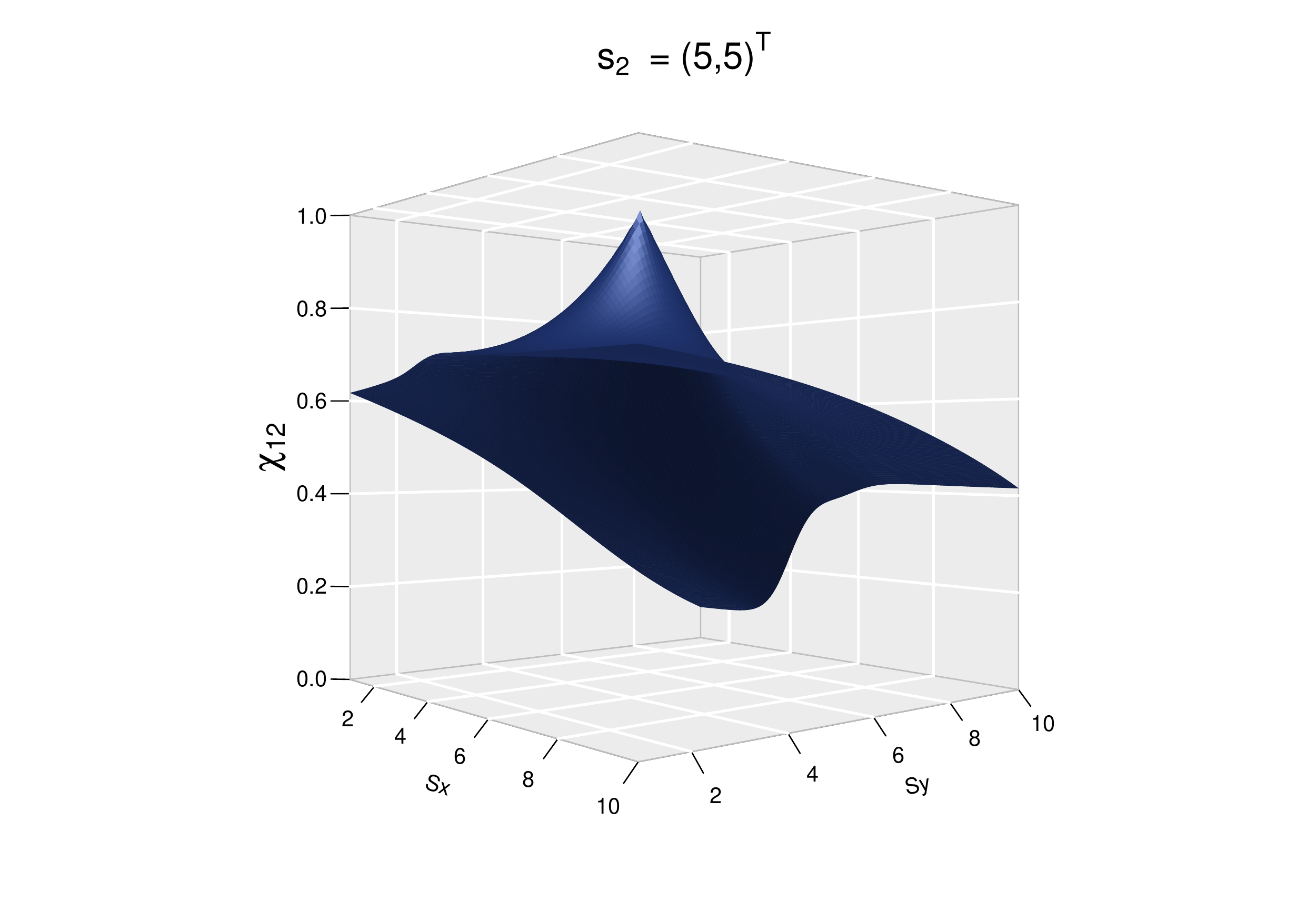}
                \vspace{-.3cm}
                \end{minipage}
                \hspace{-3.7cm}
                \begin{minipage}[c]{0.5\linewidth}
                     	\includegraphics[scale=0.23]{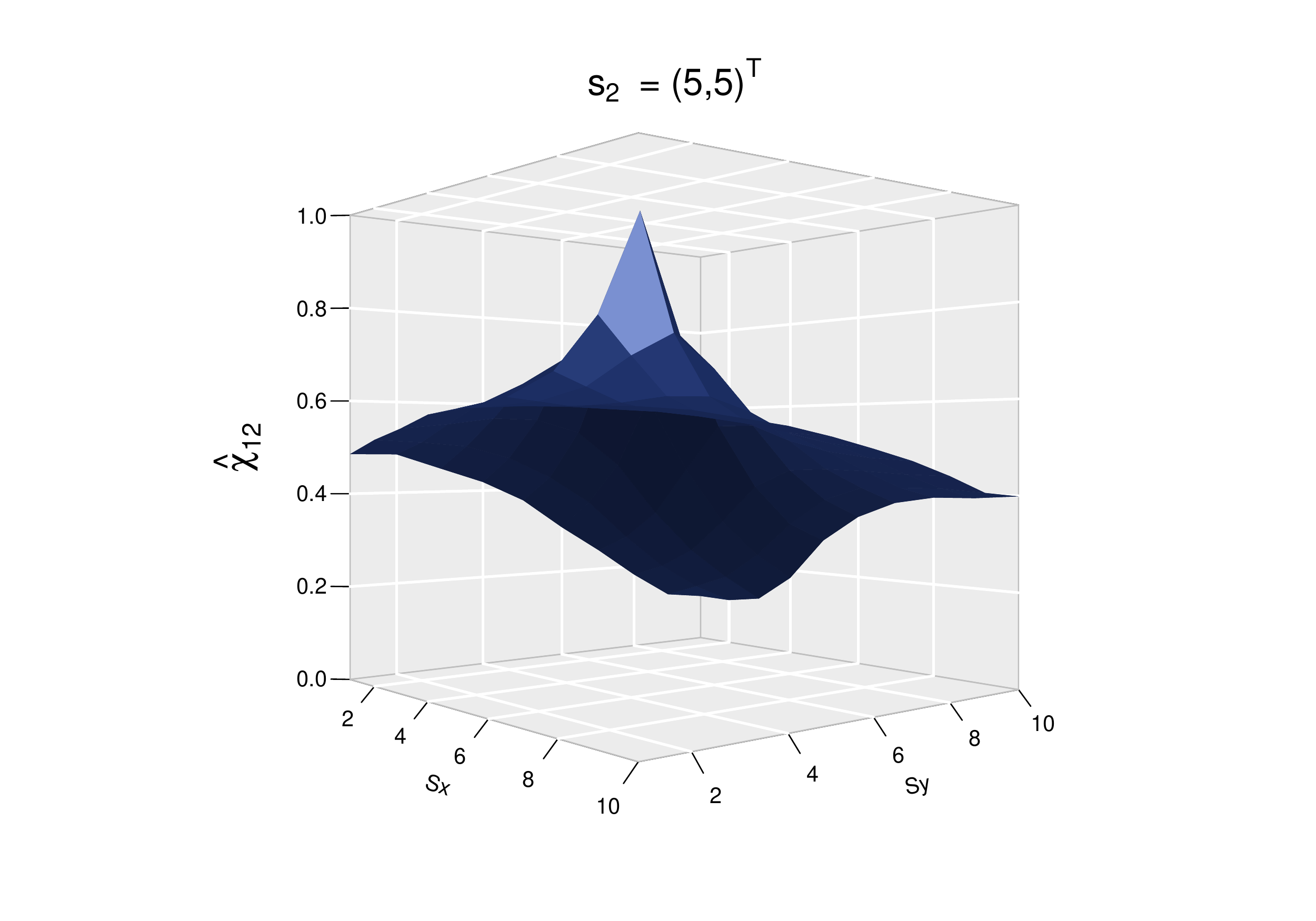}
                \vspace{-.35cm}
                \end{minipage}\\
                \footnotesize \rotatebox{90}{\textbf{\hspace{-1.7cm}Bivariate $\chi_{12}$ measure}}
                \hspace{-1cm}
      		\begin{minipage}[c]{0.5\linewidth}
        			\includegraphics[scale=0.23]{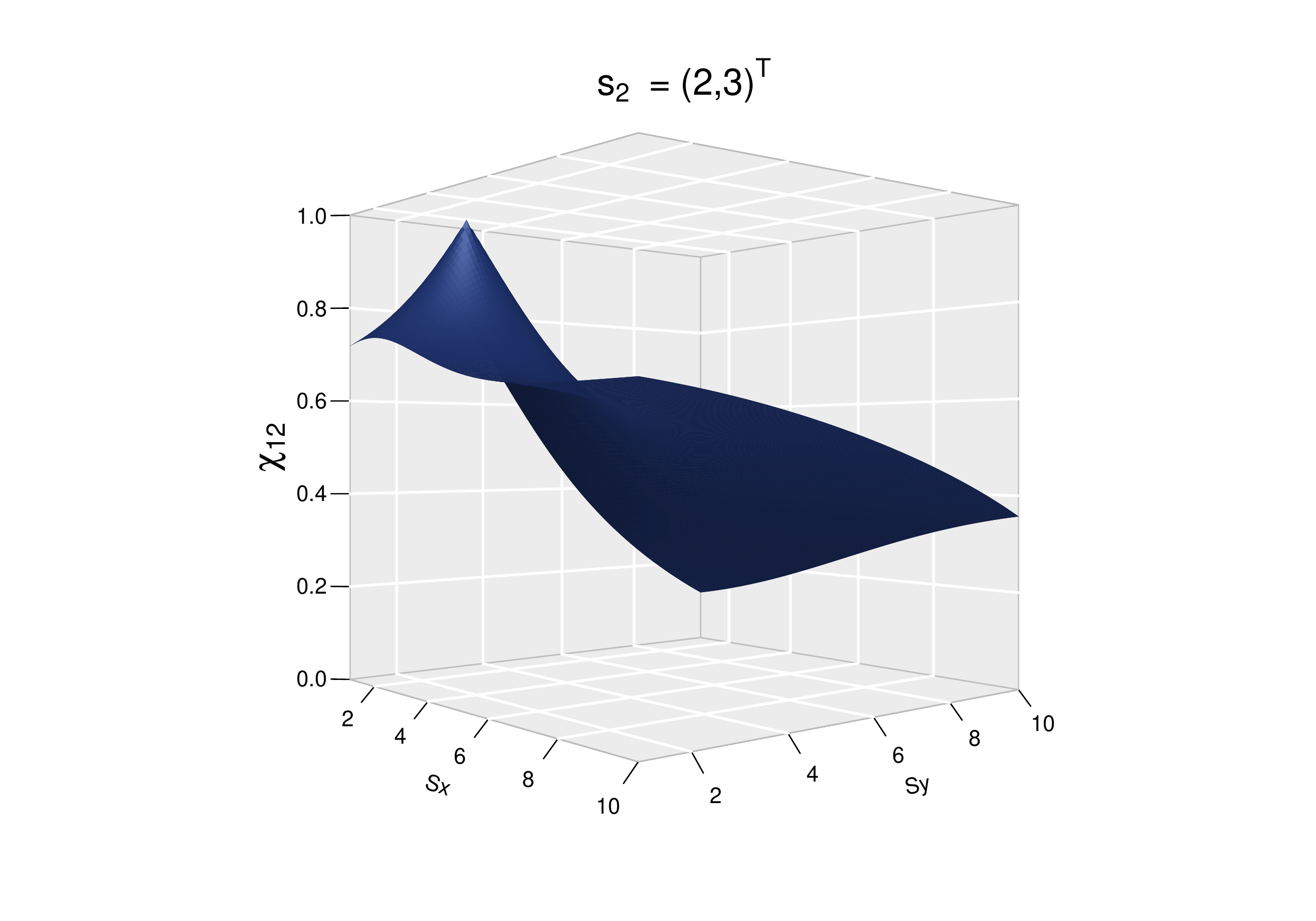}
      		\end{minipage}&
                 \hspace{-5cm}
                 \begin{minipage}[c]{0.5\linewidth}
                     	\includegraphics[scale=0.23]{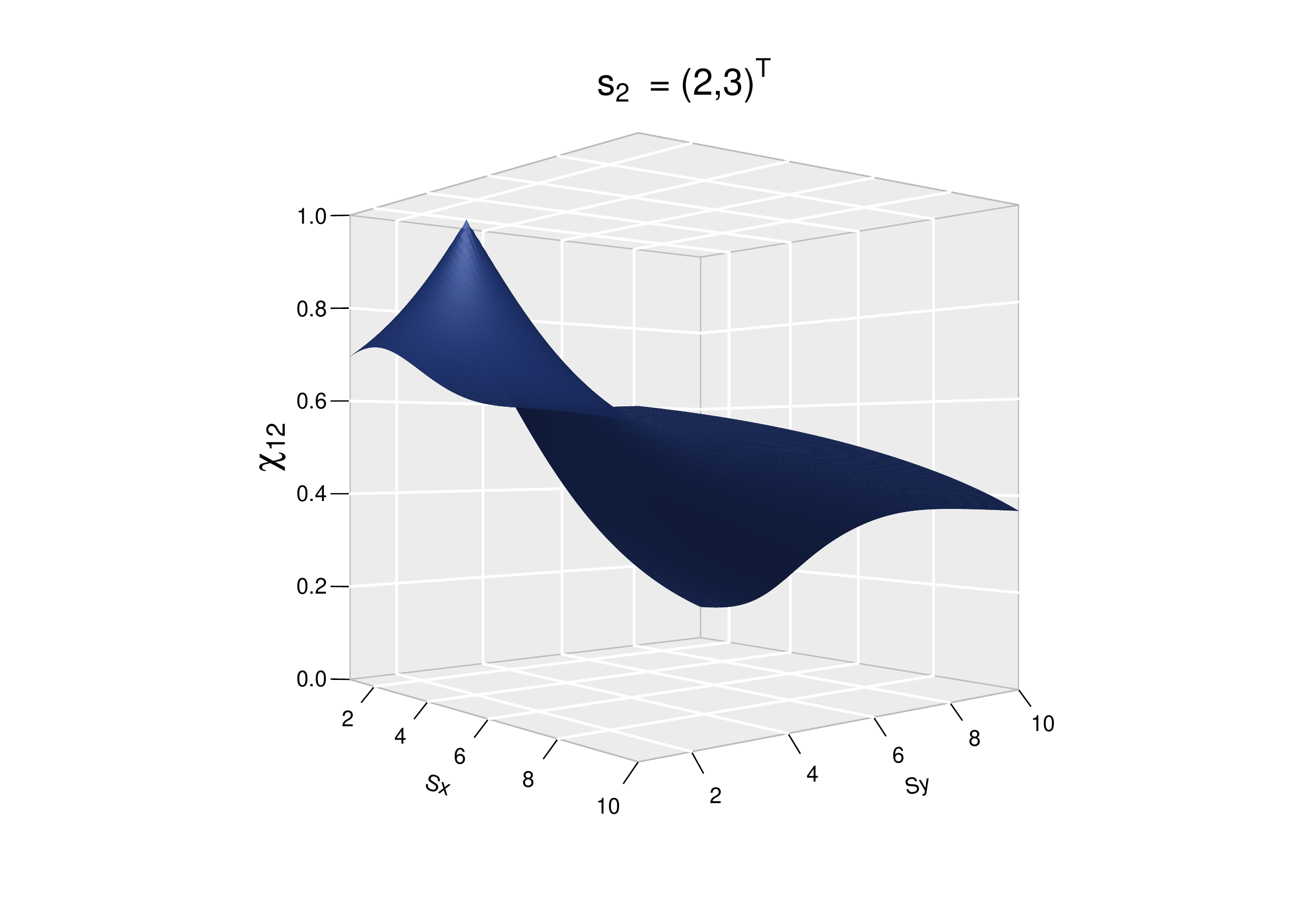}
                 \end{minipage}
                 \vspace{-.5 cm}
                 \hspace{-4.6cm}
                 \begin{minipage}[c]{0.5\linewidth}
                     	\includegraphics[scale=0.23]{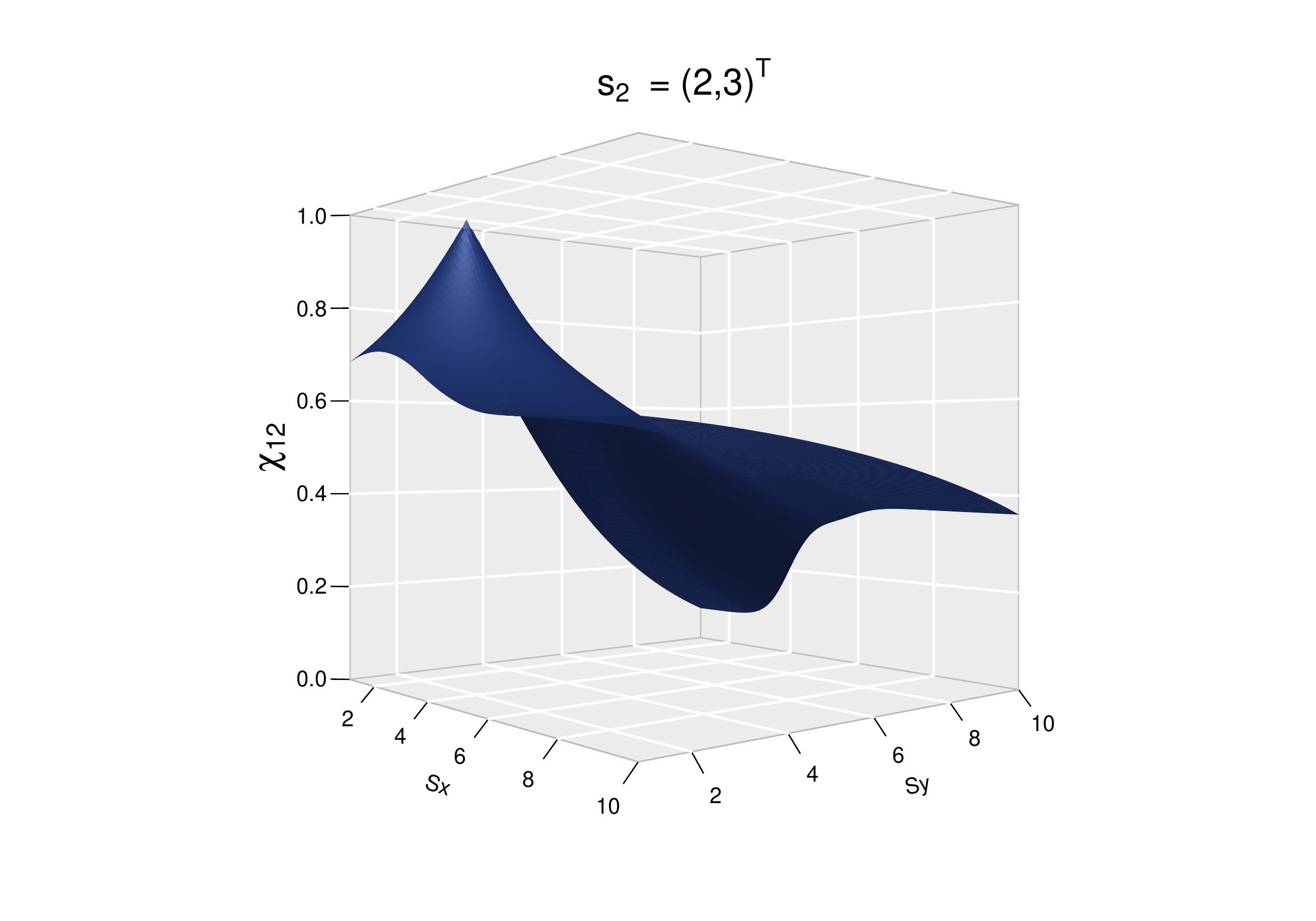}
                 \end{minipage} 
                 \hspace{-3.7cm}
                 \begin{minipage}[c]{0.5\linewidth}
                     	\includegraphics[scale=0.23]{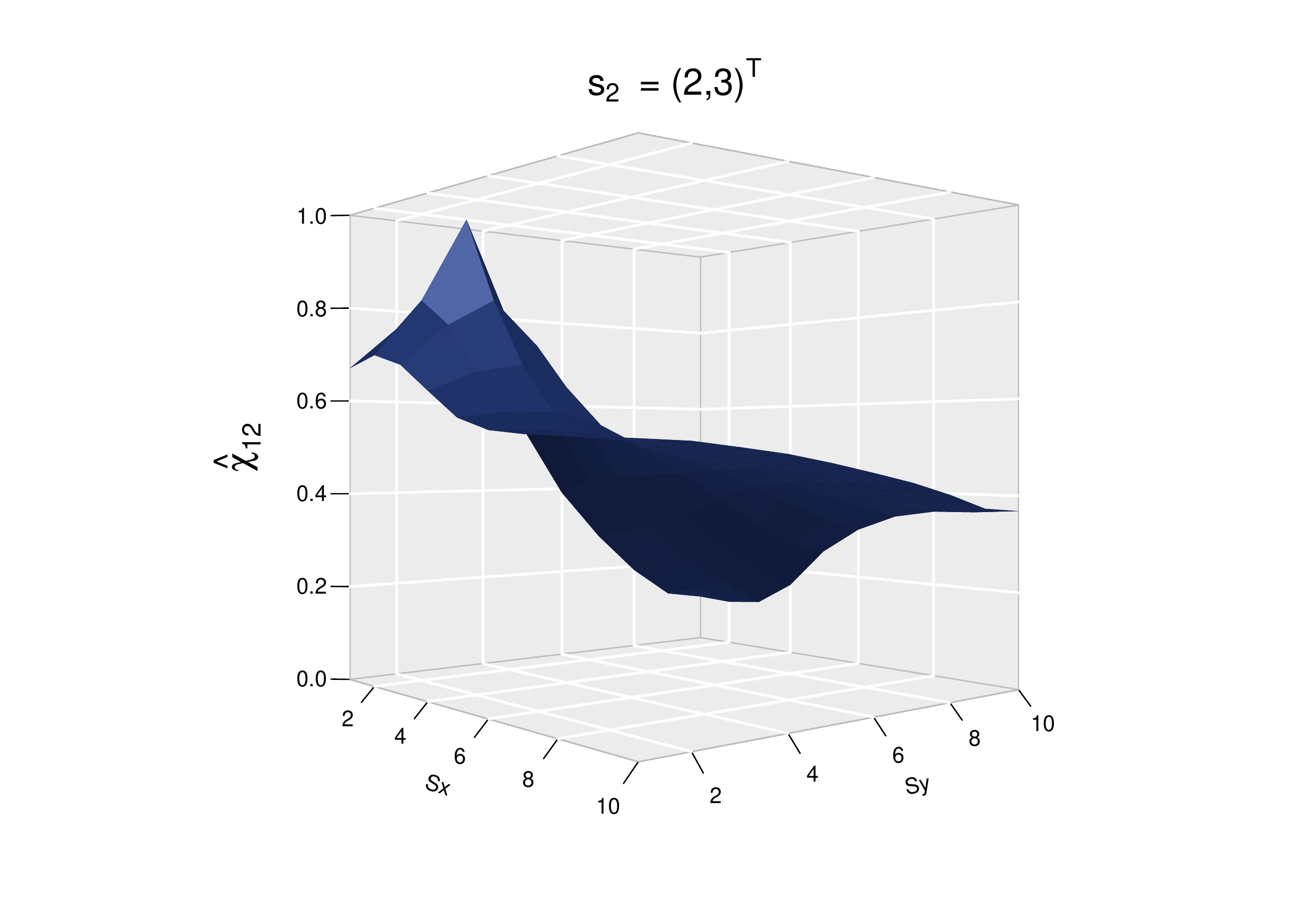}
                 \end{minipage}\\
                 \hspace{-.7cm}
      		\begin{minipage}[c]{0.5\linewidth}
        			\includegraphics[scale=0.23]{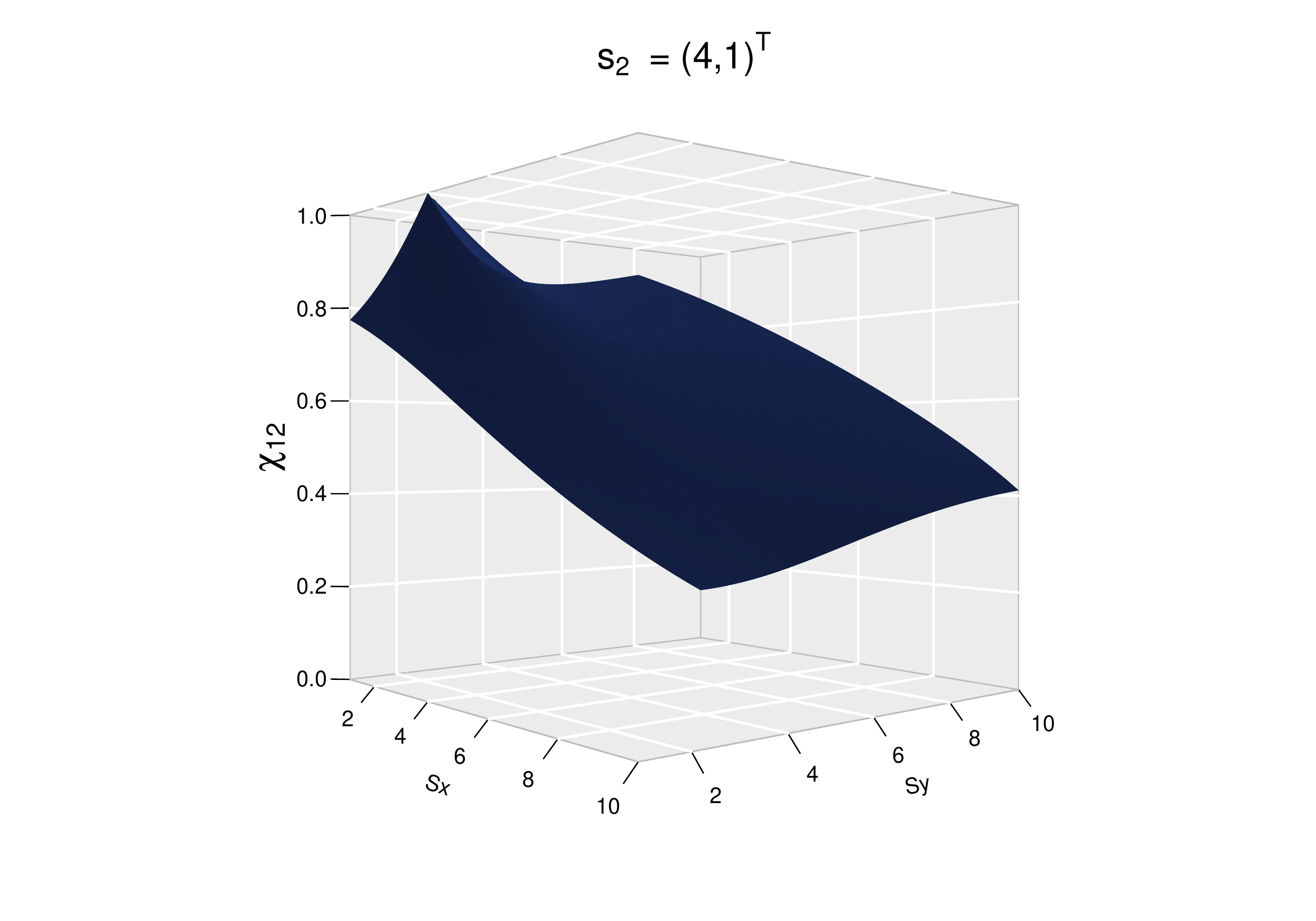}
        			\vspace{-1.5cm}
      		\end{minipage}&
                 \hspace{-5cm}
                 \begin{minipage}[c]{0.5\linewidth}
                     	\includegraphics[scale=0.23]{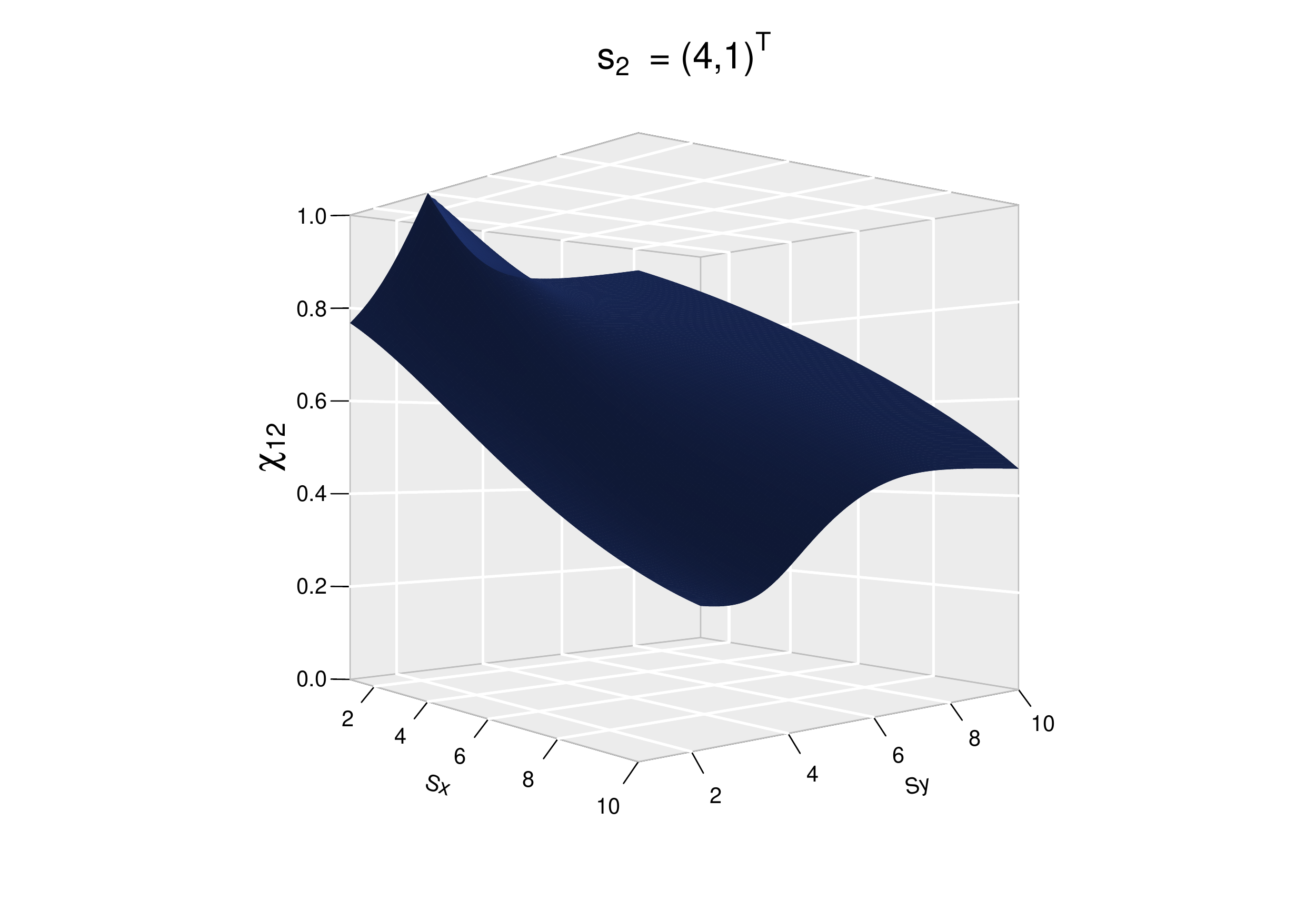}
                       	\vspace{-1.5cm}
                 \end{minipage}
                 \vspace{-.3 cm}
                 \hspace{-4.6cm}
                 \begin{minipage}[c]{0.5\linewidth}
                     	\includegraphics[scale=0.23]{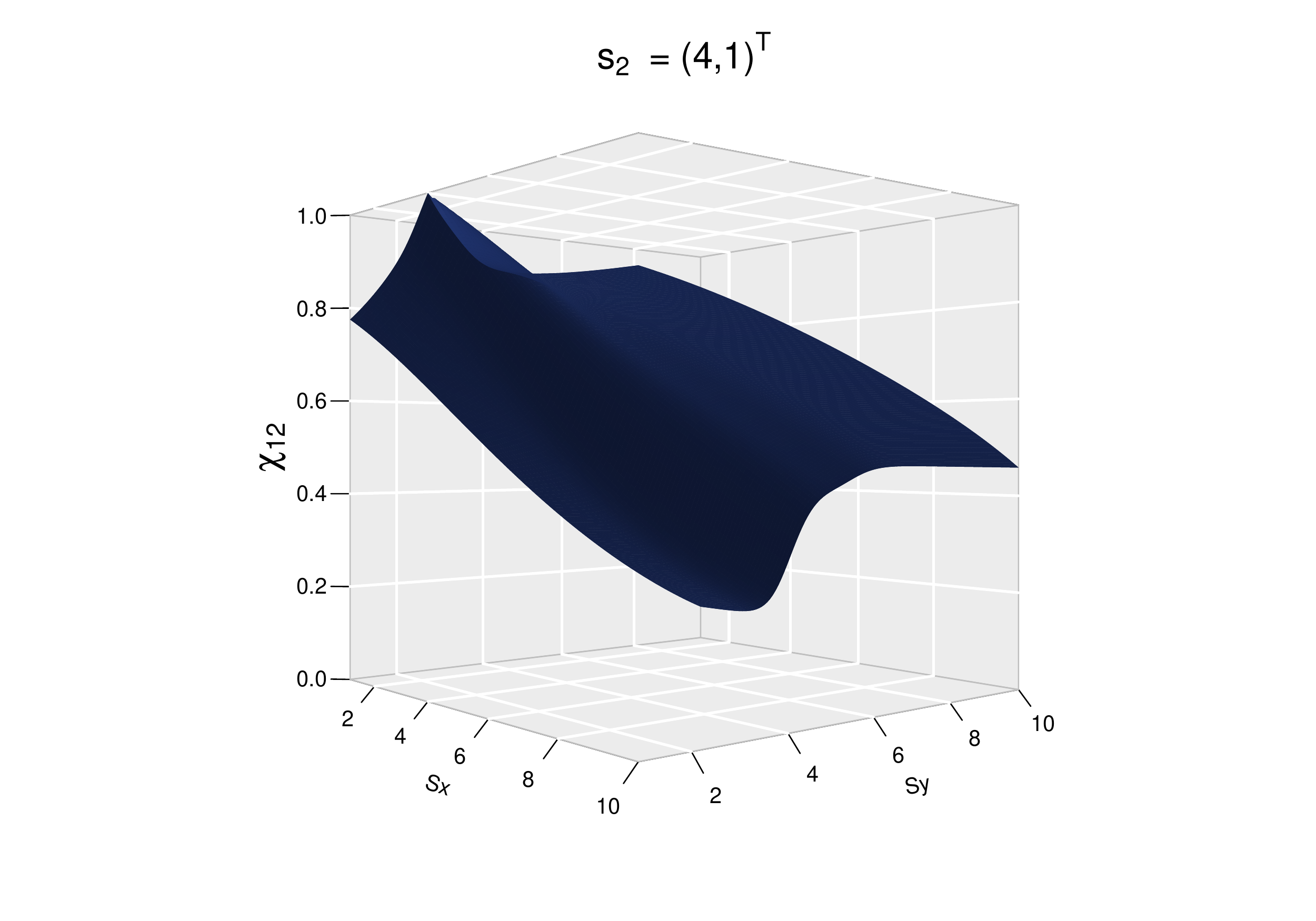}
                        	\vspace{-1.5cm}
                 \end{minipage}
                 \hspace{-3.7cm}
                 \begin{minipage}[c]{0.5\linewidth}
                     	\includegraphics[scale=0.23]{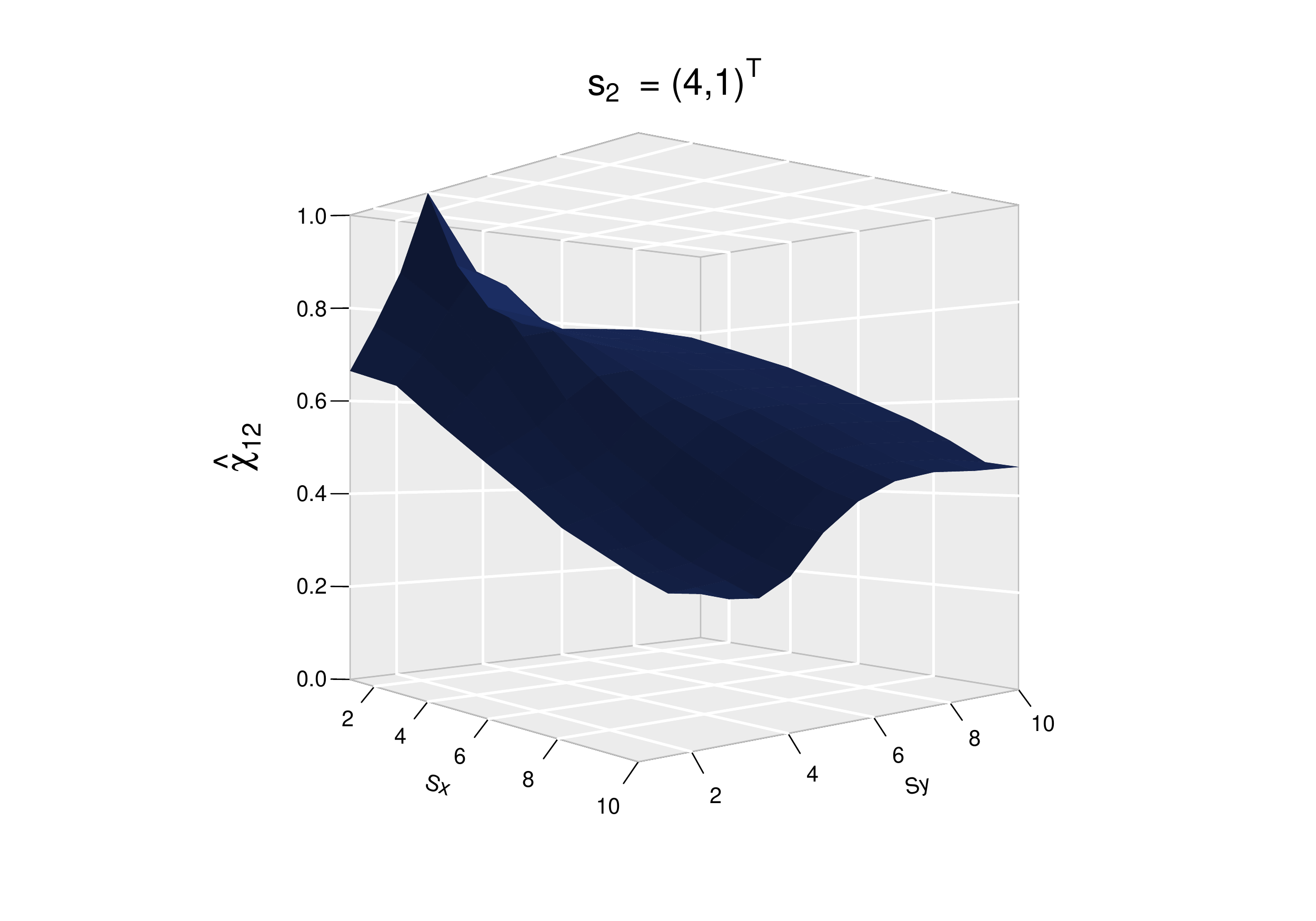}
                        	\vspace{-1.5cm}
                 \end{minipage}
    	\end{tabular}
	\vspace{1.3cm}
	\caption{\footnotesize \emph{Columns 1--3:} True spatially-varying log-rate $\log(\lambda_{\mathbf s})$ ($1^{\rm st}$ row), range $\delta_{\mathbf s}$ ($2^{\rm nd}$ row), and {bivariate $\chi_{12}$ measure} ($3^{\rm rd}$--$5^{\rm th}$ rows) for a fixed location $\mathbf s_2$. Columns 1--3 correspond to different levels of non-stationarity. \emph{Column 4:} Mean surface of estimated parameters for the mildly stationary scenario, based on $1000$ simulation experiments, using the local censored likelihood approach with thresholds $u_j^\star=0.95$ for all $j$, and $D_0=20$ nearest neighbors, as detailed in Section~\ref{sec:Inference}. The smoothness parameter is fixed to $\nu=2.5$.}    
	\label{SimSce.pdf}
\end{figure}

\subsection{Simulation results}\label{Data-generating}

We now detail the results for the mildly non-stationary scenario (second column of Figure~\ref{SimSce.pdf}) with $\nu=2.5$ and choosing $D_0=20$ nearest neighbors. Results for all simulation scenarios are summarized in Table~\ref{tableRMISE}. The rightmost column of Figure~\ref{SimSce.pdf} displays mean surfaces, computed over the $1000$ independent experiments, of the estimated log-rate $\log(\lambda_{\mathbf s})$ and range $\delta_{\mathbf s}$ parameters, as well as the fitted {bivariate $\chi_{12}$ measure} for the three different reference locations. By comparing the true surfaces to the mean estimates, the bias appears to be quite small overall, except perhaps at the boundaries of the study region, $\mathcal S$. This is a well-known drawback of local estimation approaches: because the neighborhoods are asymmetric near the boundaries, the bias is generally more severe, and this can also be noticed in our case, e.g., in regions where {$\chi_{12}\approx 0$}. This issue is similar to the boundary problem in kernel density estimation occurring with positively or compactly supported densities; some strategies have been advocated to deal with it, such as taking asymmetric kernels near the boundaries, or performing local linear regression (see, e.g.,~\citeauthor{castro2018time}, \citeyear{castro2018time}), but we do not pursue this here. Apart from this minor boundary problem, our local estimation approach succeeds in capturing the complex non-stationary dependence dynamics over space.

 \begin{table}[t!]
 \begin{center}
 \caption{ 
  \label{tableRMISE}\footnotesize RMISEs for the rate $\hat{\lambda}_{\mathbf{s}}$ and range $\hat{\delta}_{\mathbf{s}}$ parameters for each smoothness parameter $\nu=0.5,1.5,2.5$, computed over 1000 replicates from the data-generating configurations discussed in Section~\ref{Data-generating}.}
 \vspace{5pt}   
\begin{tabular}{l|*{2}{c}| *{2}{c}|*{2}{c}}
&\multicolumn{2}{c|}{$\nu = 0.5$}&\multicolumn{2}{c|}{$\nu = 1.5$} &\multicolumn{2}{c}{$\nu = 2.5$} \\\cline{2-7}
Configuration& $\hat{\lambda}_{\mathbf{s}}$&$\hat{\delta}_{\mathbf{s}}$ & $\hat{\lambda}_{\mathbf{s}}$&$\hat{\delta}_{\mathbf{s}}$ & $\hat{\lambda}_{\mathbf{s}}$&$\hat{\delta}_{\mathbf{s}}$\\\hline
  Weakly non-stationary &0.09 &0.59 &0.03 &0.01 &  0.03 & 0.00 \\
  Mildly non-stationary &0.09 &0.69 &0.06 &0.02 & 0.04&0.00 \\
  Strongly non-stationary & 0.17& 0.89& 0.15& 0.05& 0.22&0.02
\end{tabular}
\end{center}
\end{table}

To assess the variability of the parameter estimates, Figure~\ref{MeanSurf.pdf} displays a functional boxplot for the range parameter $\delta_{\mathbf{s}}$, projected onto the $x$-axis representing the only direction of variation in the true $\delta_{\mathbf{s}}$ values, and a surface boxplot for the {bivariate $\chi_{12}$ measure} for fixed $\mathbf{s}_2=(2,3)^T$. Functional and surface boxplots are the natural extensions of the classical boxplot to the case of functional data, and we refer to \citet{sun2012functional} and \citet{genton2014surface} for their precise interpretation. 
\begin{figure}[t!]
\centering
\begin{minipage}[c]{0.49\linewidth}
\centering
\includegraphics[scale=0.34]{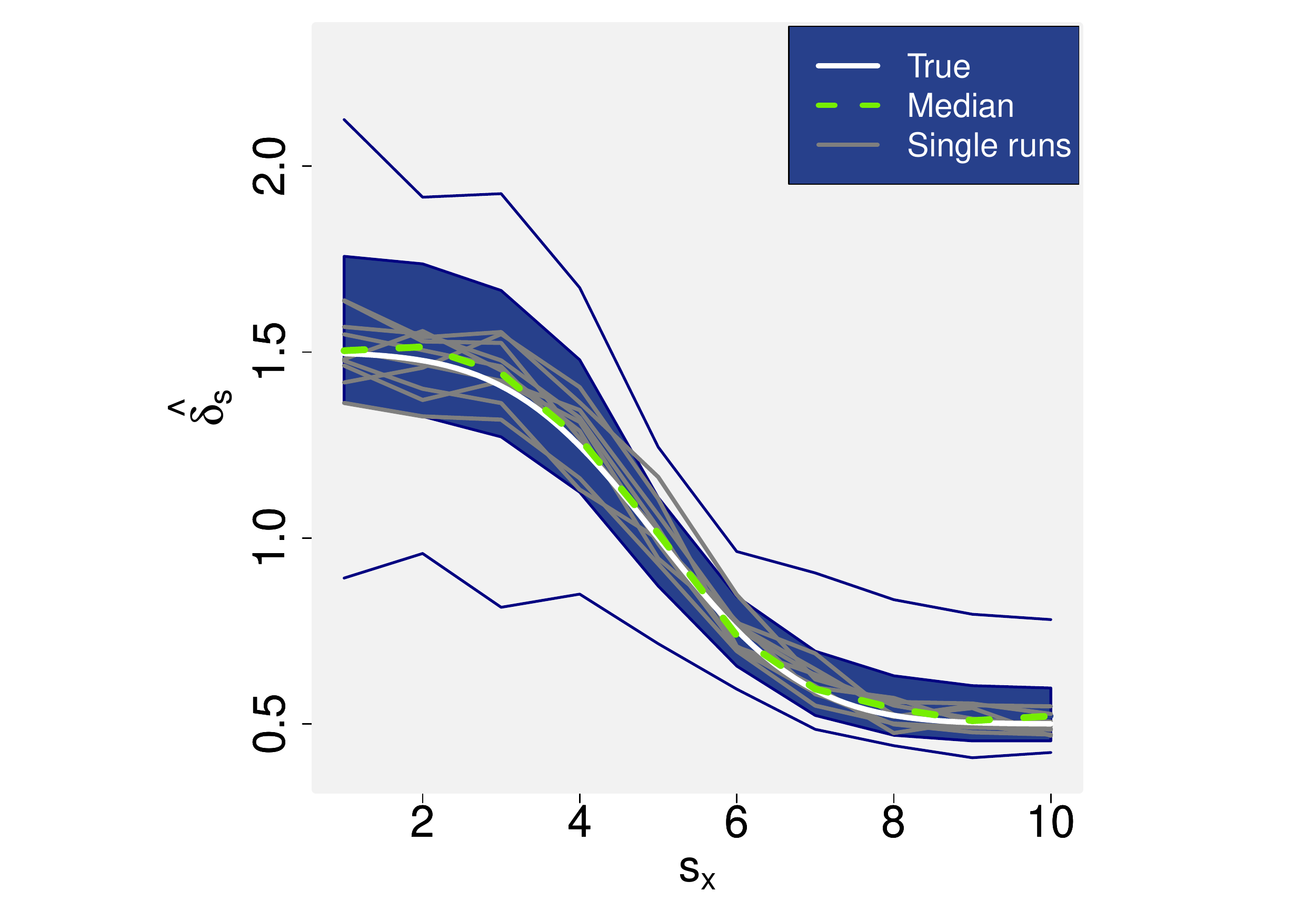}
\end{minipage}
\begin{minipage}[c]{0.49\linewidth}
\centering
\includegraphics[scale=0.4]{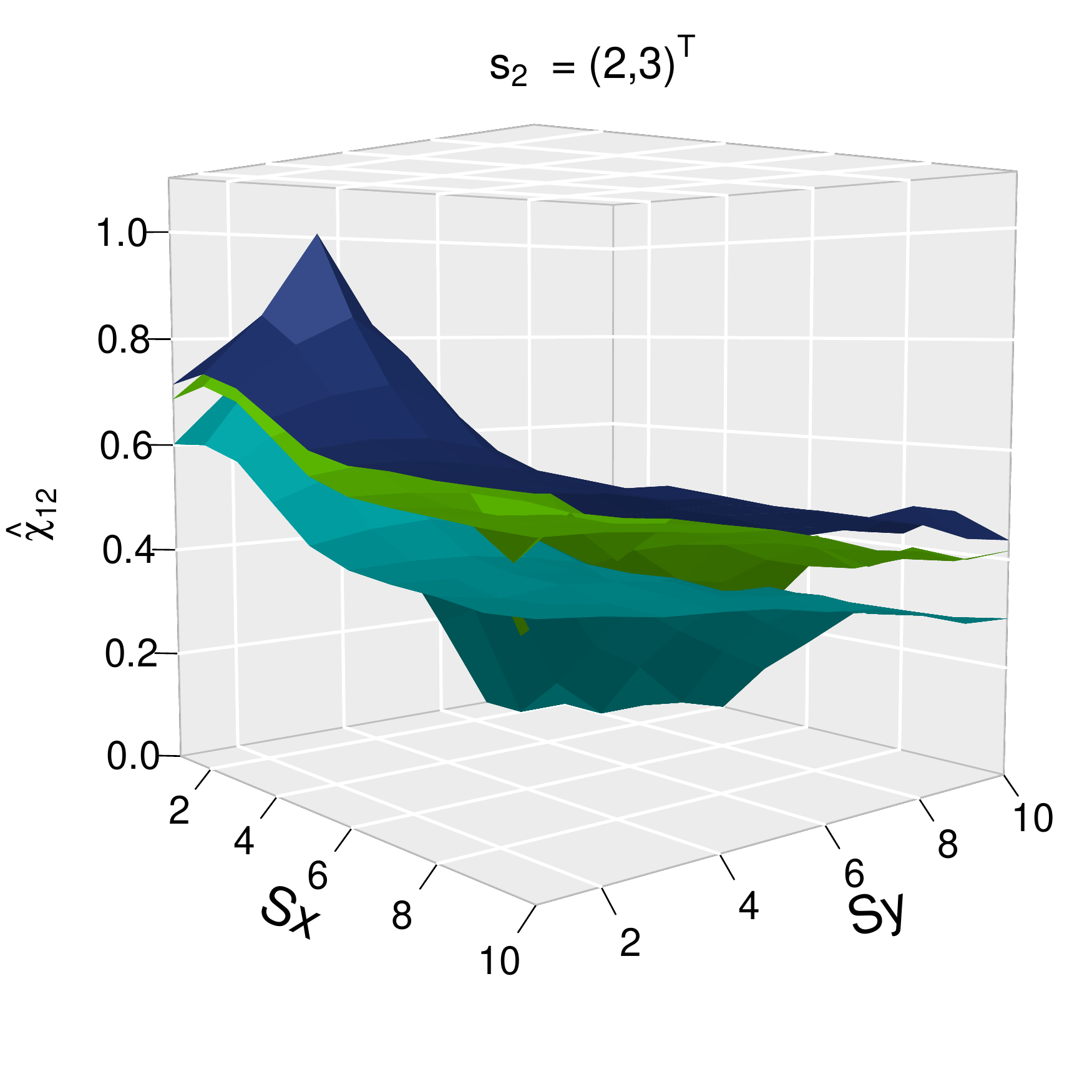}
\end{minipage}
\caption{\footnotesize Functional boxplot (left) for the range $\delta_{\mathbf{s}}$, $\mathbf{s}=(s_x,s_y)^T$, plotted with respect to $x$-coordinate $s_x$, and surface boxplot (right) for the {bivariate $\chi_{12}$ measure} plotted as a function of location $\mathbf{s}_1$ for fixed $\mathbf{s}_2=(2,3)^T$. In the surface boxplot, the dark and light blue surfaces represent the $1^{\rm st}$ and $3^{\rm rd}$ quartiles (i.e., the ``box'' of the boxplot), respectively, while the green surface is the median. The surface boxplot's ``whiskers'' are not displayed for better visualization. Both panels show the results for the mildly non-stationary case with $\nu=2.5$ and $D_0=20$ nearest neighbors in the local estimation approach.}    
\label{MeanSurf.pdf}
\end{figure}
As the true range parameter $\delta_{\mathbf{s}}$, $\mathbf{s}=(s_x,s_y)^T$, varies only with respect to $s_x$, it is easier to visualize its estimated uncertainty. As expected, $\delta_{\mathbf{s}}$ appears to be well estimated overall with higher uncertainty for larger $\delta_{\mathbf{s}}$ values. The estimated median curve follows the true curve very closely, even near the boundaries and around $s_x=1$, where the true curve is the steepest, while the functional inter-quartile range is fairly narrow for all values of $s_x$. {Moreover, the estimates for single runs, superimposed on the functional boxplot, suggest that the hard-thresholding approach produces reasonably smooth estimates, even with only about 25 threshold exceedances at each location. In our application in Section~\ref{DataApp}, we have {218} exceedances per location on average}. As for the {bivariate $\chi_{12}$ measure}, the surface boxplot suggests that it can also be rather well estimated with relatively low uncertainty.
Table~\ref{tableRMISE} reports the root mean integrated squared error (RMISE) for all the data generating configurations described in Section~\ref{ExpNonStatLocalModel}. The results are coherent with our intuition: the estimation is more difficult for higher levels of non-stationarity and rougher random fields (i.e., with smaller $\nu$). Overall, simulations confirm that our local censored estimation approach provides reasonable estimates of the true range and rate parameters while capturing their dynamics over space.

To assess the performance of our local likelihood approach as a function of the neighborhood size {$D_0$}, we fixed the rate parameter $\lambda_{\mathbf{s}}$ to the strongly non-stationary scenario, while considering the weakly, mildly and strongly non-stationary cases for the range parameter $\delta_{\mathbf{s}}$ (recall Figure~\ref{SimSce.pdf}). We then fitted the copula model using the local likelihood approach described in Section~\ref{sec:Inference} with neighborhoods $\mathcal N_{\mathbf{s}_0;D_0}$ defined in terms of $D_0=5,10,15,20,25$ nearest neighbors. Table~\ref{tabled0} reports the RMISE for all cases. While the RMISE is quite small and fairly constant overall for the range $\delta_{\mathbf{s}}$, the rate $\lambda_{\mathbf{s}}$ seems more difficult to estimate, and it improves with lower degrees of non-stationarity for $\delta_{\mathbf{s}}$ and bigger neighborhoods (i.e., with larger $D_0$), even in the strongly non-stationary scenario. This suggests that the size of neighborhoods will likely be dictated by available computational resources. Unless the non-stationarity is extremely severe, it is advisable to consider large neighborhoods, as this would improve the estimation efficiency at a fairly moderate cost in bias.

{Finally, using the mildly non-stationary scenario for the rate and range parameters and $\nu=0.5$, we tested an ad-hoc profile likelihood approach to select the smoothness parameter. We fitted our model by setting $\nu = \nu_{\text{test}}$ with $\nu_{\text{test}}\in\{0.5, 1.0, 1.5, 2.0, 2.5\}$, and counted the proportion of times that $\nu_{\text{test}}$ maximized the local log-likelihood across all grid points and the 1000 replications. Our proposed procedure is to select the ``best'' value for $\nu$ as the one with the largest proportion of maximized local likelihoods. Among all grid points and all replicates, the likelihood was here maximized in $37\%$ of the cases for $\nu_{\text{test}}=0.5$, $24\%$ of the cases for $\nu_{\text{test}}=1$, $21\%$ of the cases for $\nu_{\text{test}}=1.5$, and $18\%$ of the cases for $\nu_{\text{test}}=2.5$. {Therefore, our profile procedure seems a reasonable method to choose $\nu$ in a non-arbitrary way. It is important to notice that this is not a method to properly \emph{estimate} the smoothness parameter per se. Rather, the aim of the profile likelihood approach and the selection via frequencies is to select (or validate) $\nu$ in a non-arbitrary data-driven way. In general, estimation of the smoothness parameter is difficult; with a single realization in a fixed domain, not all the Mat\'ern correlation parameters can be estimated consistently (see~\citeauthor{zhang2004inconsistent},~\citeyear{zhang2004inconsistent}). Moreover, when $\nu$ is estimated jointly with $\delta$ and $\lambda$, we have found that the estimated parameters become very difficult to identify}. 
}

 \begin{table}[t!]
 \begin{center}
 \caption{ 
  \label{tabled0}\footnotesize RMISEs for the rate $\hat{\lambda}_{\mathbf{s}}$ and range $\hat{\delta}_{\mathbf{s}}$ parameters for $\nu=2.5$, as a function of the number $D_0$ of nearest neighbors used in the local estimation approach. The rate $\hat{\lambda}_{\mathbf{s}}$ was kept fixed to the strongly non-stationary case, while different degrees of non-stationarity are considered for the range $\delta_{\mathbf{s}}$.}
   \vspace{5pt}   
\begin{tabular}{l|*{2}{c}| *{2}{c}|*{2}{c}|*{2}{c}|*{2}{c}}
Configuration for the&\multicolumn{2}{c|}{$D_0 = 5$}&\multicolumn{2}{c|}{$D_0 = 10$} &\multicolumn{2}{c|}{$D_0 = 15$} &\multicolumn{2}{c|}{$D_0 = 20$}&\multicolumn{2}{c}{$D_0 = 25$} \\\cline{2-11}
range parameter $\delta_{\mathbf{s}}$& $\hat{\lambda}_{\mathbf{s}}$&$\delta_{\mathbf{s}}$ & $\hat{\lambda}_{\mathbf{s}}$&$\delta_{\mathbf{s}}$ &  $\hat{\lambda}_{\mathbf{s}}$&$\delta_{\mathbf{s}}$ &  $\hat{\lambda}_{\mathbf{s}}$&$\delta_{\mathbf{s}}$ & $\hat{\lambda}_{\mathbf{s}}$&$\delta_{\mathbf{s}}$\\\hline
  Weakly non-stationary &1.16 &0.02 &0.44 &0.01 &0.25 &0.00  &0.18 &0.00 & 0.15 & 0.00 \\
  Mildly non-stationary &1.33 &0.03 &0.61 &0.01 &0.44 &0.01 &0.34 &0.01 &0.30 & 0.01 \\
  Strongly non-stationary &1.38 &0.03 &0.69 &0.01 &0.55 &0.02 &0.47 &0.02 &0.45 &0.02
\end{tabular}
\end{center}
\end{table}


\section{Case study: U.S. winter precipitation extremes}\label{DataApp}
\subsection{Data description}\label{DataDesc}
Daily precipitation data, freely available online, were gathered from the U.S. Historical Climatological Network (USHCN). They are measured in hundredths of an inch and were collected from $1900$ to $2014$ at the $1218$ stations represented in Figure~\ref{fig:usmap.pdf}. To ensure data quality, we discarded all observations marked by any reliability or accuracy flag. We focus on winter data (December 21 to March 20) to remove seasonal effects, and we consider five day-cumulative precipitation, in order to capture the intensity and duration of distinct storms and to reduce the effect of temporal dependence. A preprocessing stage was conducted to check for possible temporal non-stationarity, and no evidence of non-stationarity was found (see {Section 3.1} {of} the Supplementary Material).
This procedure yields $2070$ time replicates per station with $19.8\%$ of missing data overall, which corresponds to about 2 million observations in total at all sites. The resulting dataset ranges from $0$ (no rain over five days) to $3006$ hundredths of an inch. 
\begin{figure}[t!]
\centering
\includegraphics[width=0.48\linewidth]{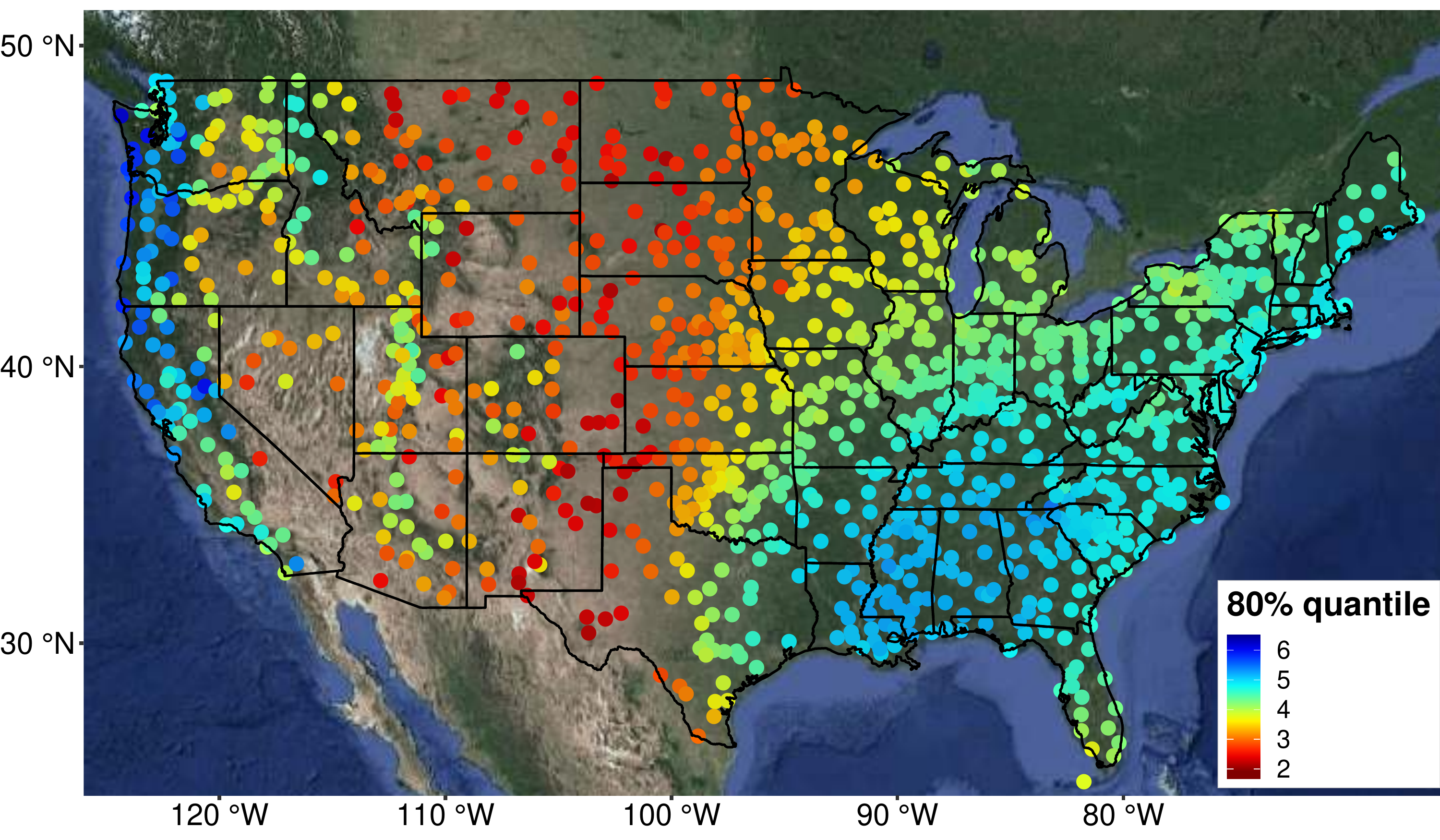}
\includegraphics[width=0.48\linewidth]{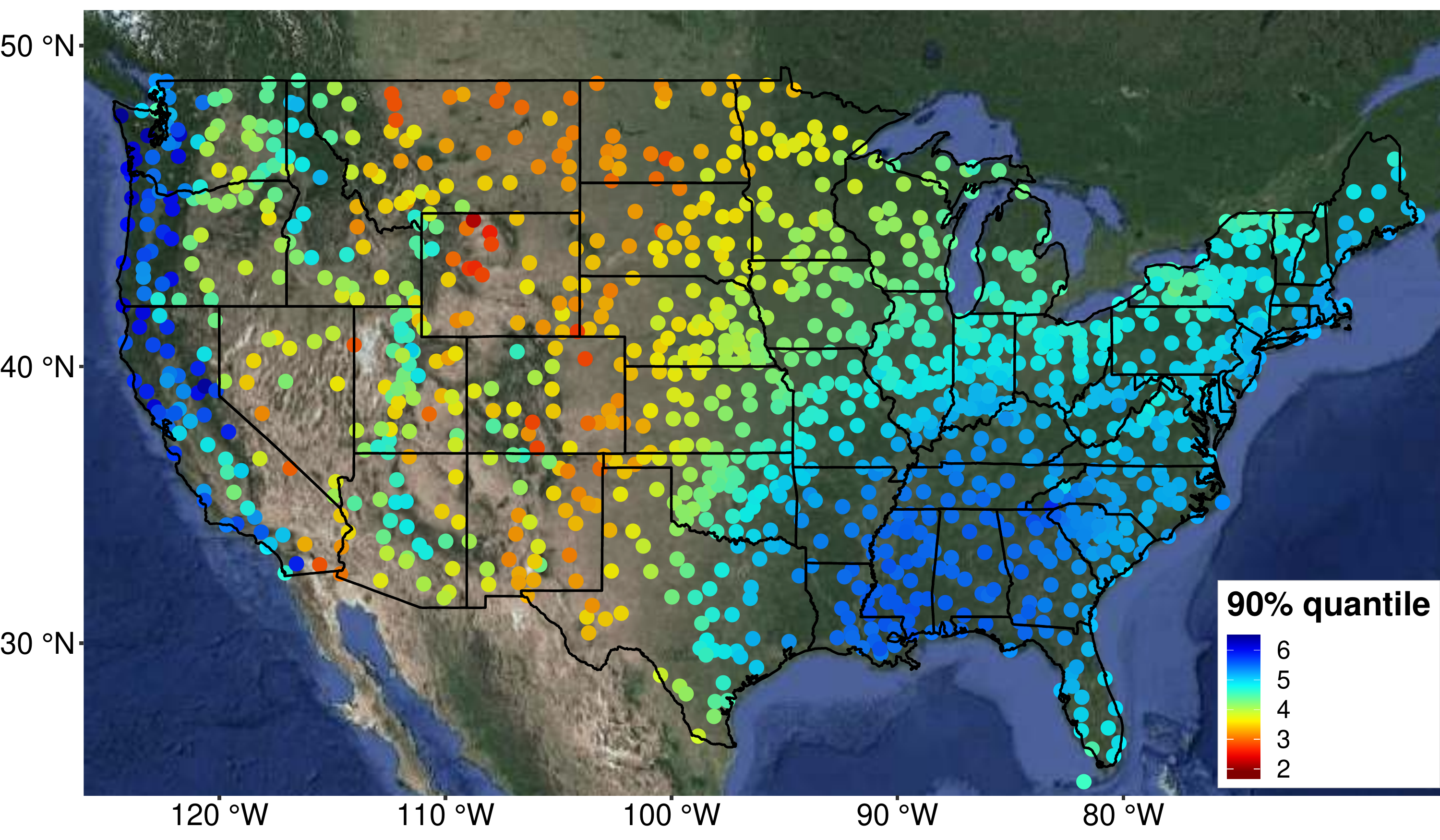}
\caption{\footnotesize Empirical {$80\%$ (left) and $90\%$} (right) quantiles of five day-cumulative winter precipitation data, observed at each of the monitoring stations, plotted on the same logarithmic color scale (units are the logarithm of hundredths of an inch).}    
\label{fig:quantiles}
\end{figure}
The empirical {$80\%$ and $90\%$} quantiles, plotted in Figure~\ref{fig:quantiles} using a logarithmic scale, reveal interesting spatial patterns that are due to marginal distributions varying smoothly over space. In order to disentangle the local marginal and dependence effects, we use the (two-step) local censored likelihood approach described in Section~\ref{sec:Inference} and fit the exponential factor copula model; recall Sections~\ref{sec:Modeling} and~\ref{sec:Inference}. 

\subsection{Estimation grid and neighborhood selection}\label{Neighborhoodselection}

To describe the local dependence structure of precipitation extremes over the U.S., we generated a regular grid $\mathcal G\subset\mathcal S$ (using the WGS84/UTM zone 14N metric coordinate system) with {$2235$} grid points at an internodal distance of $60$km. When plotted with respect to longitude and latitude, this results in a ``distorted'' grid, owing to the metric-to-degree system change.

An important step to fit our model for extremal dependence at each grid point $\mathbf{s}_0\in\mathcal G$ using the local estimation approach described in Section~\ref{sec:Inference} is to select a suitable number of nearest neighbors $D_0$, which can vary over space. Although cross-validation techniques {would} usually {be} advisable, pragmatic approaches are often adopted in practice: in the time series context, \citet{Davison.Ramesh:2000} suggested selecting the local likelihood bandwidth by the naked eye, while in the spatial context, \citet{anderes2011local} advocated a heuristic method based on a measure of spatial variation in the estimated parameters. In principle, the choice of $D_0$ should be such that the spatial dependence structure of threshold exceedances is approximately stationary within each selected neighborhood $\mathcal N_{\mathbf{s}_0;D_0}$ around $\mathbf{s}_0\in\mathcal G$. Small neighborhoods (with small $D_0$) yield good stationary approximations, but poor statistical efficiency, and vice versa, and our simulation study in Section~\ref{Simulation} suggests that $D_0$ should be as large as our computational resources permit, provided the dependence structure is not overly non-stationary. In the hydrological literature, a variety of tests to assess homogeneity of marginal distributions were proposed; see for instance~\citet{lu1992sampling}, \citet{fill1995homogeneity}, \citet{hosking1993some}, and \citet{hosking2005regional}. Testing for stationarity of the extremal dependence structure is, however, much more complicated. Here, for simplicity, we assume that the local dependence structure of extremes is stationary whenever {this} is the case for margins, and we follow the recommendations of \citet{viglione2007comparison} by testing the homogeneity of the margins through a compromise between the \citet{hosking1993some} test and a modified Anderson--Darling test \citep{scholz1987k}. Our ad-hoc neighborhood selection procedure can therefore be summarized as follows: considering an increasing nested sequence of neighborhoods, we test for marginal homogeneity until the test rejects the null hypothesis or a predefined maximum neighborhood size has been reached. In our case, we fix the maximum neighborhood to have radius $150$km and, for computational reasons, $D_0\leq 30$. {The conditions {of marginal homogeneity within a 150km radius} were not satisfied at only 35 of the 2235 grid points}. Although this procedure has no theoretical guarantee to be optimal, we have found that it yields reasonable estimates with our dataset.

\subsection{Local likelihood inference for extreme precipitation}\label{DataInference}
Following Section~\ref{sec:Inference}, we fitted the stationary exponential factor copula model~\eqref{LocalModel} within small local neighborhoods, by maximizing the censored local log-likelihood \eqref{WLL} at all 2200 grid points $\mathbf{s}_0\in\mathcal G$ (we discarded the 35 grid points where the conditions {of marginal homogeneity were not satisfied}), choosing the empirical {$80\%$ quantile} as a threshold to define extreme events. The left-hand panel of Figure~\ref{fig:quantiles} illustrates the selected threshold at each monitoring station on the scale of the observations. As the smoothness parameter $\nu$ in \eqref{corr} is difficult to estimate, we adopted {the} profile likelihood approach {described in Section~\ref{Simulation}}. Among all grid points, the likelihood was maximized in {$58.2\%$, $18.6\%$, $11.3\%$, and $11.8\%$} cases for $\nu=0.5,1,1.5,2$, respectively. To be able to easily compare rate and scale parameter estimates across space, we then chose to fix $\nu=0.5$ at all locations, which boils down to using an exponential correlation function for the underlying Gaussian process. Thanks to the embarrassingly parallel nature of our local likelihood estimation procedure, we could make an efficient use of distributed computing resources by fitting the local model at each grid point independently; {on average, each set of 32 parallel fits took 3 hours. The fits were computed on a cluster with 70 nodes of 32 cores each.}

To understand the dynamics of the extremal dependence strength over space, Figure~\ref{Output2.pdf} shows {$\chi_h(u)$ at high but finite levels (specifically, $u=0.95, 0.98$) for $h = 20, 40, 80, 160$km. 
\begin{figure}[t!]
\begin{center}
	\hskip0.5cm {\bf $u=0.95$} \hskip7cm {\bf $u=0.98$}
\end{center}\vskip-.2cm
\centering
\includegraphics[width=0.48\linewidth]{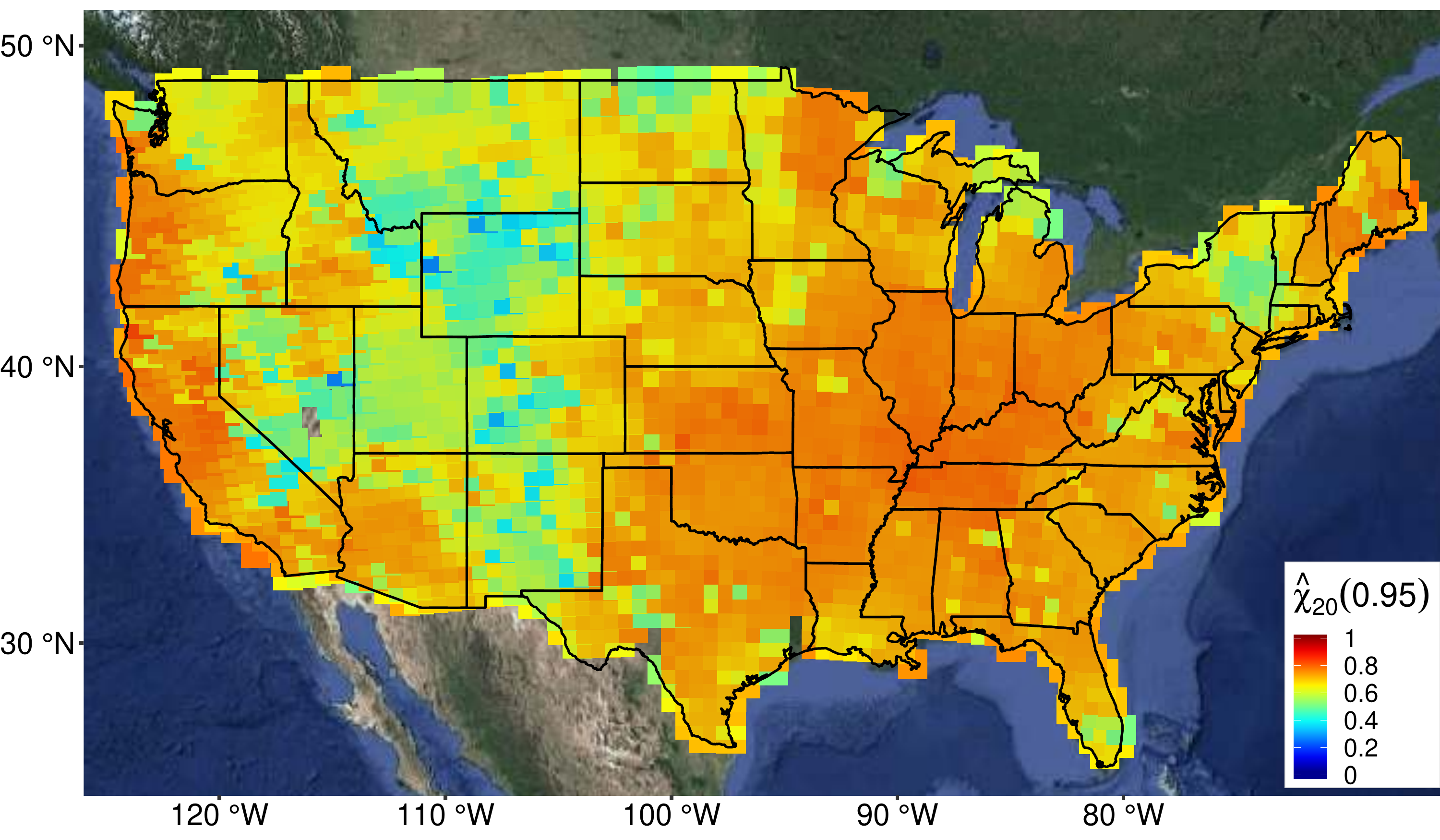}\hspace{10pt}
\includegraphics[width=0.48\linewidth]{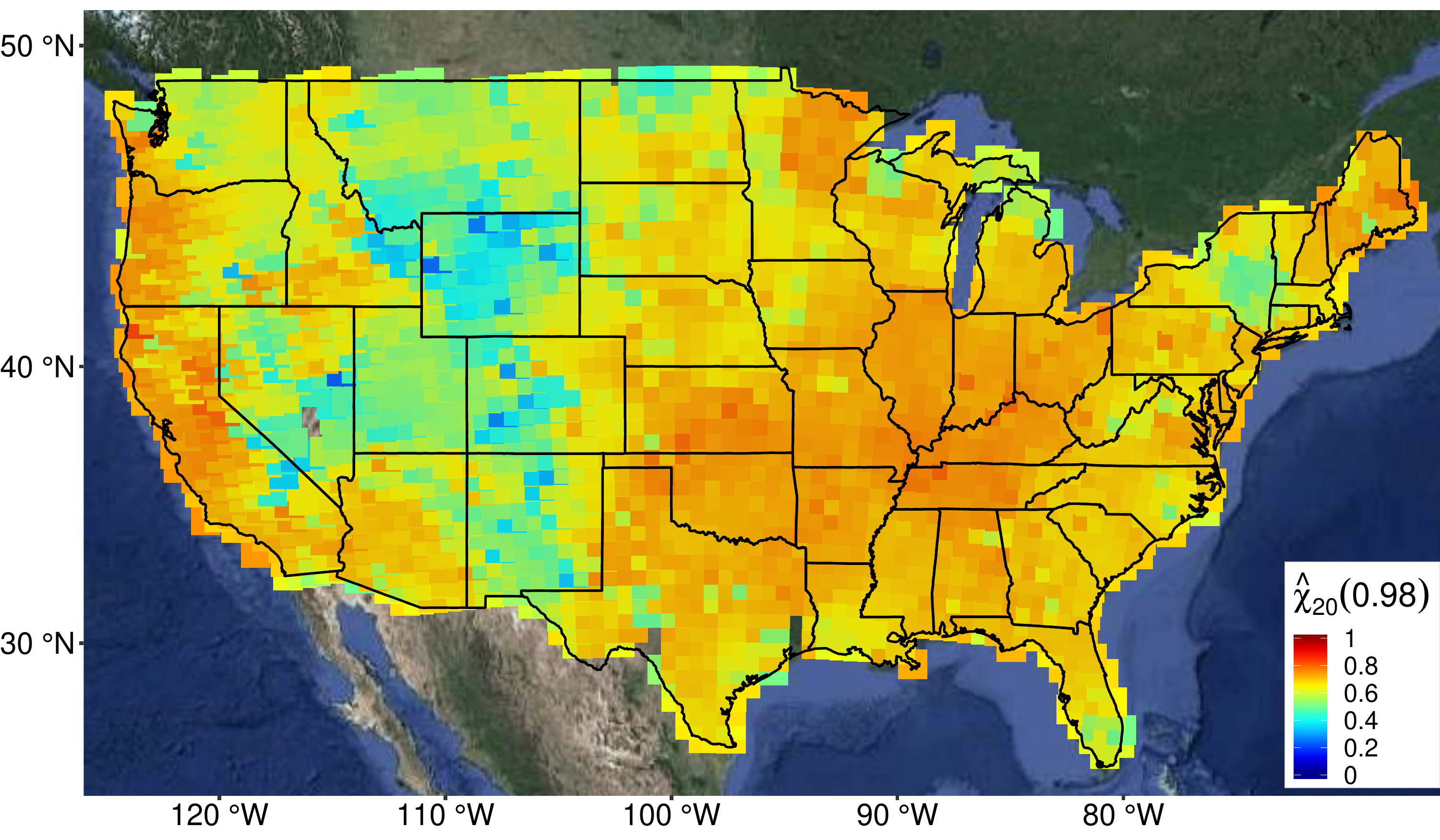}\\[5pt]
\includegraphics[width=0.48\linewidth]{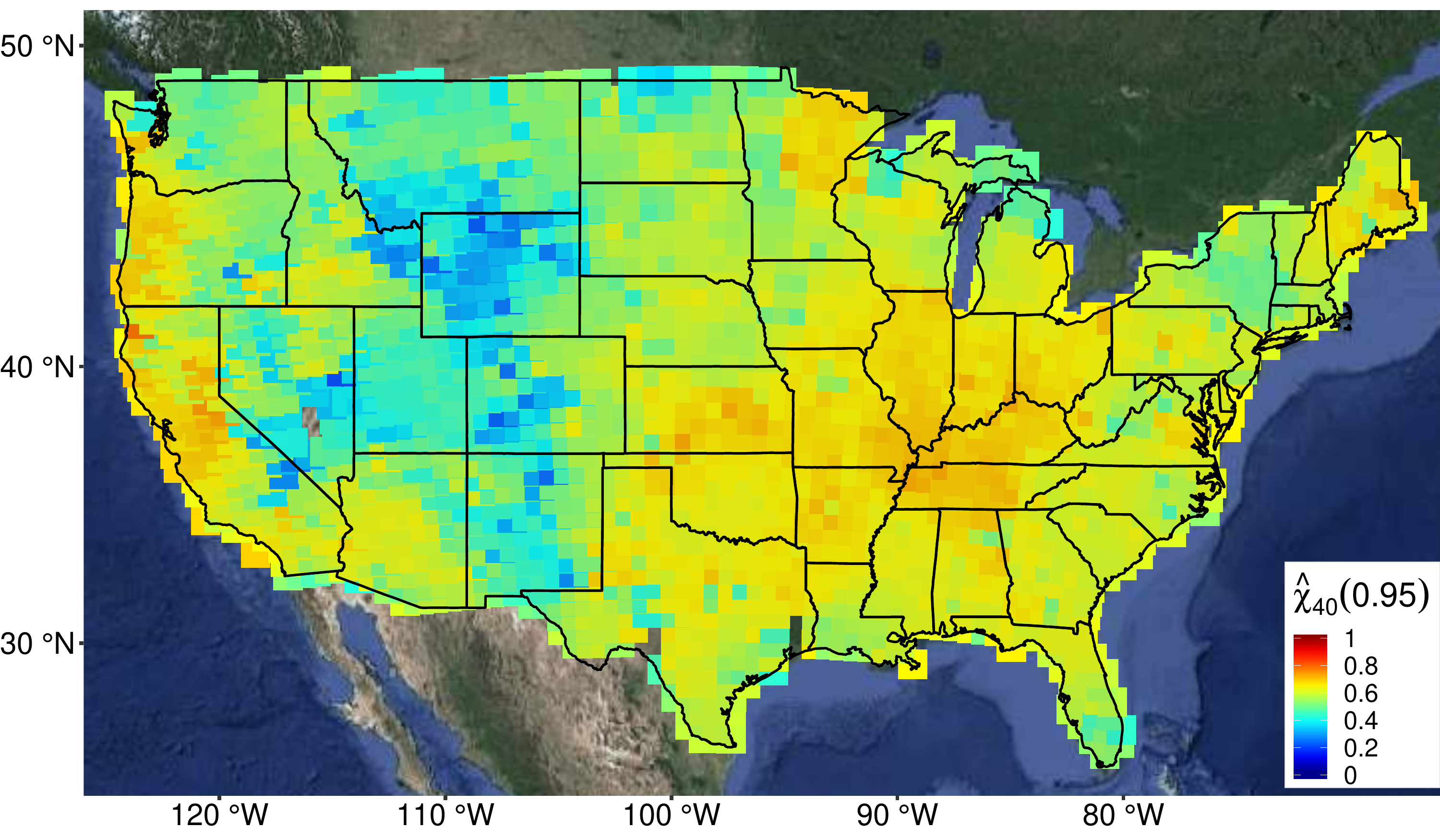}\hspace{10pt}
\includegraphics[width=0.48\linewidth]{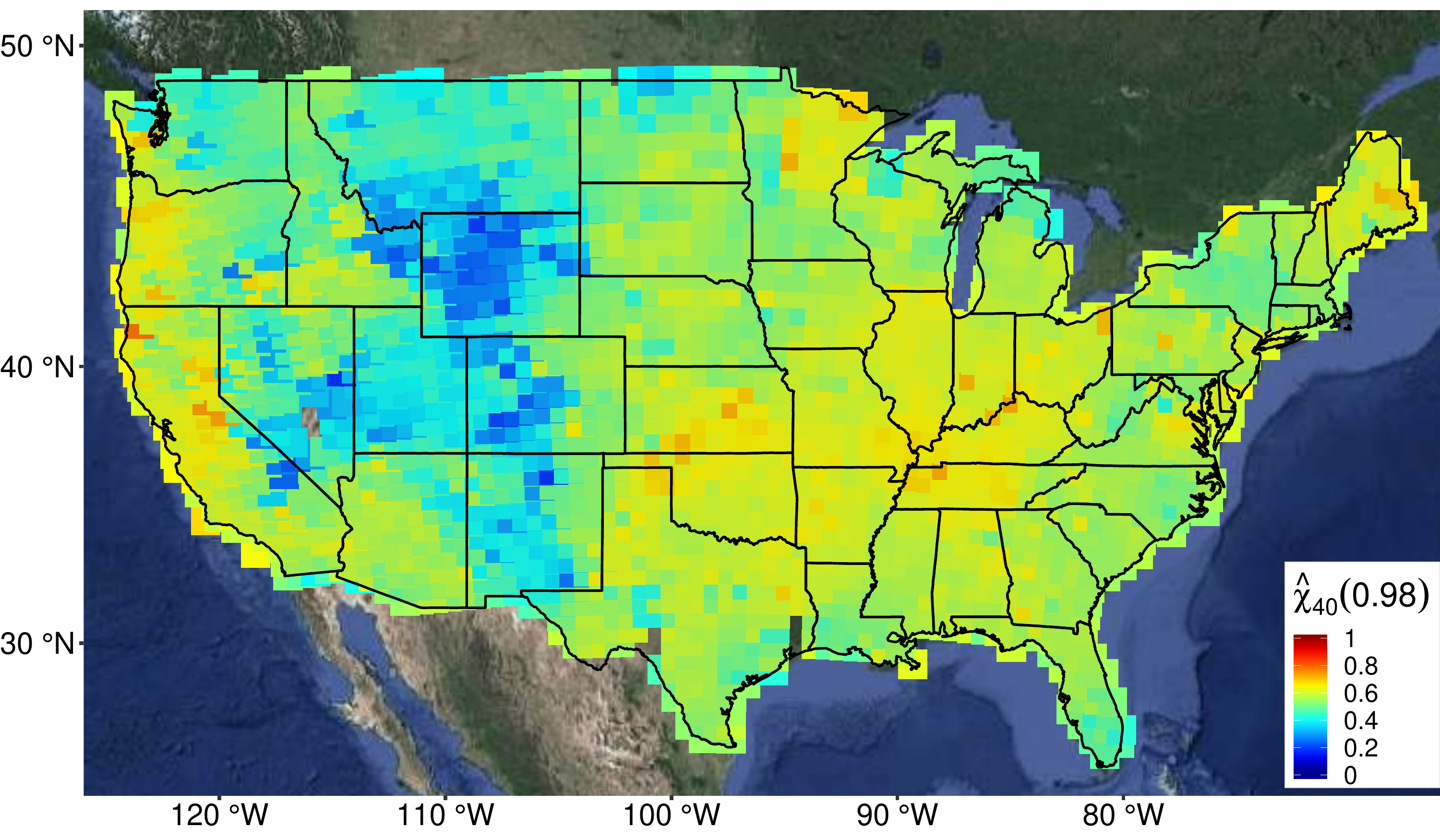}\\[5pt]
\includegraphics[width=0.48\linewidth]{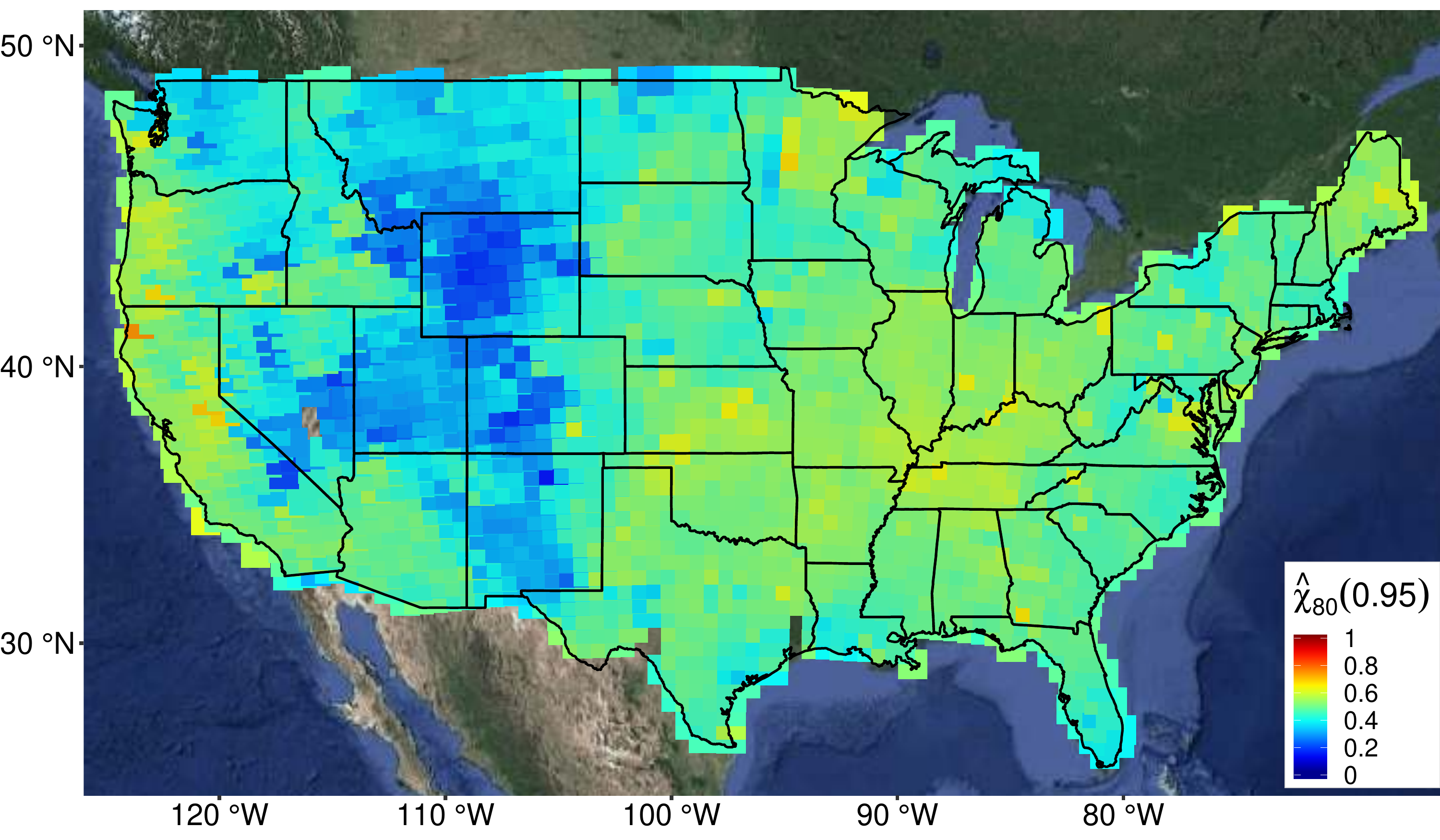}\hspace{10pt}
\includegraphics[width=0.48\linewidth]{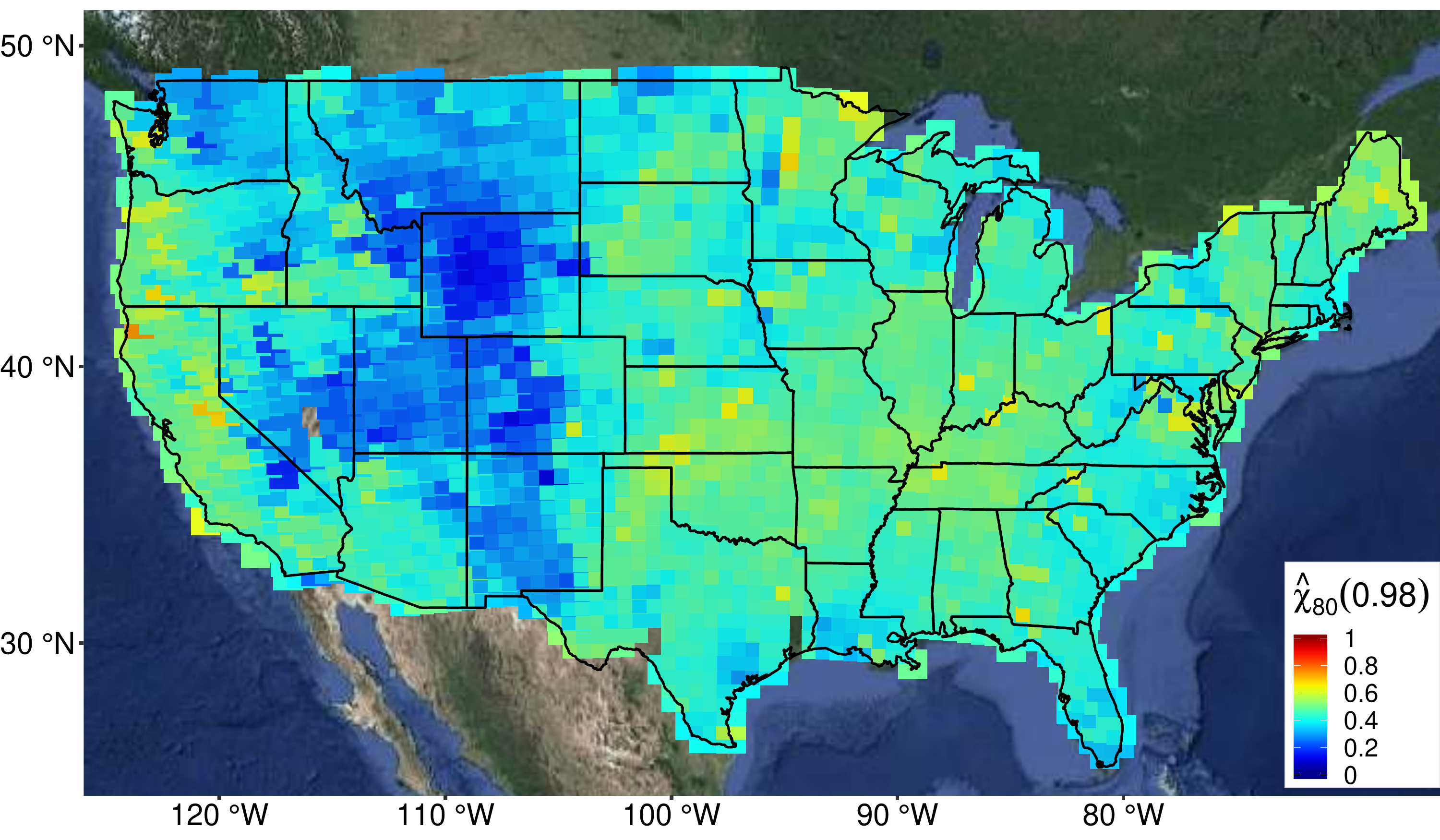}\\[5pt]
\includegraphics[width=0.48\linewidth]{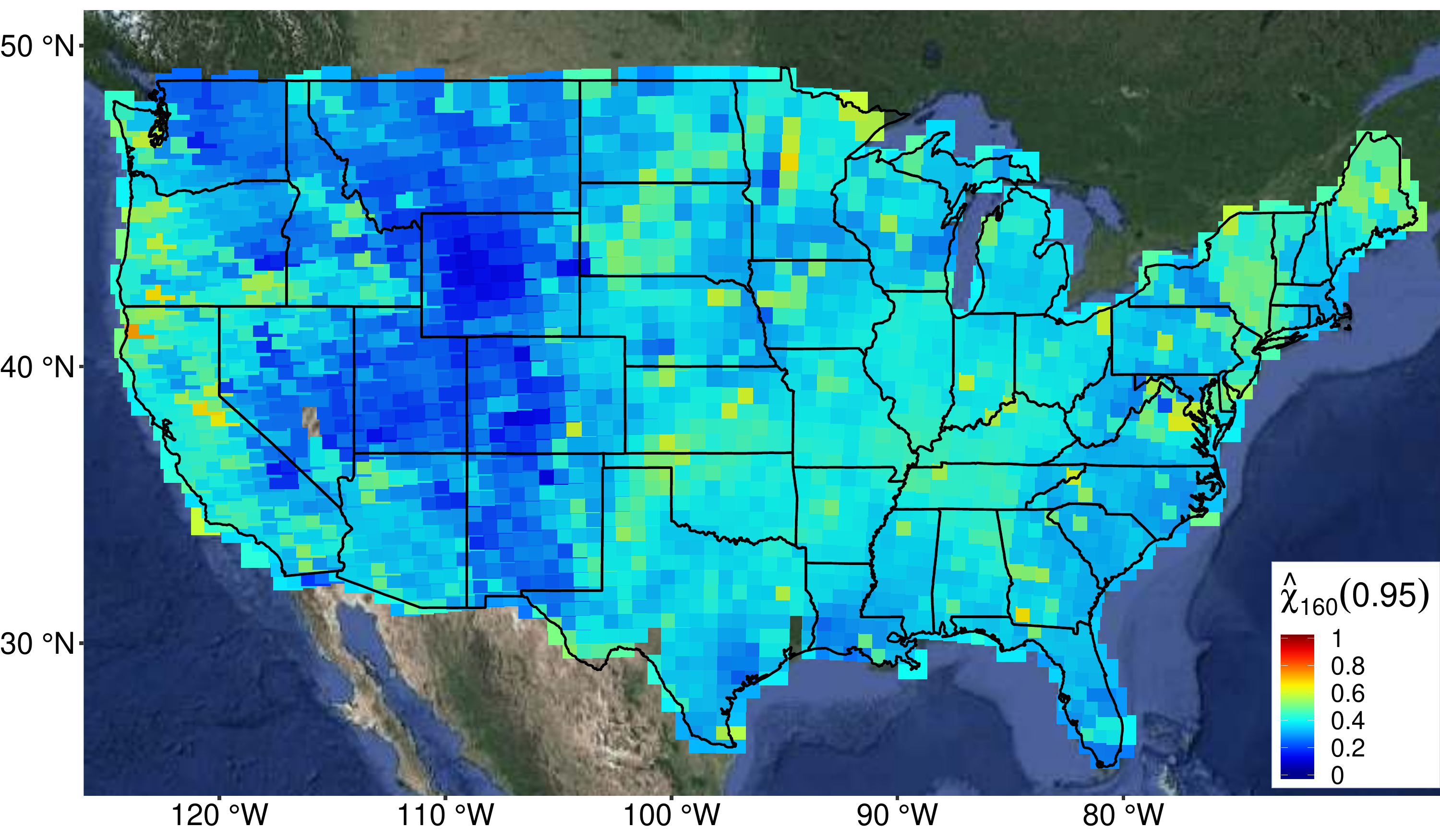}\hspace{10pt}
\includegraphics[width=0.48\linewidth]{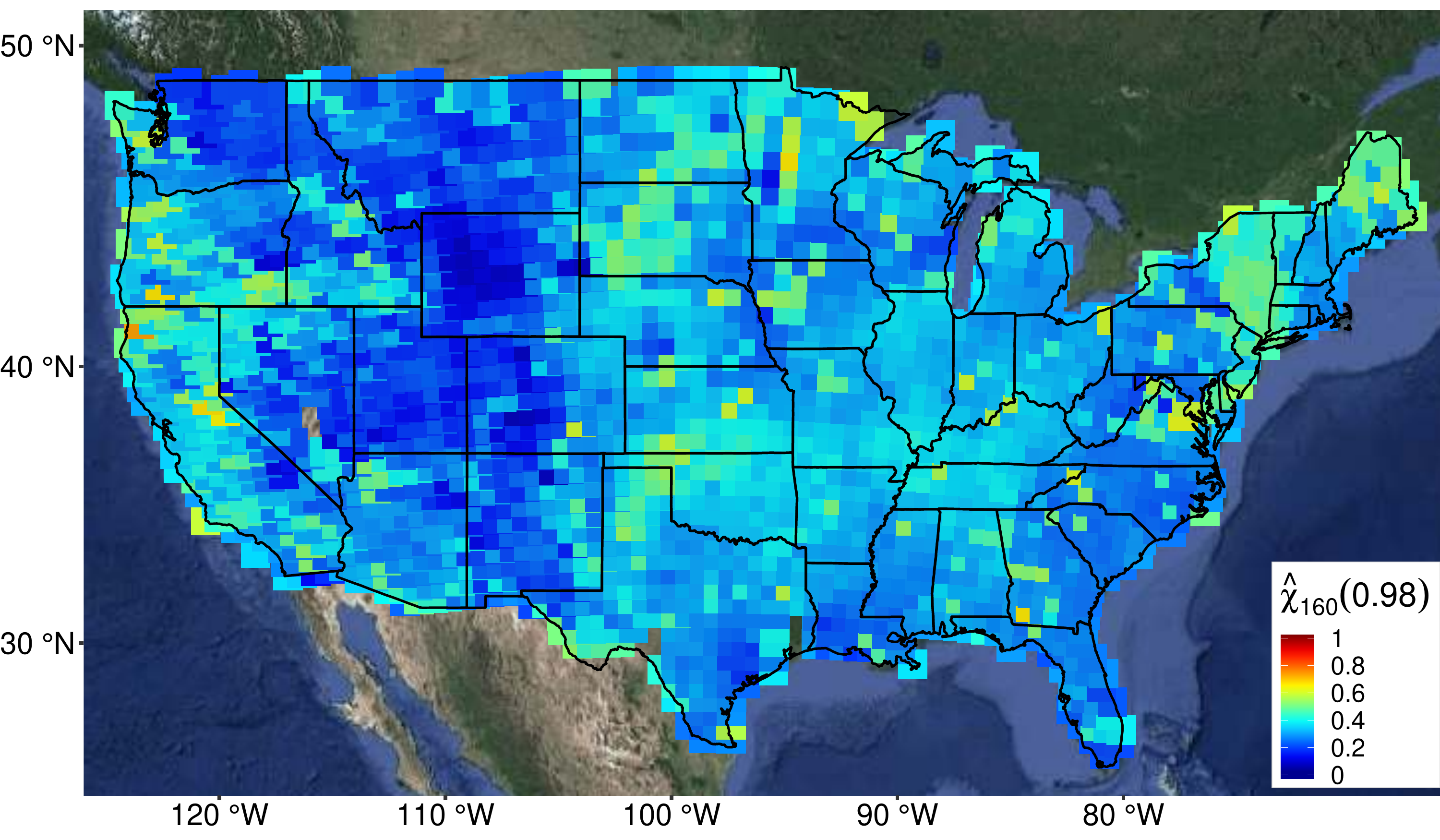}\vspace{-.3cm}
\caption{\footnotesize {Estimates of $\chi_h(u)$ for $h = 20, 40, 80, 160$km (top to bottom) and $u=0.95, 0.98$ (left to right).}}
 \label{Output2.pdf}
\end{figure}
Since our model is fitted locally, interpretation of $\chi_h(u)$ is possible within the {selected} neighborhoods, and care must be taken at locations where the radius used to fit the model is smaller than the selected distance $h${; see Section~3.2 in the Supplementary Material for details on the radius chosen {at} each grid point}}. As expected, tail dependence weakens at larger distances, and the patterns are quite smoothly varying over space. Extremal dependence is stronger in the Pacific West, central region (except for Colorado), mid-South, mid-West and North East. By contrast, some states in the South East (e.g., North/South Carolina and Florida) and the Rocky Mountains show weaker extremal dependence, suggesting that extreme precipitation events are more localized in these regions. {To assess the variability of the estimates, we computed standard deviations for the estimates of $\chi_h(u)$ from 300 block bootstrap samples with monthly blocks. As we can see from Figure~\ref{fig:sdchi.pdf}, estimates are relatively focused, with quite small standard deviations overall that are slightly higher with larger distances and thresholds.} {To complement this analysis, Section 3.3 in the Supplementary Material shows 95\% confidence intervals for $\chi_h(u)$ at some selected locations, while Sections 3.4 and 3.5 show results for $\chi_h$ (the limiting case) and the estimated log-rate and log-range parameters, respectively.}

\begin{figure}[t!]
\centering
	\includegraphics[scale = 0.28]{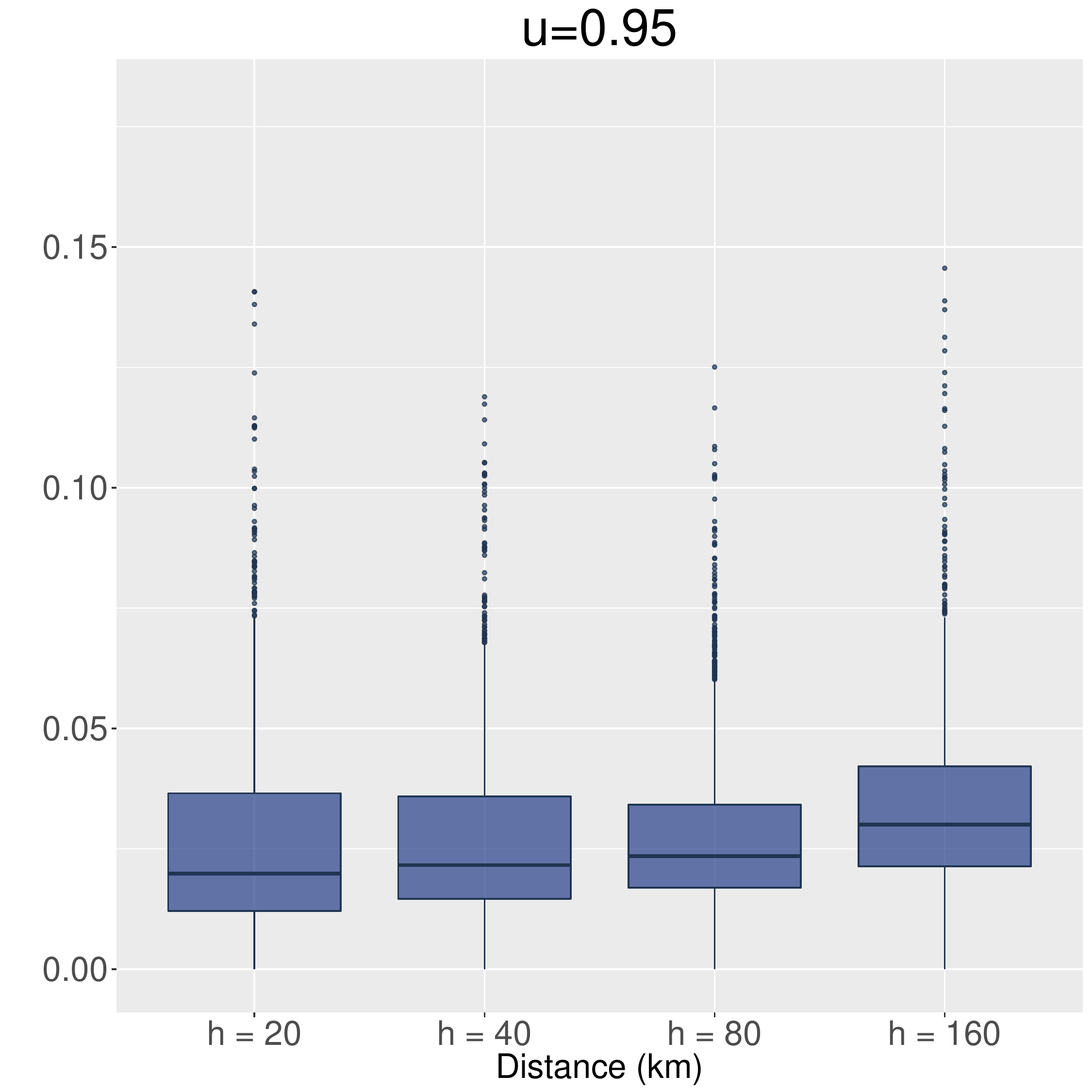}
	\includegraphics[scale = 0.28]{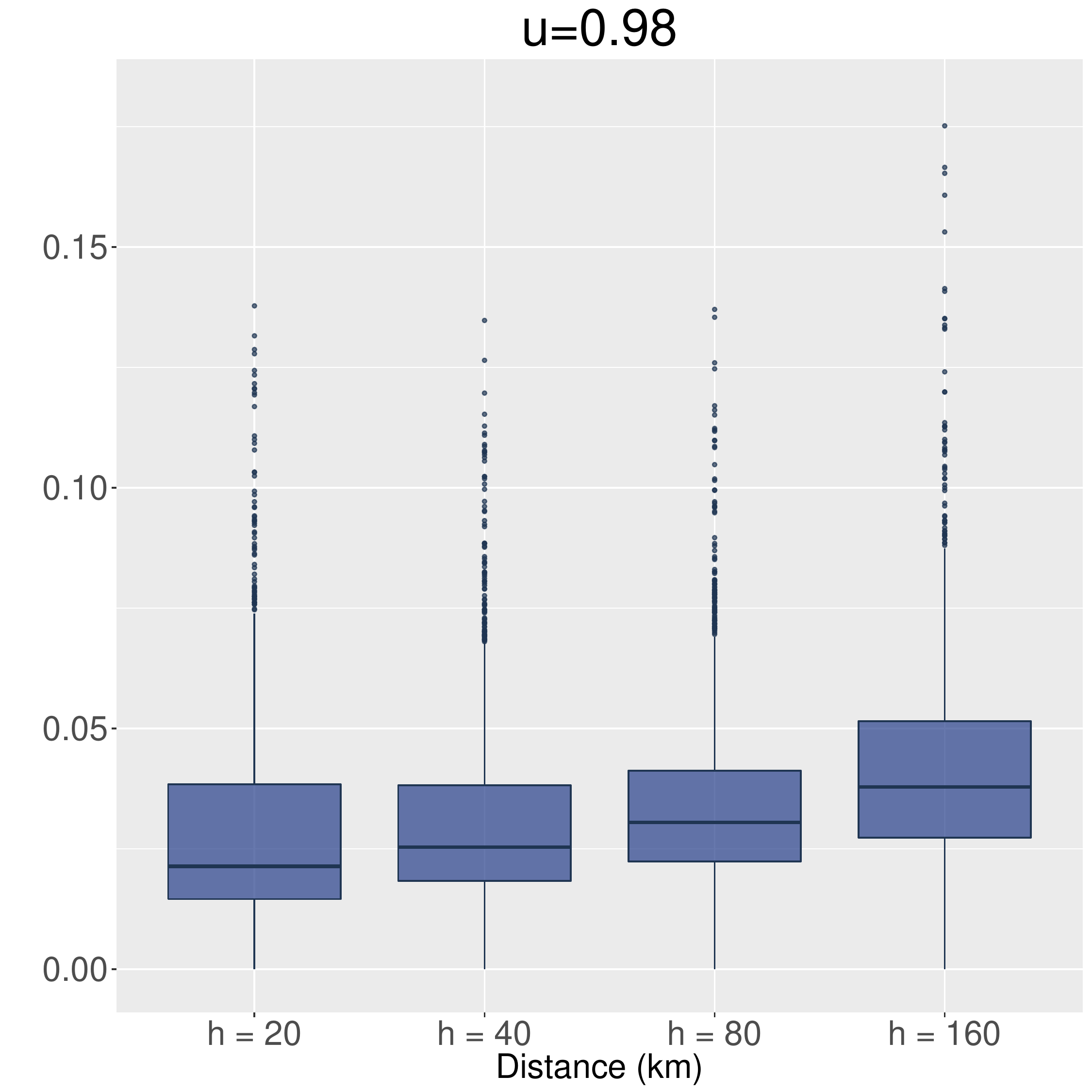}
	\caption{\footnotesize {Summary of standard deviations for estimates of $\chi_h(u)$ for $u=0.95$ ({left)} and $u=0.98$ ({right}) and each distance $h=20, 40, 80, 160$km. The standard deviations were computed from 300 block bootstrap samples with monthly blocks.}}
 \label{fig:sdchi.pdf}

\end{figure}

To further validate our fitted model, Figure~\ref{fig:chiu} compares empirical and fitted conditional tail probabilities, as defined in~\eqref{eq:chiu}. The scarcity of observations in the right tail makes the empirical estimates {of} $\chi_h(u)$ highly variable, which in turn yields wide empirical confidence envelopes (in dark blue). Nevertheless, the pointwise {model-based} estimates (in green) and the 95\% confidence envelope (in light blue) suggest that our model succeeds in capturing {the} {decaying} extremal dependence {between} close-by sites, {with relatively low uncertainty}. Overall, our model has great tail flexibility, and our local estimation approach can uncover complex non-stationary patterns of extremal dependence over space, which has important implications in terms of regional risk assessment of extreme precipitation events.

Return levels and return periods are intuitive measures to quantify the risk associated with extreme events. To assess how the different spatial tail structures affect the joint occurrence of precipitation extremes, we computed the return periods (in years) associated with the events $E_{(k)}=\{Y_1^{(k)}>F_1^{-1}(u), Y_2^{(k)}>F_2^{-1}(u),Y_3^{(k)}>F_3^{-1}(u)\}$ of observing simultaneous extreme {precipitation} at three nearby stations {within} the same state $k$, $k=1,\ldots,5$, where $u\in(0,1)$ is a high (uniform) quantile. The return periods were {estimated} {by} simulating $N = 5\times 10^5$ replicates from our fitted copula model {at each one of the 15 locations. These replicates were transformed to a common uniform scale using~\eqref{cdfW1}, and event probabilities $p_{(k)}=\pr\{E_{(k)}\}$ were then {estimated as} $\hat{p}_{(k)} = N^{-1}\sum_{i=1}^N\mathbb I\{U_{i1}^{(k)} > u, U_{i2}^{(k)} >u, U_{i3}^{(k)} >u\}$, where $U_{ij}^{(k)}$ is the $i$-th simulated value at the $j$-th station ($j=1,2,3$) in state $k$, and $\mathbb I(\cdot)$ is the indicator function. Return period {estimates} were finally obtained by taking the inverse of these {estimated} probabilities $\hat p_{(k)}$, adjusting for the time unit (years).} Table~\ref{tablePROBS} displays the average of the thresholds $F_1^{-1}(u)$, $F_2^{-1}(u)$ and $F_3^{-1}(u)$ used at the three stations on the scale of the data {corresponding to the $94\%$ and $99\%$ quantiles} for Louisiana, Mississippi, Kentucky, Florida and Tennessee, notorious to have experienced several extreme precipitation events in the recent past. 
 \begin{table}
 \begin{center}
 \caption{ 
  \label{tablePROBS}\footnotesize Pointwise estimates and 95\% confidence intervals for the return periods (in years) associated with the joint probability of observing an extreme precipitation event exceeding the {$94\%$ and $99\%$ quantiles} at three locations simultaneously. The $1^{\rm st}$ column reports the state of these locations, while the $2^{\rm nd}$ column represents the {average thresholds} (in hundredths of an inch) {for each quantile.}}
   \vspace{5pt}       
\begin{tabular}{l| c| *{3}{l}}
&Average thresholds  & \multicolumn{3}{c}{Return period (years)} \\\hline
State&$u=94\%$ / $u=99\%$ & Model $u=94\%$ & Empirical $u=94\%$ & Model $u=99\%$\\\hline
  Louisiana &343.0 / 664.9 & 0.51 (0.51, 0.52) & 0.46 (0.35,0.58) & 1.50 (1.46, 1.55)\\
  Mississippi &317.9 / 493.2 & 0.54 (0.53, 0.55) & 0.95 (0.26,1.63) & 1.69 (1.64,1.74) \\
  Kentucky & 267.8 / 514.3& 0.55 (0.54, 0.56) & 0.63 (0.45,0.82) & 1.80 (1.74,1.86)\\
  Florida &343.0 / 611.1 & 0.86 (0.84, 0.88) & 1.01 (0.33,2.34) & 4.15 (3.93, 4.37)\\
  Tennessee & 319.9 / 605.1& 0.65 (0.64, 0.66) & 0.75 (0.52,0.98) & 2.24 (2.16,2.32)\\
\end{tabular}
\end{center}
\end{table}
{Return period estimates along with bootstrap-based 95\% confidence intervals show that the selected locations in the five states are {at} risk of simultaneous extreme events. According to our fitted model, such spatial extreme events occur on average once every 6 to 11 months when $u=94\%$ and once every 1 to 4 years when $u=99\%$. {Empirical estimates and confidence intervals for the return periods {(provided only for $u=94\%$ due to the high variability of empirical estimates when $u=99\%$)} show that our estimates are sensible when we take the associated uncertainty into account}. These return period estimates, of course, strongly depend on the relative positions of the selected locations. Note that the uncertainty, which accounts for both marginal and dependence estimation uncertainty, is quite low for the model-based estimates.



\vspace{-.3cm}

\section{Concluding remarks}\label{conclusion}
In this paper, we have developed a censored local likelihood methodology based on factor copula models for sub-asymptotic spatial extremes observed over large heterogeneous regions. Our proposed modeling approach can capture complex non-stationary patterns and is well-suited when the dependence strength weakens as events become more extreme. This contrasts with the current spatial extremes literature, which often relies on asymptotic models with a more rigid tail structure. 
Our extensive simulation experiments demonstrate the flexibility of our model and the efficiency of our local estimation approach based on high threshold exceedances, while providing some guidance on the selection of regional neighborhoods.


{Our methodology provides a general picture of the local dependence characteristics governing precipitation extremes across the contiguous U.S.. Specifically, our fitted model revealed a diverse and complex tail dependence structure, with rich and intuitive spatial patterns. Because of the small-scale nature of the precipitation process {(see, e.g.,~\citealp{wilby1998statistical,katz2002statistics})}, our paper was focused on uncovering local dependence characteristics rather than trying to accurately capture long-distance properties. Our proposed model lacks long-distance independence, which is a common feature of most models for spatial extremes. Thus, our model must be interpreted locally, or perhaps regionally, and care is needed when extrapolating to distances higher than the ones used to fit the model.} 
Further research should be devoted to developing models for spatial and spatio-temporal extremes that can more realistically capture long-range independence. {\cite{morris2017space} discuss one way to achieve this based on random partitions of the study region}.


Because our sub-asymptotic model captures weakening of dependence as events become more extreme, it is well-suited at levels less extreme than usually considered in practice. In our data application, we fitted our model at the $80\%$-quantile, and showed that the fit was satisfactory for a wide range of quantiles. Our approach has the additional advantage of using more data for inference, which yields lower estimation variability, as illustrated in Figure~\ref{fig:boxplots0809.pdf} which compares the fits at the $80\%$ and $90\%$ thresholds.
\begin{figure}[!t]
\centering
	\includegraphics[scale = 0.33]{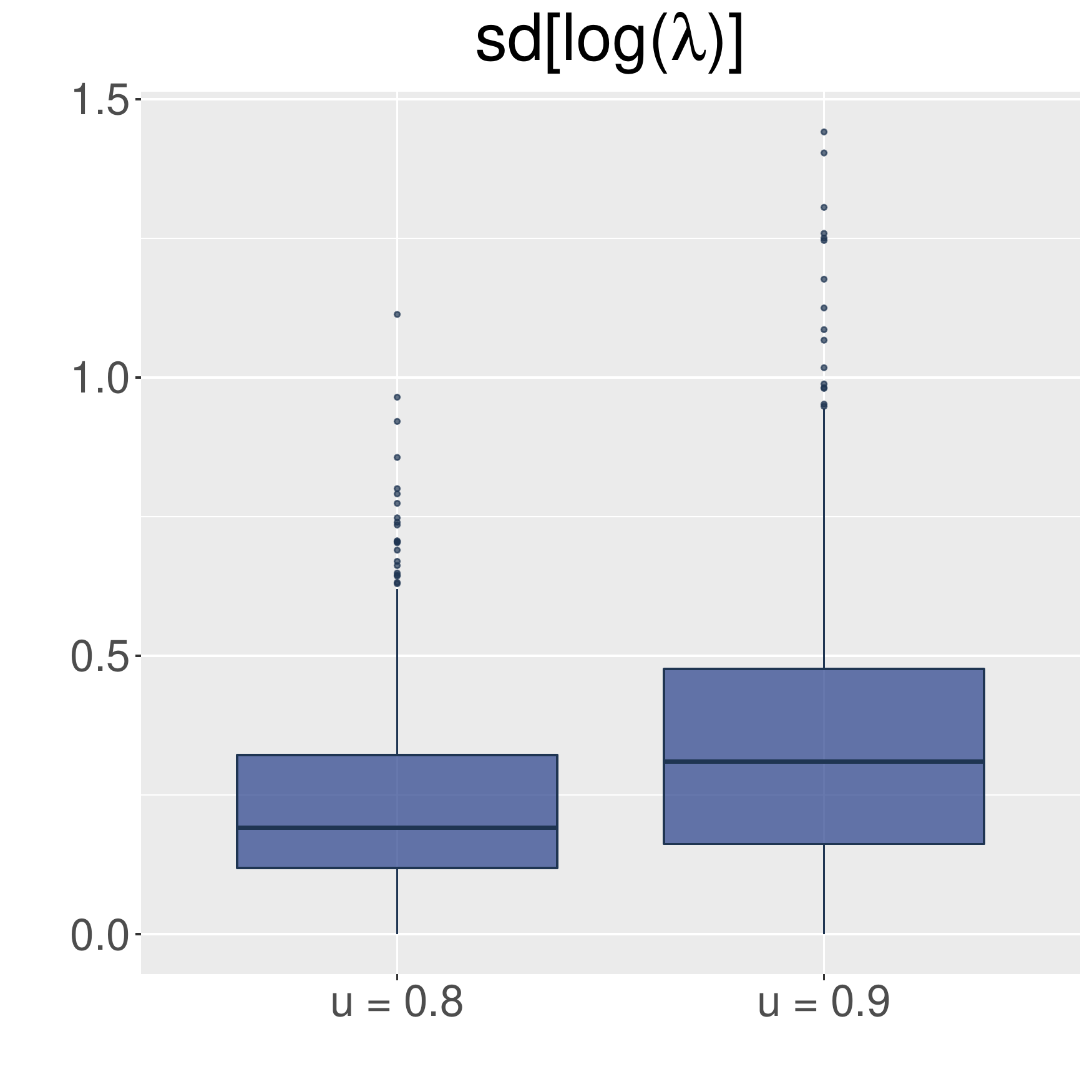}
	\includegraphics[scale = 0.33]{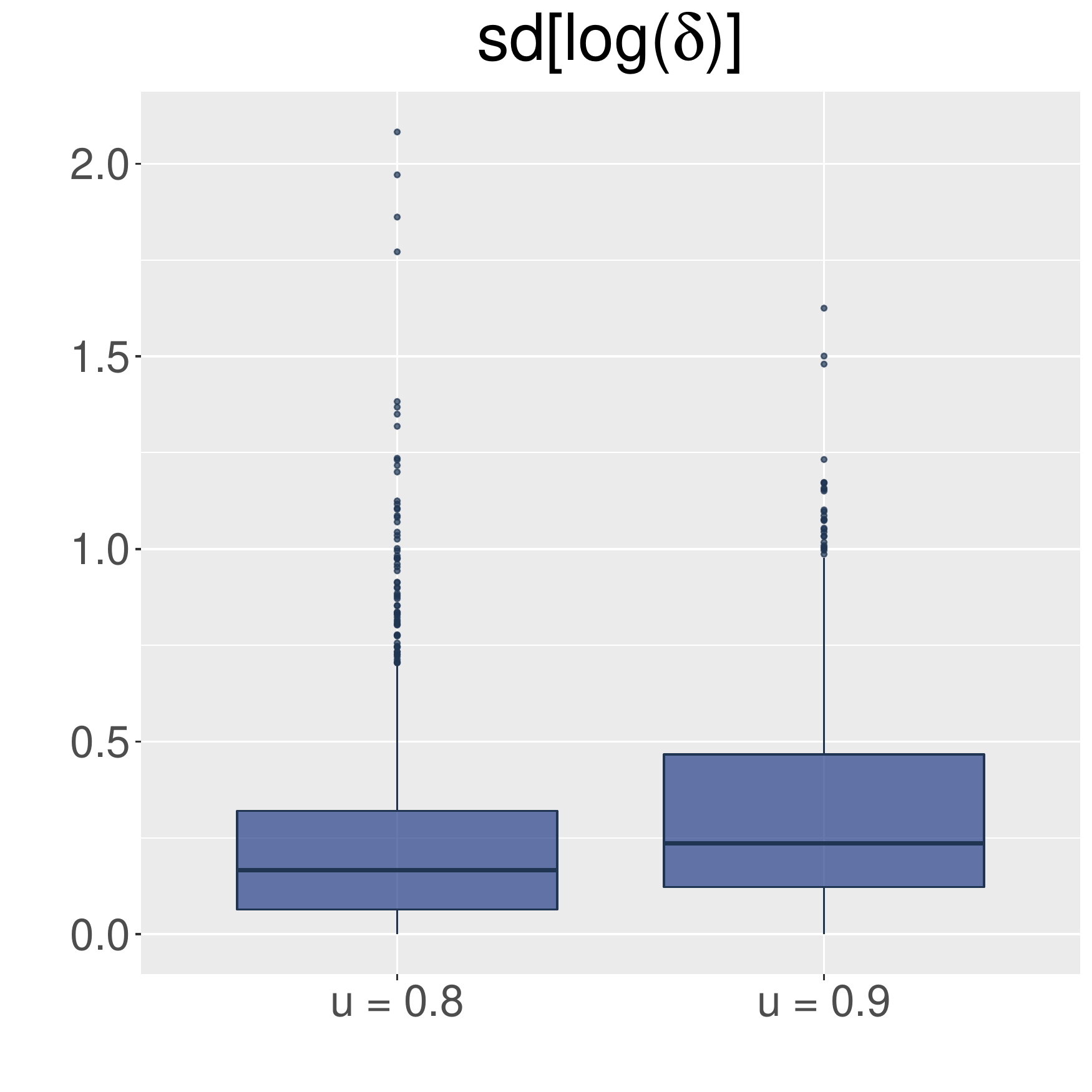}
	\caption{\footnotesize Boxplots of the standard deviation of the estimates of the log-rate (left) and log-range (right) parameters using the $80\%$ and $90\%$ quantiles.}
	\label{fig:boxplots0809.pdf}
\end{figure}

{In our application, standard deviations were computed using a block bootstrap procedure, which can be computationally expensive. It would also be possible to exploit the faster-to-compute Fisher information, provided the assumption of temporal independence is satisfied.}

\section*{Acknowledgements}
We thank Luigi Lombardo (University of Twente) for cartographic support and Eduardo Gonz\'alez (KAUST) for computational support. We extend our thanks to Dan Cooley (Colorado State University) for helpful comments and suggestions. Support from the KAUST Supercomputing Laboratory and access to Shaheen is also gratefully acknowledged. We are particularly grateful to the two referees for their comments and suggestions that have led to a much improved version of this paper.

\appendix\label{Appendix}
\section{Simplified expressions for the model likelihood}\label{simplified}
In this appendix, we obtain simplified expressions for the local censored likelihood function~\eqref{WLL} based on Model~\eqref{LocalModel}, in order to compute the maximum likelihood estimates in a reasonable time, while avoiding numerical instabilities.

\subsection{Joint density of the $W$ process in \eqref{LocalModel}}
From expression \eqref{pdfW}, {we have}
\begin{align*}
 f_{D}^{\mathbf{W}}(\bold{w})&= \lambda\int_0^\infty \phi_D(\bold{w}-v\bold{1}_D;\mathbf{\Sigma}_{\mathbf{Z}})\exp(-\lambda v)\dt v\\
 &=  \lambda\int_0^\infty (2\pi)^{-D/2}(\det\mathbf{\Sigma}_{\mathbf{Z}})^{-1/2}\exp\left(\frac{{a_4}^2a_3 - a_1}{2}\right)\exp\left\{-\frac{1}{2a_3^{-1}}\left(v - a_4\right)^2\right\}\dt v\\
 &=  \lambda (2\pi)^{-(D-1)/2}a_3^{-1/2}(\det\mathbf{\Sigma}_{\mathbf{Z}})^{-1/2}\exp\left(\frac{{a_4}^2a_3 - a_1}{2}\right)\Phi(a_3^{1/2}a_4),
\end{align*}
where $a_1 = \bold{w}^T\mathbf{\Sigma}_{\mathbf{Z}}^{-1}\bold{w}$, $a_2 = \bold{1}_D^T\mathbf{\Sigma}_{\mathbf{Z}}^{-1}\bold{w}$, $a_3 = \bold{1}_D^T\mathbf{\Sigma}_{\mathbf{Z}}^{-1}\bold{1}_D$, $a_4 = (a_2 - \lambda)/a_3$, $\bold{w} = (w_{1},\ldots,w_{D})^T$, $\bold{1}_D = (1,\ldots,1)^T\in\mathbb R^D$. The marginal density $f_1^{\mathbf{W}}(w;\lambda)$ with $D = 1$ can be easily deduced.

\subsection{Joint distribution of the $W$ process in \eqref{LocalModel}}
Here we express the finite-dimensional distribution $F_D^{\mathbf{W}}$ as a function of conditional multivariate normal distributions. This alternative formulation allows us to use efficient routines implemented in \texttt{R}, substantially improving accuracy and execution time. Using integration by parts in expression \eqref{cdfW}, {we} obtain
\begin{align*}
F_D^{\mathbf{W}}(\bold{w}) &= \Phi_{D}(\bold{w};\mathbf{\Sigma}_{\mathbf{Z}}) - \sum_{j = 1}^D\int_0^\infty\Phi_{D-1}(\bold{w}_{-j} - v\bold{1}_{D-1}- \mathbf{\mu}_{\mid j}; \mathbf{\Sigma}_{\mathbf{Z}\mid j})\phi(w_j - v)\exp(-\lambda v)\dt v\\
&= \Phi_{D}(\bold{w};\mathbf{\Sigma}_{\mathbf{Z}}) - \sum_{j = 1}^D\exp\left(\frac{\lambda^2}{2} - \lambda w_j\right)\int_0^\infty\Phi_{D-1}\left(\bold{w}_{-j} - v\bold{1}_{D-1}- \mathbf{\mu}_{\mid j}; \mathbf{\Sigma}_{\mathbf{Z}\mid j}\right)\phi(v-w_j + \lambda)\dt v,
\end{align*}
where $\bold{w}_{-j} = (w_1,\ldots, w_{j-1}, w_{j+1},\ldots,w_D)^T\in\mathbb R^{D-1}$ and $\mathbf{\mu}_{\mid j}$ and $ \mathbf{\Sigma}_{\mathbf{Z}\mid j}$ are the conditional mean and covariance matrix of $\bold{w}_{-j}|w_j$, respectively. Precisely,
\begin{equation*}
\mathbf{\mu}_{\mid j} = \mathbf{\Sigma}_{\mathbf{Z};-j,j}(w_j - v),\qquad \mathbf{\Sigma}_{\mathbf{Z}\mid j} = \mathbf{\Sigma}_{\mathbf{Z};-j,-j} - \mathbf{\Sigma}_{\mathbf{Z};-j,j}\mathbf{\Sigma}_{\mathbf{Z};-j,j}^T,\qquad j=1,\ldots,D,
\end{equation*}
where $\mathbf{\Sigma}_{\mathbf{Z};-j,j}$ denotes the $j$th column of the matrix $\mathbf{\Sigma}_{\mathbf{Z}}$ with the $j$th row removed, etc. Now,
\begin{align}
\int_0^\infty\Phi_{D-1}&\left(\bold{w}_{-j} - v\bold{1}_{D-1}- \mathbf{\mu}_{\mid j}; \mathbf{\Sigma}_{\mathbf{Z}\mid j}\right)\phi(v-w_j + \lambda)\dt v \nonumber\\
&= \int_0^\infty\Phi_{D-1}\left\{\bold{w}_{-j} - \mathbf{\Sigma}_{\mathbf{Z};-j,j}w_j-v(\bold{1}_{D-1} - \mathbf{\Sigma}_{\mathbf{Z};-j,j});\mathbf{\Sigma}_{\mathbf{Z}\mid j}\right\}\phi(v-w_j + \lambda)\dt v \nonumber\\
&= \pr\left(\bold{Y}\leq \mathbf{\Sigma}_{\mathbf{Z}\mid j}^{-1/2}\left\{\bold{w}_{-j} - \mathbf{\Sigma}_{\mathbf{Z};-j,j}w_j-\tilde{V}(\bold{1}_{D-1} - \mathbf{\Sigma}_{\mathbf{Z};-j,j})\right\}, \tilde{V}>0\right)\nonumber\\
&= \pr\left(\bold{Q}_j\leq \bold{w}_{-j} - \mathbf{\Sigma}_{\mathbf{Z};-j,j}w_j,-\tilde{V}\leq 0\right)\label{eq:proba}
\end{align}
where $\bold{Y}\sim\mathcal{N}(\bold{0}, \bold{I}_{D-1})$ follows the standard multivariate normal distribution in dimension $D-1$, and is independent of the univariate normal random variable $\tilde{V}$ with mean $w_j - \lambda$ and unit variance, and where $\bold{Q}_j=\mathbf{\Sigma}_{\mathbf{Z}\mid j}^{1/2}\bold{Y} + \tilde{V}(\bold{1}_{D-1} - \mathbf{\Sigma}_{\mathbf{Z}; -j,j})$. To compute \eqref{eq:proba}, notice that the $D$-dimensional vector $\bold{Q}_{j,0} =(\bold{Q}_j^T,-\tilde{V})^T$ is jointly multivariate normal, and that $\bold{Q}_j$ has mean $(w_j -\lambda)(\bold{1}_{D-1} - \mathbf{\Sigma}_{\mathbf{Z};-j,j})$ and covariance matrix $\mathbf{\Sigma}_{\mathbf{Z}\mid j} + (\bold{1}_{D-1} - \mathbf{\Sigma}_{\mathbf{Z};-j,j})(\bold{1}_{D-1} - \mathbf{\Sigma}_{\mathbf{Z};-j,j})^T=\bold{1}_{D-1}\bold{1}_{D-1}^T - 2\bold{1}_{D-1}\mathbf{\Sigma}_{\mathbf{Z};-j,j}^T+\mathbf{\Sigma}_{\mathbf{Z};-j,-j}$. Hence, {we have}
\begin{equation*}
\pr\left(\bold{Q}_j\leq \bold{w}_{-j} - \mathbf{\Sigma}_{\mathbf{Z};-j,j}w_j,-\tilde{V}\leq 0\right)= \pr\left(\bold{Q}_{j,0}\leq \bold{q}_{j,0}\right),\quad \bold{Q}_{j,0} = (\bold{Q}_j^T, -\tilde{V})^T\sim\mathcal{N}_D\left(\mathbf{\mu}_{j,0},\mathbf\Omega_{j,0}\right),
\end{equation*}
where 
\begin{align*}
\bold{q}_{j,0} &= \begin{pmatrix}\bold{w}_{-j} - \mathbf{\Sigma}_{\mathbf{Z};-j,j}w_j\\ 0\end{pmatrix},  \quad \mathbf{\mu}_{j,0} = \begin{pmatrix} (w_j - \lambda)(\bold{1}_{D-1} - \mathbf{\Sigma}_{\mathbf{Z};-j,j})\\ \lambda -w_j \end{pmatrix},\\
\mathbf\Omega_{j,0} &=\begin{pmatrix} \bold{1}_{D-1}\bold{1}_{D-1}^T - 2\mathbf{\Sigma}_{\mathbf{Z};-j,j}+\mathbf{\Sigma}_{\mathbf{Z};-j,-j}&\mathbf{\Sigma}_{\mathbf{Z}; -j,j}-\bold{1}_{D-1}\\\mathbf{\Sigma}_{\mathbf{Z}; -j,j}-\bold{1}_{D-1}&1\end{pmatrix}.
\end{align*}
Finally, {we} obtain
\begin{align}
\label{eq:multivariatecdf}
F_D^{\mathbf{W}}(\bold{w}) &= \Phi_{D}(\bold{w};\mathbf{\Sigma}_{\mathbf{Z}})  - \sum_{j = 1}^D\exp\left(\frac{\lambda^2}{2} - \lambda w_j\right)\Phi_{D}\left(\bold{q}_{j,0}- \mathbf{\mu}_{j,0};\mathbf\Omega_{j,0}\right).
\end{align}
In particular, by setting $D=1$, the marginal distribution may be written as 
\begin{equation}
\label{eq:marginalcdf}
F_1^{\mathbf{W}}(w;\lambda) = \Phi(w) - \exp(\lambda^2/2 - \lambda w)\Phi(w - \lambda).
\end{equation}

\subsection{Partial derivatives of the joint distribution $F_D^{\mathbf{W}}$}
We want to compute $\partial_{J_i}F_D^{\mathbf{W}}$, the partial derivative of $F_D^{\mathbf{W}}$ with respect to the variables indexed by the set $J_i\subseteq\{1,\ldots,D\}$. Without loss of generality, we here assume that $J_i=\{1,\ldots,k\}$, for some integer $1\leq k\leq D-1$. Writing $r=D-k$, {we} can express the covariance matrix $\mathbf{\Sigma}_{\mathbf{Z}}$ in block notation as follows:
$$\mathbf{\Sigma}_{\mathbf{Z}} = \begin{pmatrix} \mathbf{\Sigma}_{\mathbf{Z},k}&\Sigma_{\mathbf{Z},rk}^T\\\mathbf{\Sigma}_{\mathbf{Z},rk}&\mathbf{\Sigma}_{\mathbf{Z},r}\end{pmatrix}.$$
Then, writing $\bold{w}=(\bold{w}_k^T,\bold{w}_r^T)^T$, with $\bold{w}_k=(w_1,\ldots,w_k)^T$ and $\bold{w}_r=(w_{k+1},\ldots,w_D)^T$, {we} have
\begin{align*}
\partial_{J_i}F_D^{\mathbf{W}}(\bold{w}) &=\partial_{1:k}^k\lambda\int_0^\infty\Phi_{D}(\bold{w}-v\bold{1}_D;\mathbf{\Sigma}_{\mathbf{Z}})\exp(-\lambda v)\dt v\\
&= \lambda\int_0^\infty\phi_{k}(\bold{w}_k-v\bold{1}_k;\mathbf{\Sigma}_{\mathbf{Z},k})\Phi_r(\bold{w}_r-v\bold{1}_r-\mu_{|k}; \mathbf{\Sigma}_{\mathbf{Z}\mid k})\exp(-\lambda v)\dt v,
\end{align*}
where $\mu_{\mid k}$ and $\mathbf{\Sigma}_{\mathbf{Z}\mid k}$ are the conditional mean and covariance of $\bold{w}_r\mid\bold{w}_k$, respectively, and are obtained as
$
\mu_{\mid k} = \mathbf{\Sigma}_{\mathbf{Z},rk}\mathbf{\Sigma}_{\mathbf{Z},k}^{-1}(\bold{w}_k - v\bold{1}_k)$, $\mathbf{\Sigma}_{\mathbf{Z}\mid k} = \mathbf{\Sigma}_{\mathbf{Z},r} - \mathbf{\Sigma}_{\mathbf{Z},rk}\mathbf{\Sigma}_{\mathbf{Z},k}^{-1}\mathbf{\Sigma}_{\mathbf{Z},rk}^T.
$
Tedious but straightforward calculations yield $\phi_{k}(\bold{w}_k-v\bold{1}_k;\mathbf{\Sigma}_{\mathbf{Z},k})\exp(-\lambda v) = C\; \phi\{b_3^{1/2}(v-b_4)\}$, where $C = (2\pi)^{-(k-1)/2}b_3^{-1/2}(\det{\mathbf{\Sigma}_{\mathbf{Z},k}})^{-1/2}\exp\left\{(b_4^2b_3 - b_1)/2\right\}$, with $b_1 = \bold{w}_k^T\mathbf{\Sigma}_{\mathbf{Z},k}^{-1}\bold{w}_k$, $b_2 = \bold{1}_k^T\mathbf{\Sigma}_{\mathbf{Z},k}^{-1}\bold{w}_k$, $b_3 = \bold{1}_k^T\mathbf{\Sigma}_{\mathbf{Z},k}^{-1}\bold{1}_k$, and $b_4 = (b_2 - \lambda)/b_3$. Therefore,
\begin{align}
\partial_{J_i} F_D^{\mathbf{W}}(\bold{w}) &= \lambda C\int_0^\infty\phi\{b_3^{1/2}(v-b_4)\}\Phi_r(\bold{w}_r - v\bold{1}_r-\mu_{\mid k}; \mathbf{\Sigma}_{\mathbf{Z}\mid k})\dt v\nonumber\\
&= \lambda C\int_0^\infty\phi\{b_3^{1/2}(v-b_4)\}\Phi_r\left\{\bold{w}_r - \mathbf{\Sigma}_{\mathbf{Z},rk}\mathbf{\Sigma}_{\mathbf{Z},k}^{-1}\bold{w}_k - v(\bold{1}_r- \mathbf{\Sigma}_{\mathbf{Z},rk}\mathbf{\Sigma}_{\mathbf{Z},k}^{-1}\bold{1}_k);\mathbf{\Sigma}_{\mathbf{Z}\mid k}\right\}\dt v\nonumber\\
&= \lambda C\;\pr\left(\bold{Y}\leq \mathbf{\Sigma}_{\mathbf{Z}\mid k}^{-1/2}\left\{\bold{w}_r - \mathbf{\Sigma}_{\mathbf{Z},rk}\mathbf{\Sigma}_{\mathbf{Z},k}^{-1}\bold{w}_k - \tilde{V}(\bold{1}_r- \mathbf{\Sigma}_{\mathbf{Z},rk}\mathbf{\Sigma}_{\mathbf{Z},k}^{-1}\bold{1}_k)\right\},\tilde{V}>0\right)\nonumber\\
&= \lambda C\;\pr\left(\mathbf{Q}_r\leq \bold{w}_r - \mathbf{\Sigma}_{\mathbf{Z},rk}\mathbf{\Sigma}_{\mathbf{Z},k}^{-1}\bold{w}_k,-\tilde{V}\leq 0\right),\label{eq:derivformula}
\end{align}
where $\bold{Y}\sim\mathcal{N}(\bold{0}, \bold{I}_r)$ has the standard multivariate normal distribution in dimension $r$, and is independent of the univariate normal random variable $\tilde{V}$ with mean $b_4$ and variance $b_3^{-1}$, and where $\mathbf{Q}_r=\mathbf{\Sigma}_{\mathbf{Z}\mid k}^{1/2}\bold{Y}+\tilde{V}(\bold{1}_r- \mathbf{\Sigma}_{\mathbf{Z},rk}\mathbf{\Sigma}_{\mathbf{Z},k}^{-1}\bold{1}_k)$. To compute \eqref{eq:derivformula}, notice that the $(r+1)$-dimensional vector $\bold{Q}_{r,0} =(\bold{Q}_r^T,-\tilde{V})^T$ is jointly multivariate normal, and that $\bold{Q}_r$ has mean $b_4(\bold{1}_r- \mathbf{\Sigma}_{\mathbf{Z},rk}\mathbf{\Sigma}_{\mathbf{Z},k}^{-1}\bold{1}_k)$ and covariance $\mathbf{\Sigma}_{\mathbf{Z}\mid k} + b_3^{-1}(\bold{1}_r- \mathbf{\Sigma}_{\mathbf{Z},rk}\mathbf{\Sigma}_{\mathbf{Z},k}^{-1}\bold{1}_k)(\bold{1}_r- \mathbf{\Sigma}_{\mathbf{Z},rk}\mathbf{\Sigma}_{\mathbf{Z},k}^{-1}\bold{1}_k)^T$. Hence, 
\begin{equation*}
\pr\left(\mathbf{Q}_r\leq \bold{w}_r - \mathbf{\Sigma}_{\mathbf{Z},rk}\mathbf{\Sigma}_{\mathbf{Z},k}^{-1}\bold{w}_k,-\tilde{V}\leq 0\right)= \pr\left(\bold{Q}_{r,0}\leq \bold{q}_{r,0}\right),
\end{equation*}
where $\bold{Q}_{r,0} = (\bold{Q}_r^T, -\tilde{V})^T\sim\mathcal{N}_{r+1}\left(\mathbf{\mu}_{r,0},\mathbf\Omega_{r,0}\right)$ and
\begin{align*}
\bold{q}_{r,0} &= \begin{pmatrix}\bold{w}_r - \mathbf{\Sigma}_{\mathbf{Z},rk}\mathbf{\Sigma}_{\mathbf{Z},k}^{-1}\bold{w}_k\\ 0\end{pmatrix},  \quad \mathbf{\mu}_{r,0} = \begin{pmatrix} b_4(\bold{1}_r- \mathbf{\Sigma}_{\mathbf{Z},rk}\mathbf{\Sigma}_{\mathbf{Z},k}^{-1}\bold{1}_k)\\ -b_4 \end{pmatrix},\\
\mathbf\Omega_{r,0} &=\begin{pmatrix} \mathbf{\Sigma}_{\mathbf{Z}\mid k} + b_3^{-1}(\bold{1}_r- \mathbf{\Sigma}_{\mathbf{Z},rk}\mathbf{\Sigma}_{\mathbf{Z},k}^{-1}\bold{1}_k)(\bold{1}_r- \mathbf{\Sigma}_{\mathbf{Z},rk}\mathbf{\Sigma}_{\mathbf{Z},k}^{-1}\bold{1}_k)^T&b_3^{-1}(\mathbf{\Sigma}_{\mathbf{Z},rk}\mathbf{\Sigma}_{\mathbf{Z},k}^{-1}\bold{1}_k-\bold{1}_r)\\b_3^{-1}(\mathbf{\Sigma}_{\mathbf{Z},rk}\mathbf{\Sigma}_{\mathbf{Z},k}^{-1}\bold{1}_k-\bold{1}_r)^T&b_3^{-1}\end{pmatrix}.
\end{align*}
Finally, {we} obtain
\begin{align*}
\partial_{J_i} F_D^{\mathbf{W}}(\bold{w}) &= \lambda C\; \Phi_{r+1}\left(\bold{q}_{r,0} - \mathbf{\mu}_{r,0};\mathbf\Omega_{r,0}\right).
\end{align*}

\section{Tail properties of the non-stationary exponential factor copula model}\label{tailprop}

In this appendix, we provide detailed information on the sub-asymptotic joint tail behavior of the stationary exponential factor model~\eqref{LocalModel}.

\begin{lemma}\label{lemma1}
In Model~\eqref{GlobalModel}, the marginal distribution at site $\mathbf{s}\in\mathcal{S}$, $F_{1;\mathbf{s}}^{\mathbf{W}}(w;\lambda_{\mathbf{s}})$, satisfies
\begin{align*}
F_{1;\mathbf{s}}^{\mathbf{W}}(w;\lambda_{\mathbf{s}}) &= 1-\exp(\lambda_{\mathbf{s}}^2/2-\lambda_{\mathbf{s}} w)+\lambda_{\mathbf{s}}\phi(w)\{w(w-\lambda_{\mathbf{s}})\}^{-1} + O\{\phi(w)w^{-4}\},\qquad w\to\infty.
\end{align*}
\end{lemma}

\begin{proof}
From \eqref{eq:marginalcdf}, {we} have $F_{1;\mathbf{s}}^{\mathbf{W}}(w;\lambda_{\mathbf{s}})=\Phi(w)-\Phi(w-\lambda_{\mathbf{s}})\exp(\lambda_{\mathbf{s}}^2/2-\lambda_{\mathbf{s}} w)$.
Furthermore, thanks to a well-known expansion of the Gaussian tail, {we have} 
\begin{equation}
\label{MillsRatio}
1-\Phi(w)=\phi(w) \{w^{-1}-w^{-3}+O(w^{-5})\},\qquad w\to\infty.
\end{equation}
Plugging \eqref{MillsRatio} into the expression for $F_{1;\mathbf{s}}^{\mathbf{W}}(w;\lambda_{\mathbf{s}})$ yields, as $w\to\infty$,
\begin{eqnarray*}
F_{1;\mathbf{s}}^{\mathbf{W}}(w;\lambda_{\mathbf{s}})&=&1-\phi(w) \{w^{-1}-w^{-3}+O(w^{-5})\}\\
&&-\left[1-\phi(w-\lambda_{\mathbf{s}}) \{(w-\lambda_{\mathbf{s}})^{-1}-(w-\lambda_{\mathbf{s}})^{-3}+O(w^{-5})\}\right]\exp(\lambda_{\mathbf{s}}^2/2-\lambda_{\mathbf{s}} w)\\
&=&1-\exp(\lambda_{\mathbf{s}}^2/2-\lambda_{\mathbf{s}} w)+\phi(w)\left[(w-\lambda_{\mathbf{s}})^{-1}-(w-\lambda_{\mathbf{s}})^{-3}-w^{-1}+w^{-3}+O(w^{-5})\right]\\
&=&1-\exp(\lambda_{\mathbf{s}}^2/2-\lambda_{\mathbf{s}} w)+\lambda_{\mathbf{s}}\phi(w)\{w(w-\lambda_{\mathbf{s}})\}^{-1} + O\{\phi(w)w^{-4}\}.
\end{eqnarray*}
\end{proof}

\begin{lemma}\label{lemma3}
In Model~\eqref{GlobalModel}, the marginal quantile function at location $\mathbf{s}\in\mathcal{S}$, $q^{\mathbf{W}}_{\mathbf{s}}(t;\lambda_{\mathbf{s}}) = {F_{1;\mathbf{s}}^{\mathbf{W}}}^{-1}(1-t^{-1};\lambda_{\mathbf{s}})$, admits the expansion
\begin{align*}
q^{\mathbf{W}}_{\mathbf{s}}(t;\lambda_{\mathbf{s}}) &= \lambda_{\mathbf{s}}^{-1}\log t + \lambda_{\mathbf{s}}/2 -t{\phi(\lambda_{\mathbf{s}}^{-1}\log t + \lambda_{\mathbf{s}}/2)\over \lambda_{\mathbf{s}}^{-2}\log^2 t-\lambda_{\mathbf{s}}^2/4}\{1+o(1)\},\quad t\to\infty.
\end{align*}
\end{lemma}

\begin{proof}
From notational convenience, we shall write $q(t)\equiv q^{\mathbf{W}}_{\mathbf{s}}(t;\lambda_{\mathbf{s}})$. Because $1-F_{1;\mathbf{s}}^{\mathbf{W}}\{q(t)\}=t^{-1}$ and $q(t)\to\infty$ as $t\to\infty$, {we have} from Lemma~\ref{lemma1} that
\begin{equation}
\label{eq:formula}
\exp\{\lambda_{\mathbf{s}}^2/2-\lambda_{\mathbf{s}} q(t)\}-\lambda_{\mathbf{s}}\phi\{q(t)\}\left[q(t)\{q(t)-\lambda_{\mathbf{s}}\}\right]^{-1} + O[\phi\{q(t)\}q(t)^{-4}]=t^{-1}.
\end{equation}
Noting that the leftmost term in \eqref{eq:formula} is dominant, this yields
\begin{equation}
\label{eq:1}
q(t) = \lambda_{\mathbf{s}}^{-1}\log t + \lambda_{\mathbf{s}}/2 + r(t),
\end{equation}
where $r(t)\to0$, as $t\to\infty$. Using \eqref{eq:1} back into \eqref{eq:formula} gives 
\begin{equation*}
t^{-1}\exp\{-\lambda_{\mathbf{s}}r(t)\} - \lambda_{\mathbf{s}}{\phi(\lambda_{\mathbf{s}}^{-1}\log t + \lambda_{\mathbf{s}}/2)\over \lambda_{\mathbf{s}}^{-2}\log^2 t-\lambda_{\mathbf{s}}^2/4}\{1+o(1)\}=t^{-1}.
\end{equation*}
Thus, because $\exp(-x)=1-x\{1+o(1)\}$, as $x\to0$, {we} obtain
\begin{equation*}
r(t) = -t{\phi(\lambda_{\mathbf{s}}^{-1}\log t + \lambda_{\mathbf{s}}/2)\over \lambda_{\mathbf{s}}^{-2}\log^2 t-\lambda_{\mathbf{s}}^2/4}\{1+o(1)\},
\end{equation*}
which, combined with \eqref{eq:1}, concludes the proof.
\end{proof}

\begin{proof}[Proof of Proposition~\ref{prop2}]
Consider the stationary exponential factor model~\eqref{LocalModel}, and let $z(u)={F_1^{\mathbf{W}}}^{-1}(u;\lambda)$ denote univariate quantiles. Thanks to \eqref{eq:multivariatecdf}, {we} have
\begin{align}
C_2^{\mathbf{W}}(u,u)&=F_2^{\mathbf{W}}\left\{z(u),z(u)\right\}\nonumber\\
&=\Phi_2\{z(u),z(u);\mathbf{\Sigma}_{\mathbf{Z}}\} - 2\exp\left\{\lambda^2/2 - \lambda z(u)\right\}\Phi_2\left[\lambda\{1-\rho(h)\},z(u)-\lambda;\mathbf\Omega\right],\nonumber\\
&=\Phi_2\{z(u),z(u);\mathbf{\Sigma}_{\mathbf{Z}}\} - 2\exp\left\{\lambda^2/2 - \lambda z(u)\right\}\Phi_2\left[\lambda\sqrt{\{1-\rho(h)\}/2},z(u)-\lambda;\mathbf{\tilde\Omega}\right],\label{eq:bivariatecdf}
\end{align}
where $\rho(h)$ is the underlying correlation function, and the covariance matrices $\mathbf{\Omega}$ and $\mathbf{\tilde\Omega}$ are
$$\mathbf{\Omega}=\begin{pmatrix}
2\{1-\rho(h)\} & -\{1-\rho(h)\}\\
-\{1-\rho(h)\} & 1
\end{pmatrix},
\quad\mathbf{\tilde\Omega}=\begin{pmatrix}
1 & -\sqrt{\{1-\rho(h)\}/2}\\
-\sqrt{\{1-\rho(h)\}/2} & 1
\end{pmatrix}.$$
Now, thanks to Lemma~\ref{lemma3}, $z(u)=q^{\mathbf{W}}_{\mathbf{s}}\{(1-u)^{-1};\lambda\}=-\lambda^{-1}\log(1-u)+\lambda/2+s(u)$, with 
\begin{equation}
\label{eq:su}
s(u)=-{\phi\{z(u)\}\over (1-u)z(u)\{z(u)-\lambda\}}\{1+o(1)\},\quad u\to1,
\end{equation}
such that $s(u)\to0$, $z(u)=-\lambda^{-1}\log(1-u)+\lambda/2\;\{1+o(1)\}$, and $z(u)\to\infty$, as $u\to1$.

Consequently, from the definition \eqref{eq:chiu} and \eqref{eq:bivariatecdf}--\eqref{eq:su}, {we} can write
\begin{equation}
\label{eq:chiuformula}
\chi_h(u)={1-2u+C_2^{\mathbf{W}}(u,u)\over 1-u}=2-{1-C_2^{\mathbf{W}}(u,u)\over 1-u}=2-f(u)-g(u)h(u),
\end{equation}
where, using the bivariate Gaussian survivor function $\overline\Phi_2$, {we have}
\begin{align}
f(u)&={1-\Phi_2\{z(u),z(u);\mathbf{\Sigma}_{\mathbf{Z}}\}\over 1-u}={2[1-\Phi\{z(u)\}]-\overline\Phi_2\{z(u),z(u);\mathbf{\Sigma}_{\mathbf{Z}}\}\over 1-u}\label{eq:fu}\\
&={2[1-\Phi\{z(u)\}]\{1+o(1)\}\over 1-u}={2\phi\{z(u)\}\over (1-u)z(u)}\{1+o(1)\}\to0,\qquad u\to1;\nonumber\\
g(u)&={\exp\left\{\lambda^2/2 - \lambda z(u)\right\}\over1-u}=\exp\{-\lambda s(u)\}\to1,\quad u\to1;\label{eq:gu}\\
h(u)&=2\Phi_2\left[\lambda\sqrt{\{1-\rho(h)\}/2},z(u)-\lambda;\mathbf{\tilde\Omega}\right]\to 2\Phi\left[\lambda\sqrt{\{1-\rho(h)\}/2}\right],\quad u\to 1.\label{eq:hu}
\end{align}
The $3^{\rm rd}$ equality in \eqref{eq:fu} is true as bivariate Gaussian random vectors $(Z_1,Z_2)^T$ with correlation $\rho<1$ are tail-independent, i.e., $\pr(Z_1>z\mid Z_2>z)\to0$, as $z\to\infty$ \citep{Ledford.Tawn:1996}; the expansion in the $4^{\rm th}$ equality in \eqref{eq:fu} is due to \eqref{MillsRatio}; finally, $f(u)\to0$, as $\phi\{z(u)\}/(1-u)\sim (2\pi)^{-1/2}\exp\{-\log^2(1-u)/2-\log(1-u)\}\to0$, as $u\to1$.

Equations \eqref{eq:chiuformula}--\eqref{eq:hu} imply that $\chi_h=\lim_{u\to1}\chi_h(u)=2-2\Phi[\lambda\sqrt{\{1-\rho(h)\}/2}]$. Hence, by writing $A=\lambda\sqrt{\{1-\rho(h)\}/2}$, {we} obtain as $u\to1$
\begin{align}
\chi_h(u)-\chi_h&=2\Phi(A)-f(u) -g(u)h(u)\nonumber\\
&=-f(u) + 2\Phi(A)\left[1-\exp\{-\lambda s(u)\}{\Phi_2\left\{A,z(u)-\lambda;\mathbf{\tilde\Omega}\right\}\over \Phi(A)}\right]\nonumber\\
&=-f(u) + 2\Phi(A)\left[1-\{1-\lambda s(u)\}\{1-k(u)\}\{1+o(1)\}\right]\nonumber\\
&=-f(u) + 2\Phi(A)\left\{\lambda s(u) + k(u)\right\}\{1+o(1)\}\label{eq:chiu.min.chi},
\end{align}
where $k(u)=\left[\Phi\left(A\right)-\Phi_2\left\{A,z(u)-\lambda;\mathbf{\tilde\Omega}\right\}\right]/\Phi\left(A\right)\to 0$ {as $u\to 1$.} Notice that by the l'Hospital's rule, it can be verified that for $\rho(h)>0$, $k(u)$ satisfies {the expression}
\begin{align}
\label{eq:ku2}
k(u)&={1-\Phi\{z(u)-\lambda\}\over \Phi(A)}\{1+o(1)\}={\phi\{z(u)-\lambda\}\over \Phi(A)\{z(u)-\lambda\}}\{1+o(1)\}\to0,\quad u\to1.
\end{align}
Comparing the rates of convergence of $s(u)$, $f(u)$ and $k(u)$ in \eqref{eq:su}, \eqref{eq:fu} and \eqref{eq:ku2}, respectively, we deduce that $f(u)$ is dominant as $u\to1$, and therefore from \eqref{eq:chiu.min.chi}, {we have}
\begin{equation*}
\chi_h(u)-\chi_h=-f(u)\{1+o(1)\}=-{2\phi\{z(u)\}\over (1-u)z(u)}\{1+o(1)\}=-{2\phi\{-\lambda^{-1}\log(1-u)+\lambda/2\}\over(1-u)\{-\lambda^{-1}\log(1-u)+\lambda/2\}}\{1+o(1)\},
\end{equation*}
which concludes the proof.
\end{proof}

\baselineskip 16pt

\bibliographystyle{CUP}
\bibliography{References.bib}

\begin{thebibliography}{64}

\bibitem[Anderes and Stein(2011)]{anderes2011local}
Anderes, E.~B. and Stein, M.~L. (2011) Local likelihood estimation for
  nonstationary random fields.
\newblock \emph{Journal of Multivariate Analysis} \textbf{102}(3), 506--520.

\bibitem[Asadi \emph{et~al.}(2015)Asadi, Davison and Engelke]{Asadi.etal:2015}
Asadi, P., Davison, A.~C. and Engelke, S. (2015) {Extremes on river networks}.
\newblock \emph{{Annals of Applied Statistics}} \textbf{9}(4), 2023--2050.

\bibitem[Blanchet and Davison(2011)]{Blanchet.Davison:2011}
Blanchet, J. and Davison, A.~C. (2011) {Spatial modelling of extreme snow
  depth}.
\newblock \emph{{Annals of Applied Statistics}} \textbf{5}(3), 1699--1725.

\bibitem[Bopp \emph{et~al.}(2019)Bopp, Shaby and Huser]{bopp2018hierarchical}
Bopp, G.~P., Shaby, B.~A. and Huser, R. (2019) A hierarchical max-infinitely
  divisible process for extreme areal precipitation over watersheds.
\newblock \emph{arXiv preprint 1805.06084} .

\bibitem[Castro-Camilo and de~Carvalho(2017)]{castro2016spectral}
Castro-Camilo, D. and de~Carvalho, M. (2017) Spectral density regression for
  bivariate extremes.
\newblock \emph{Stochastic environmental research and risk assessment}
  \textbf{31}(7), 1603--1613.

\bibitem[Castro-Camilo \emph{et~al.}(2018)Castro-Camilo, de~Carvalho and
  Wadsworth]{castro2018time}
Castro-Camilo, D., de~Carvalho, M. and Wadsworth, J. (2018) Time-varying
  extreme value dependence with application to leading european stock markets.
\newblock \emph{The Annals of Applied Statistics} \textbf{12}(1), 283--309.

\bibitem[Castruccio \emph{et~al.}(2016)Castruccio, Huser and
  Genton]{castruccio2016high}
Castruccio, S., Huser, R. and Genton, M.~G. (2016) High-order composite
  likelihood inference for max-stable distributions and processes.
\newblock \emph{Journal of Computational and Graphical Statistics}
  \textbf{25}(4), 1212--1229.

\bibitem[Cooley \emph{et~al.}(2007)Cooley, Naveau and Nychka]{Cooley.etal:2007}
Cooley, D.~S., Naveau, P. and Nychka, D. (2007) {Bayesian Spatial Modeling of
  Extreme Precipitation Return Levels}.
\newblock \emph{{Journal of American Statistical Association}}
  \textbf{102}(479), 824--840.

\bibitem[Davison and Gholamrezaee(2012)]{Davison.Gholamrezaee:2012}
Davison, A.~C. and Gholamrezaee, M.~M. (2012) {Geostatistics of extremes}.
\newblock \emph{{Proceedings of the Royal Society A: Mathematical, Physical \&
  Engineering Sciences}} \textbf{468}(2138), 581--608.

\bibitem[Davison \emph{et~al.}(2019)Davison, Huser and
  Thibaud]{davison2019spatial}
Davison, A.~C., Huser, R. and Thibaud, E. (2019) {Spatial extremes}.
\newblock In \emph{Handbook of Environmental and Ecological Statistics},
  chapter~31, pp. 711--744. CRC Press.

\bibitem[Davison \emph{et~al.}(2012)Davison, Padoan and
  Ribatet]{davison2012statistical}
Davison, A.~C., Padoan, S. and Ribatet, M. (2012) {Statistical Modelling of
  Spatial Extremes (with Discussion)}.
\newblock \emph{{Statistical Science}} \textbf{27}(2), 161--186.

\bibitem[Davison and Ramesh(2000)]{Davison.Ramesh:2000}
Davison, A.~C. and Ramesh, N. (2000) {Local likelihood smoothing of sample
  extremes}.
\newblock \emph{{Journal of the Royal Statistical Society: Series B
  (Statistical Methodology)}} \textbf{62}(1), 191--208.

\bibitem[{de Fondeville} and Davison(2018)]{deFondevilleDavison16}
{de Fondeville}, R. and Davison, A.~C. (2018) High-dimensional
  peaks-over-threshold inference.
\newblock \emph{Biometrika} \textbf{105}(3), 575--592.

\bibitem[Engelke \emph{et~al.}(2015)Engelke, Malinowski, Kabluchko and
  Schlather]{Engelkeetal15}
Engelke, S., Malinowski, A., Kabluchko, Z. and Schlather, M. (2015) Estimation
  of h{\"u}sler--reiss distributions and brown--resnick processes.
\newblock \emph{Journal of the Royal Statistical Society: Series B (Statistical
  Methodology)} \textbf{77}(1), 239--265.

\bibitem[Ferreira \emph{et~al.}(2014)Ferreira, De~Haan
  \emph{et~al.}]{Ferreira.deHaan:2014}
Ferreira, A., De~Haan, L. \emph{et~al.} (2014) The generalized pareto process;
  with a view towards application and simulation.
\newblock \emph{Bernoulli} \textbf{20}(4), 1717--1737.

\bibitem[Fill and Stedinger(1995)]{fill1995homogeneity}
Fill, H.~D. and Stedinger, J.~R. (1995) Homogeneity tests based upon gumbel
  distribution and a critical appraisal of dalrymple's test.
\newblock \emph{Journal of Hydrology} \textbf{166}(1--2), 81--105.

\bibitem[Fischer and Knutti(2016)]{Fischer.Knutti:2016}
Fischer, E.~M. and Knutti, R. (2016) {Observed heavy precipitation increase
  confirms theory and early models}.
\newblock \emph{{Nature Climate Change}} \textbf{6}(11), 986--991.

\bibitem[Fuentes(2001)]{Fuentes:2001}
Fuentes, M. (2001) {A High Frequency Kriging Approach for Non-Stationary
  Environmental Processes}.
\newblock \emph{{Environmetrics}} \textbf{12}(5), 469--483.

\bibitem[Genest \emph{et~al.}(1995)Genest, Ghoudi and Rivest]{Genest.etal:1995}
Genest, C., Ghoudi, K. and Rivest, L.-P. (1995) {A semiparametric estimation
  procedure of dependence parameters in multivariate families of
  distributions}.
\newblock \emph{{Biometrika}} \textbf{82}(3), 543--552.

\bibitem[Genton \emph{et~al.}(2014)Genton, Johnson, Potter, Stenchikov and
  Sun]{genton2014surface}
Genton, M.~G., Johnson, C., Potter, K., Stenchikov, G. and Sun, Y. (2014)
  Surface boxplots.
\newblock \emph{Stat} \textbf{3}(1), 1--11.

\bibitem[Gneiting \emph{et~al.}(2006)Gneiting, Genton and
  Guttorp]{gneiting2006geostatistical}
Gneiting, T., Genton, M.~G. and Guttorp, P. (2006) Geostatistical space-time
  models, stationarity, separability, and full symmetry.
\newblock \emph{Monographs On Statistics and Applied Probability} \textbf{107},
  151.

\bibitem[Higdon(1998)]{higdon1998process}
Higdon, D. (1998) A process-convolution approach to modelling temperatures in
  the north atlantic ocean.
\newblock \emph{Environmental and Ecological Statistics} \textbf{5}(2),
  173--190.

\bibitem[Hoerling \emph{et~al.}(2016)Hoerling, Eischeid, Perlwitz, Quan, Wolter
  and Cheng]{Hoerling.etal:2016}
Hoerling, M., Eischeid, J., Perlwitz, J., Quan, X.-W., Wolter, K. and Cheng, L.
  (2016) {Characterizing recent trends in U.S. heavy precipitation}.
\newblock \emph{Journal of Climate} \textbf{29}(1), 2313--2332.

\bibitem[Hosking and Wallis(1993)]{hosking1993some}
Hosking, J. and Wallis, J. (1993) Some statistics useful in regional frequency
  analysis.
\newblock \emph{Water resources research} \textbf{29}(2), 271--281.

\bibitem[Hosking and Wallis(2005)]{hosking2005regional}
Hosking, J. R.~M. and Wallis, J.~R. (2005) \emph{Regional frequency analysis:
  an approach based on L-moments}.
\newblock Cambridge University Press.

\bibitem[Huser and Davison(2013)]{Huser.Davison:2013a}
Huser, R. and Davison, A.~C. (2013) {Composite likelihood estimation for the
  Brown--Resnick process}.
\newblock \emph{{Biometrika}} \textbf{100}(2), 511--518.

\bibitem[Huser and Davison(2014)]{huser2014space}
Huser, R. and Davison, A.~C. (2014) Space--time modelling of extreme events.
\newblock \emph{Journal of the Royal Statistical Society: Series B (Statistical
  Methodology)} \textbf{76}(2), 439--461.

\bibitem[Huser \emph{et~al.}(2016)Huser, Davison and Genton]{Huseretal16}
Huser, R., Davison, A.~C. and Genton, M.~G. (2016) Likelihood estimators for
  multivariate extremes.
\newblock \emph{Extremes} \textbf{19}(1), 79--103.

\bibitem[Huser \emph{et~al.}(2019)Huser, Dombry, Ribatet and
  Genton]{huser2019full}
Huser, R., Dombry, C., Ribatet, M. and Genton, M.~G. (2019) Full likelihood
  inference for max-stable data.
\newblock \emph{Stat} \textbf{8}(1), e218.

\bibitem[Huser and Genton(2016)]{Huser2016}
Huser, R. and Genton, M.~G. (2016) Non-stationary dependence structures for
  spatial extremes.
\newblock \emph{Journal of agricultural, biological, and environmental
  statistics} \textbf{21}(3), 470--491.

\bibitem[Huser \emph{et~al.}(2017)Huser, Opitz and Thibaud]{Huser.etal:2016}
Huser, R., Opitz, T. and Thibaud, E. (2017) Bridging asymptotic independence
  and dependence in spatial extremes using gaussian scale mixtures.
\newblock \emph{Spatial Statistics} \textbf{21}, 166--186.

\bibitem[Huser \emph{et~al.}(2018)Huser, Opitz and
  Thibaud]{huser2018penultimate}
Huser, R., Opitz, T. and Thibaud, E. (2018) Penultimate modeling of spatial
  extremes: statistical inference for max-infinitely divisible processes.
\newblock \emph{arXiv preprint 1801.02946} .

\bibitem[Huser and Wadsworth(2019)]{Huser.Wadsworth:2017}
Huser, R. and Wadsworth, J.~L. (2019) Modeling spatial processes with unknown
  extremal dependence class.
\newblock \emph{Journal of the American Statistical Association. \normalfont{To
  appear}} .

\bibitem[H{\"u}sler and Reiss(1989)]{husler1989maxima}
H{\"u}sler, J. and Reiss, R.-D. (1989) Maxima of normal random vectors: between
  independence and complete dependence.
\newblock \emph{Statistics \& Probability Letters} \textbf{7}(4), 283--286.

\bibitem[Joe(2014)]{joe2014dependence}
Joe, H. (2014) \emph{Dependence modeling with copulas}.
\newblock CRC Press.

\bibitem[Katz \emph{et~al.}(2002)Katz, Parlange and Naveau]{katz2002statistics}
Katz, R.~W., Parlange, M.~B. and Naveau, P. (2002) Statistics of extremes in
  hydrology.
\newblock \emph{Advances in water resources} \textbf{25}(8-12), 1287--1304.

\bibitem[Krupskii \emph{et~al.}(2018)Krupskii, Huser and
  Genton]{krupskii2018factor}
Krupskii, P., Huser, R. and Genton, M.~G. (2018) Factor copula models for
  replicated spatial data.
\newblock \emph{Journal of the American Statistical Association}
  \textbf{113}(521), 467--479.

\bibitem[Krupskii and Joe(2015)]{Krupskii.Joe:2015}
Krupskii, P. and Joe, H. (2015) {Structured factor copula models: theory,
  inference and computation}.
\newblock \emph{{Journal of Multivariate Analysis}} \textbf{138}, 53--73.

\bibitem[Ledford and Tawn(1996)]{Ledford.Tawn:1996}
Ledford, A.~W. and Tawn, J.~A. (1996) {Statistics for near independence in
  multivariate extreme values}.
\newblock \emph{{Biometrika}} \textbf{83}(1), 169--187.

\bibitem[Lu and Stedinger(1992)]{lu1992sampling}
Lu, L.-H. and Stedinger, J.~R. (1992) Sampling variance of normalized gev/pwm
  quantile estimators and a regional homogeneity test.
\newblock \emph{Journal of Hydrology} \textbf{138}(1-2), 223--245.

\bibitem[Morris \emph{et~al.}(2017)Morris, Reich, Thibaud and
  Cooley]{morris2017space}
Morris, S.~A., Reich, B.~J., Thibaud, E. and Cooley, D. (2017) A space-time
  skew-t model for threshold exceedances.
\newblock \emph{Biometrics} \textbf{73}(3), 749--758.

\bibitem[Nychka \emph{et~al.}(2002)Nychka, Wikle and Royle]{Nychka.etal:2002}
Nychka, D., Wikle, C.~K. and Royle, J.~A. (2002) {Multiresolution Models for
  Nonstationary Spatial Covariance Functions}.
\newblock \emph{{Statistical Modelling}} \textbf{2}(4), 315--331.

\bibitem[Opitz(2016)]{Opitz:2016}
Opitz, T. (2016) {Modeling asymptotically independent spatial extremes based on
  Laplace random fields}.
\newblock \emph{{Spatial Statistics}} \textbf{16}(1), 1--18.

\bibitem[Opitz \emph{et~al.}(2018)Opitz, Huser, Bakka and Rue]{opitz2018inla}
Opitz, T., Huser, R., Bakka, H. and Rue, H. (2018) {INLA} goes extreme:
  Bayesian tail regression for the estimation of high spatio-temporal
  quantiles.
\newblock \emph{Extremes} \textbf{21}(3), 441--462.

\bibitem[Paciorek and Schervish(2006)]{paciorek2006spatial}
Paciorek, C.~J. and Schervish, M.~J. (2006) Spatial modelling using a new class
  of nonstationary covariance functions.
\newblock \emph{Environmetrics} \textbf{17}(5), 483--506.

\bibitem[Padoan \emph{et~al.}(2010)Padoan, Ribatet and
  Sisson]{padoan2010likelihood}
Padoan, S.~A., Ribatet, M. and Sisson, S.~A. (2010) Likelihood-based inference
  for max-stable processes.
\newblock \emph{Journal of the American Statistical Association}
  \textbf{105}(489), 263--277.

\bibitem[Risser and Calder(2017)]{risser2015local}
Risser, M.~D. and Calder, C.~A. (2017) Local likelihood estimation for
  covariance functions with spatially-varying parameters: the convo{SPAT}
  package for {R}.
\newblock \emph{Journal of Statistical Software} \textbf{81}(14), 1--32.

\bibitem[R{\o}islien and Omre(2006)]{roislien2006t}
R{\o}islien, J. and Omre, H. (2006) T-distributed random fields: A parametric
  model for heavy-tailedwell-log data1.
\newblock \emph{Mathematical Geology} \textbf{38}(7), 821--849.

\bibitem[Rootz{\'e}n \emph{et~al.}(2018)Rootz{\'e}n, Segers and
  Wadsworth]{rootzen2018multivariate}
Rootz{\'e}n, H., Segers, J. and Wadsworth, J.~L. (2018) Multivariate
  generalized pareto distributions: Parametrizations, representations, and
  properties.
\newblock \emph{Journal of Multivariate Analysis} \textbf{165}, 117--131.

\bibitem[Scholz and Stephens(1987)]{scholz1987k}
Scholz, F.~W. and Stephens, M.~A. (1987) K-sample anderson--darling tests.
\newblock \emph{Journal of the American Statistical Association}
  \textbf{82}(399), 918--924.

\bibitem[Sklar(1959)]{sklar1959fonctions}
Sklar, M. (1959) Fonctions de repartition an dimensions et leurs marges.
\newblock \emph{Publ. inst. statist. univ. Paris} \textbf{8}, 229--231.

\bibitem[Stein(1999)]{Stein:1999}
Stein, M.~L. (1999) \emph{{Interpolation of Spatial Data: Some Theory for
  Kriging}}. First edition.
\newblock New York: {Springer}.

\bibitem[Stein(2005)]{Stein:2005b}
Stein, M.~L. (2005) {Nonstationary Spatial Covariance Functions}.
\newblock Unpublished.

\bibitem[Sun and Genton(2011)]{sun2012functional}
Sun, Y. and Genton, M.~G. (2011) Functional boxplots.
\newblock \emph{Journal of Computational and Graphical Statistics}
  \textbf{20}(2), 316--334.

\bibitem[Thibaud \emph{et~al.}(2013)Thibaud, Mutzner and
  Davison]{Thibaud.etal:2013}
Thibaud, E., Mutzner, R. and Davison, A.~C. (2013) {Threshold modeling of
  extreme spatial rainfall}.
\newblock \emph{{Water Resources Research}} \textbf{49}(8), 4633--4644.

\bibitem[Thibaud and Opitz(2015)]{ThibaudOpitz15}
Thibaud, E. and Opitz, T. (2015) Efficient inference and simulation for
  elliptical {Pareto} processes.
\newblock \emph{Biometrika} \textbf{102}(4), 855--870.

\bibitem[Vettori \emph{et~al.}(2019)Vettori, Huser and
  Genton]{Vettori.etal:2019}
Vettori, S., Huser, R. and Genton, M.~G. (2019) {Bayesian modeling of air
  pollution extremes using nested multivariate max-stable processes}.
\newblock \emph{{Biometrics}} To appear.

\bibitem[Viglione \emph{et~al.}(2007)Viglione, Laio and
  Claps]{viglione2007comparison}
Viglione, A., Laio, F. and Claps, P. (2007) A comparison of homogeneity tests
  for regional frequency analysis.
\newblock \emph{Water Resources Research} \textbf{43}(3).

\bibitem[Wadsworth and Tawn(2012)]{Wadsworth.Tawn:2012b}
Wadsworth, J.~L. and Tawn, J.~A. (2012) {Dependence modelling for spatial
  extremes}.
\newblock \emph{{Biometrika}} \textbf{99}(2), 253--272.

\bibitem[Wadsworth and Tawn(2014)]{WadsworthTawn14}
Wadsworth, J.~L. and Tawn, J.~A. (2014) Efficient inference for spatial extreme
  value processes associated to log-{G}aussian random functions.
\newblock \emph{Biometrika} \textbf{101}(1), 1--15.

\bibitem[Westra \emph{et~al.}(2013)Westra, Alexander and
  Zwiers]{Westra.etal:2013}
Westra, S., Alexander, L.~V. and Zwiers, F.~W. (2013) Global increasing trends
  in annual maximum daily precipitation.
\newblock \emph{Journal of Climate} \textbf{26}(11), 3904--3918.

\bibitem[Westra and Sisson(2011)]{westra2011detection}
Westra, S. and Sisson, S.~A. (2011) Detection of non-stationarity in
  precipitation extremes using a max-stable process model.
\newblock \emph{Journal of Hydrology} \textbf{406}(1), 119--128.

\bibitem[Wilby \emph{et~al.}(1998)Wilby, Wigley, Conway, Jones, Hewitson, Main
  and Wilks]{wilby1998statistical}
Wilby, R.~L., Wigley, T., Conway, D., Jones, P., Hewitson, B., Main, J. and
  Wilks, D. (1998) Statistical downscaling of general circulation model output:
  A comparison of methods.
\newblock \emph{Water resources research} \textbf{34}(11), 2995--3008.

\bibitem[Zhang(2004)]{zhang2004inconsistent}
Zhang, H. (2004) Inconsistent estimation and asymptotically equal
  interpolations in model-based geostatistics.
\newblock \emph{Journal of the American Statistical Association}
  \textbf{99}(465), 250--261.

\end{thebibliography}

\baselineskip 10pt

\end{document}